\documentclass[11pt]{amsart}
\usepackage{amsmath,amsxtra,amssymb,amsthm}
\usepackage[final]{showkeys}
\usepackage{color}
\definecolor{dblue}{rgb}{0, 0, 0.72}
\usepackage{algorithm}
\usepackage{algpseudocode}
\usepackage{amsmath}
\usepackage{color}
\usepackage{graphicx}
\usepackage{framed}
\usepackage{comment}
\usepackage{tikz}
\usetikzlibrary{arrows.meta,decorations.pathreplacing,patterns}
\usepackage{braket} 
\usepackage{mathrsfs}
\usepackage{hyperref}

\numberwithin{equation}{section}

\makeatletter
\renewcommand{\paragraph}{\@startsection{paragraph}{4}%
  \z@\z@{-\fontdimen2\font}%
  {\normalfont\bfseries}}
\makeatother

\renewcommand{\for}{\begin{eqnarray*}}
\newcommand{\mel}{\end{eqnarray*}}
\def\fr{\begin{align*}}

\newcommand{\kl}{\pl \le \pl}

\newcommand{\lel}{\pl = \pl}

\newcommand{\supp}{\operatorname{supp}}

\newcommand{\rz}{{\mathbb R}}

\newcommand{\pl}{\hspace{.1cm}}

\newcommand{\om}{\omega}

\newcommand{\si}{\sigma}

\newcommand{\la}{\lambda}
\newcommand{\eps}{\varepsilon}

\let\dotlessi\i
\DeclareRobustCommand{\i}{\ifmmode\subset\else\dotlessi\fi}

\newcommand{\ez}{{\mathbb E}\vspace{0.1cm}}

\newtheorem{lemma}{Lemma}[section]
\newtheorem{prop}[lemma]{Proposition}
\newtheorem{theorem}[lemma]{Theorem}
\newtheorem{cor}[lemma]{Corollary}

\newtheorem{rem}[lemma]{Remark}

\newcommand{\re}{\begin{rem}\rm}
  \newcommand{\mar}{\end{rem}}

\newcommand{\qd}{\end{proof}\vspace{0.5ex}}

\newcommand{\prf}{\begin{proof}[\bf Proof:]}

\newcommand{\xspace}{\hbox{\kern-2.5pt}}

\allowdisplaybreaks

\theoremstyle{remark}
\newtheorem{assumption}{Assumption}
\theoremstyle{plain}
\definecolor{LightGray}{rgb}{0.94,0.94,0.94}
\definecolor{VeryLightBlue}{rgb}{0.9,0.9,1}
\definecolor{LightBlue}{rgb}{0.8,0.8,1}
\definecolor{DarkBlue}{rgb}{0,0,0.6}
\definecolor{LightGreen}{rgb}{0.88,1,0.88}
\definecolor{MidGreen}{rgb}{0.6,1,0.6}
\definecolor{DarkGreen}{rgb}{0,0.6,0}
\definecolor{DarkGrreen}{rgb}{0,0.8,0}

\definecolor{VeryLightYellow}{rgb}{1,1,0.9}
\definecolor{LightYellow}{rgb}{1,1,0.6}
\definecolor{MidYellow}{rgb}{1,1,0.5}
\definecolor{DarkYellow}{rgb}{0.8,1,0.3}
\definecolor{VeryLightRed}{rgb}{1,0.9,0.9}
\definecolor{LightRed}{rgb}{1,0.8,0.8}
\definecolor{DarkRed}{rgb}{0.8,0.2,0}
\definecolor{DarkRedb}{rgb}{0.6,0.2,0}
\definecolor{DarkLila}{rgb}{0.8,0,1}
\definecolor{Beige}{rgb}{0.96,0.96,0.86}
\definecolor{Gold}{rgb}{1.,0.84,0.}
\definecolor{Goldb}{rgb}{0.7,0.3,0.5}
\definecolor{MyYellow}{rgb}{1.,0.84,0.8}

\DeclareRobustCommand{\notationrev}[1]{{#1}}

\begin{document}
\title[Driven-dissipative ground state preparation]{%
  Driven-Dissipative Ground State Preparation:\\[0.5em]
  {\small Mixing Time and Randomness}%
}

\author{Paul Cazeaux}
\address[Paul Cazeaux]{Department of Mathematics, Virginia Tech,
Blacksburg, VA 24061, USA}

\author{Marius Junge}
\address[Marius Junge]{Department of Mathematics, University of Illinois Urbana-Champaign,
Urbana, IL 61801, USA}

\author{Diyi Liu}
\address[Diyi Liu]{Applied Mathematics and Computational Research Division,
Lawrence Berkeley National Laboratory, Berkeley, CA 94720, USA}

\begin{abstract} 
Dissipative protocols offer a route to ground state preparation that does not require a large initial overlap with the ground state. To move beyond time independent Lindbladians and randomized Lindbladians sampled from a fixed distribution, we study driven-dissipative protocols in which the generator varies over time. Our construction uses a family of randomized Lindbladians whose jump operators are obtained by averaging the unitary evolution of random matrices under a fixed target Hamiltonian. We show that empirical averages of these Lindbladians drive arbitrary initial density matrices to low energy states. The resulting mixing time bounds depend on the Hamiltonian only through its eigenvalue distribution and remain valid even for cold starts.

\end{abstract} 

\maketitle

\section{Introduction}

Preparing ground states of large Hamiltonian is one of the major challenges in quantum information science.  Indeed, starting with a relevant state serves as an input for solving many problems on quantum computers.  In experiments it appears that cooling a system is the universal recipe for improving a quantum task. Somehow nature knows how to perform efficient cooling, even if general algorithms performing ground state preparation for general classes of Hamiltonian remain elusive.
Successful classical calculations in quantum chemistry,
materials science, and molecular modeling also motivate the search for
quantum preparation algorithms for these systems.


There are many existing quantum algorithms for computing the ground state of a given Hamiltonian, or an approximation of its ground state. 
Coherent quantum algorithms for ground state preparation include phase
estimation~\cite{AbramsLloyd1999}, adiabatic ground state
preparation~\cite{FarhiGoldstoneGutmannSipser2000,JansenRuskaiSeiler2007,AlbashLidar2018},
and spectral
filtering~\cite{PoulinWocjan2009,GeTuraCirac2019,LinTong2020,DongLinTong2022}.
The cost of phase estimation and spectral filtering grows with the inverse
spectral gap and with the inverse overlap
between an initial state and the ground state, and this overlap typically decays
exponentially with the system size~\cite{Lee2023}.
Adiabatic preparation starts from an easily prepared ground state, and its
cost depends on the minimum spectral gap and the variation of the Hamiltonian
along the interpolation to the target~\cite{JansenRuskaiSeiler2007,AlbashLidar2018}.

Dissipative dynamics, inspired or mimicking the cooling of open quantum systems, has also been used for state preparation and quantum
computation~\cite{Kraus2008,VerstraeteWolfCirac2009}. 
Reservoir engineering~\cite{PoyatosCiracZoller1996} and its
many-body extensions~\cite{Kraus2008,Diehl2008} prepare
the ground states of frustration-free many body  Hamiltonians as unique stationary states of quasi-local Lindbladians,
and dissipation alone can perform universal quantum
computation~\cite{VerstraeteWolfCirac2009}. 
Dissipative methods require no initial overlap, since the ground
state is an attracting fixed point.
More recent
constructions filter energy differences in the Heisenberg picture, with jump
operators $a_f=\int f(t)e^{it\mathcal H}Ae^{-it\mathcal H}\,dt$ modeled on a 
coupling operator $A$, for Gibbs
states~\cite{ChenKastoryanoBrandaoGilyen2025,ChenKastoryanoGilyen2023} and for
ground states~\cite{DingChenLin2024,Li2025Dissipative,Zhan2026Rapid}.
They
require only Hamiltonian simulation~\cite{DingLiLin2024}, and a ground-state
protocol of this type has been run on a trapped-ion quantum
computer~\cite{Seki2026}.  
Related algorithms couple the system to ancillas that are repeatedly reset~\cite{PollaHerasymenkoOBrien2021,Matthies2024,MartiMansurogluHartmann2025,DingZhanPreskillLin2026} and a new method based on bath-assisted cooling is developed without reset~\cite{yu2026bathassistedcoolingresets}.
Concrete models provide further benchmarks.  For the toric code, these
include dissipative encoding~\cite{DengisKoenigPastawski2014}, locality
lower bounds on generating topological order~\cite{KoenigPastawski2014},
and polynomial-time preparation of low-temperature Gibbs
states~\cite{DingLandauLiLinZhang2026}.  


Our approach to state preparation uses time dependent Lindbladians, the generator of time evolution for a memoryless open quantum system. More specifically, we propose two possible methods, each constructing iteratively jump operators which reduce the energy levels step by step without implementing obvious spectral projections. In each step the number of independently chosen jump operators and the mixing time is controlled as a property of the spectrum distribution. These methods are independent of the eigenbasis of the Hamiltonian, but heavily rely on assuming a-priori knowledge of the spectral distribution which can be verified from a given histogram. 
We construct dissipative channels that drive states towards successively
smaller subspaces at low energies. 
Figure~\ref{fig:intro} shows the first method applied to a randomly chosen Hamiltonian drawn from the Gaussian Unitary Ensemble (GUE), for which the 
eigenvalue distribution is asymptotically known. 
\begin{figure}[t]
\centering
\includegraphics[width=.9\linewidth]{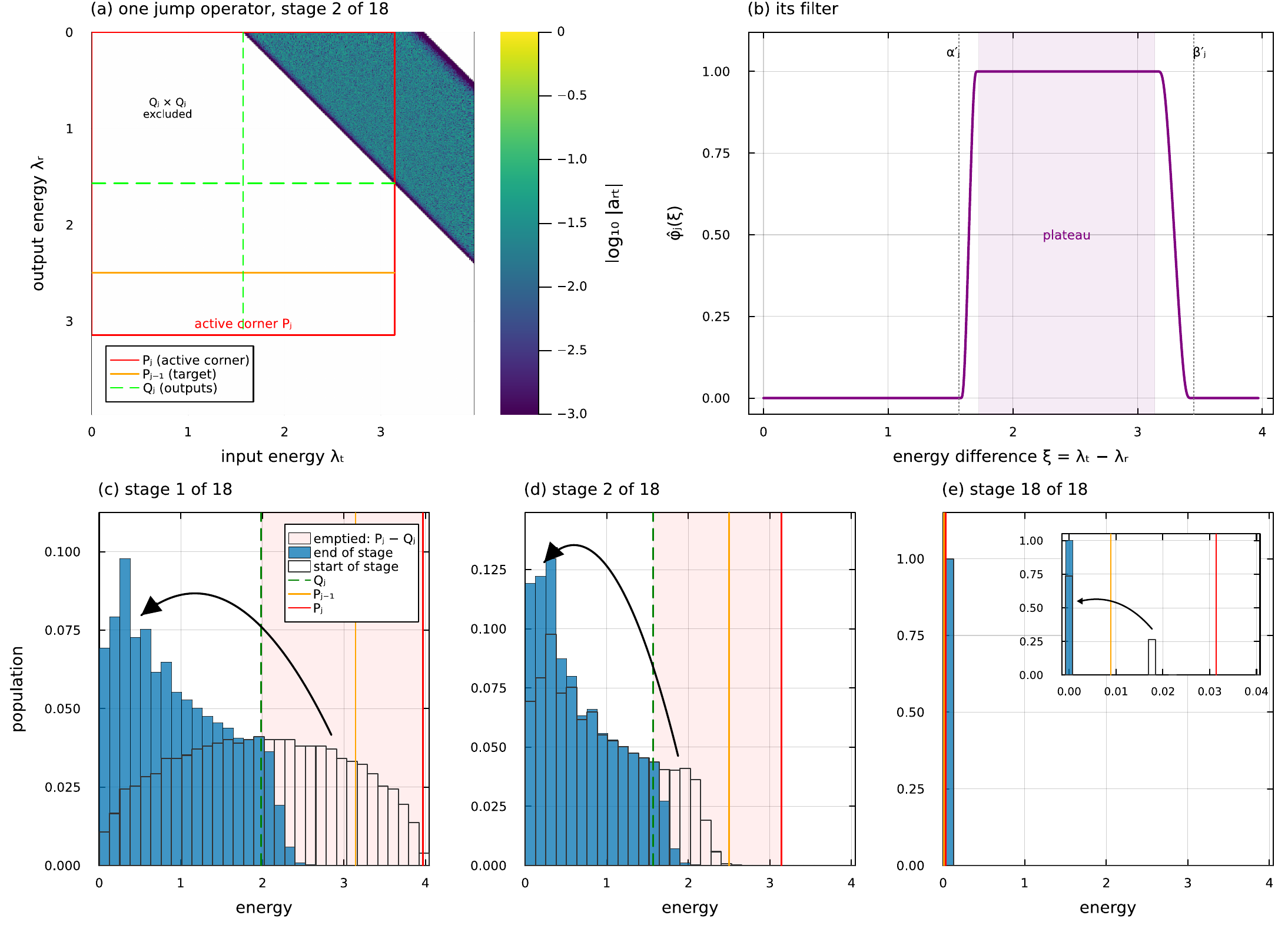}
\caption{The first method on a GUE spectrum, $d=1024$, $18$ stages
(Section~\ref{subsec:ex-semicircle}).  (a) A random jump operator of the second
stage in the eigenbasis: on the active corner $P_j$ it maps the inputs
outside $Q_j$ into $Q_j$.  (b) Its filter.  (c), (d) and (e) Populations of
the eigenvalues at the start and end of the first, second and last stages, from the uniform superposition.}
\label{fig:intro}
\end{figure}

\begin{theorem}\label{th:informal} (informal) Let $H$ be a Hamiltonian  on $n$ qubits and spectral projection $P_{0}$ on the lowest energy level.   Under regularity assumptions on the spectrum of $H$ there are $m$ Lindblad operators $L_1,...,L_m$ of the form
  \[ L_j \lel \frac{1}{M_j} \sum_{k=1}^{M_j} L_{a_j(\om_k)}, \qquad L_{a}(\rho) \lel 2a\rho a^* - a^* a\rho - \rho a^* a, \] 
and times $t_1,...,t_m$ such that for every state encoded by a density matrix $\rho$,
 \[ \|e^{t_1L_1}\cdots e^{t_mL_m}(\rho)-\si\|_\diamond <\eps \text{ and } \si \lel P_{0}\si P_{0} \pl. \]
The times $t_j$, at a normalization under which the jump operators have norm of order one on the subspace carrying the state, the number of copies $M_j$ and the number of steps are $O(\text{poly}(n))$.
\end{theorem}

For a ground state preparation problem with dissipative dynamics, we have to specify the required resources and the a priori knowledge of the underlying Hamiltonian. In our approach we insist on a \textit{cold} start, with no prior knowledge about the initial state, eigenstates of the Hamiltonian, or locality.
We present two approaches for constructing the random jump operators. 
The first method (Section~\ref{sec:first-method}) seems to work well when the histogram of eigenvalues follows
a regular profile.
The second method (Section~\ref{sec:second-method}) was particularly designed to handle the small perturbation of $H=\sum_j Z_j$ with will separated plateaus in the spectrum.

For dynamics designed to converge to a target state, the difficulty lies in
bounding the mixing time.  The general theory bounds it through spectral gaps~\cite{Temme2010,KastoryanoTemme2013,GaoRouze2022}.  Rapid mixing is known
for commuting local Hamiltonians~\cite{KastoryanoBrandao2016,Bardet2023}, at
high temperature for Hamiltonians with a Lieb-Robinson
bound~\cite{RouzeStilckFrancaAlhambra2026}, and for weakly interacting
systems~\cite{TongZhan2025,Zhan2026Rapid}.  At inverse temperatures polynomial
in the system size, the same dynamics are computationally equivalent to
polynomial-time quantum computation~\cite{RouzeStilckFrancaAlhambra2026}, and
unconditional convergence from local measurements holds at a cost that is
exponential in general~\cite{Cubitt2023}.  Randomness enters several of these
approaches: through the eigenstate thermalization
hypothesis~\cite{ChenBrandao2021,Brunner2025,ShtankoMovassagh2021}, through
random interactions with ancillas~\cite{HaganWiebe2025,Langbehn2025} or
randomized cooling cycles~\cite{Molpeceres2025}, through coupling operators
sampled from a fixed set~\cite{Temme2011,DingZhanPreskillLin2026}, and in the
randomized simulation of Lindbladians~\cite{ChenLiLuYing2025}.  In our
construction we have to randomize the coupling operator $A=g$ where $g$ is random gaussian or unitary matrix. Only the eigenvalue distribution of $\mathcal H$ enters, and
from an arbitrary initial state we bound the evolution time and the number of
random matrices.

Section~\ref{sec:outline} describes succintly the iteration procedure encoded by the sequence of Lindblad operators presented in Theorem~\ref{th:informal}.
Section~\ref{sec:first-method} constructs corner jump operators with different filters
for each stage and compares their empirical generators with the ideal average.
Section~\ref{sec:second-method} separates frequency components so that,
for every realization, $\Psi_a$ lowers the corner and $\mathcal M_a$ preserves it.
Section~\ref{sec:examples} develops spectral criteria and examples for
the first method, followed by clustered and compressed constructions for
the second.
Appendix~\ref{app:probability} collects the probability estimates used
by both methods.

\section*{Notation}
Let $\mathcal{K}$ be an $N$-dimensional complex Hilbert space, and let
$H$ be a self-adjoint Hamiltonian on $\mathcal{K}$.  Fix an
orthonormal eigenbasis $\{\psi_r\}_{r=1}^N$, with
$H\psi_r=\lambda_r\psi_r$, and define
$E_{rs}=|\psi_r\rangle\langle\psi_s|$.  We write $A^*$ for the adjoint of
$A$ and $\operatorname{Tr}$ for the unnormalized trace, so that
$\operatorname{Tr}I=N$.  For an operator $A$ on $\mathcal{K}$, we use the
Schatten infinity norm
\[
    \|A\|_\infty:=\sup_{\|v\|_{\mathcal{K}}=1}\|Av\|_{\mathcal{K}},
\]
which is the operator norm induced by the Hilbert-space norm, or
equivalently, the largest singular value of $A$, and the Schatten-$1$
norm, or trace norm,
\[
    \|A\|_1:=\operatorname{Tr}|A|,\qquad |A|=(A^*A)^{1/2},
\]
which is the sum of the singular values of $A$.  The trace and the trace
norm agree on positive operators: $\operatorname{Tr}A=\|A\|_1$ for
$A\geq0$, while $|\operatorname{Tr}A|\leq\|A\|_1$ in general.  In
particular, for a density matrix $\rho$ and an orthogonal projection $P$,
$\operatorname{Tr}(P\rho)=\operatorname{Tr}(P\rho P)=\|P\rho P\|_1$.

The eigenvalues are ordered as $\lambda_1\leq\lambda_2\leq\cdots\leq\lambda_N$,
repeated with multiplicity, so that $E_{ll}$ is the eigenprojection onto
$\psi_l$ and $HE_{ll}=\lambda_lE_{ll}$.  For ranks
$1\leq N_0<N_1<\cdots<N_m=N$, the cumulative projections are
\[
    P_j:=\sum_{l=1}^{N_j}E_{ll},\qquad P_0\leq P_1\leq\cdots\leq P_m=I.
\]
If $\lambda_{N_j}<\lambda_{N_j+1}$, then $P_j=1_{(-\infty,\lambda_{N_j}]}(H)$
is a spectral projection of $H$.  In Section~\ref{sec:second-method},
$P_j$ is the sum of the cluster projections $\Pi_0,\ldots,\Pi_j$.

The diamond norm of quantum channels is further defined as 
\begin{equation}
    \|\Phi\|_\diamond
      =\sup_{k\geq1}\ \sup_{x\ne0}
        \frac{\|(\Phi\otimes\mathrm{id}_k)(x)\|_1}{\|x\|_1},
    \label{eq:sm-diamond-norm}
\end{equation}
where $\mathrm{id}_k$ acts on operators of an arbitrary $k$-dimensional
ancillary system, and $x$ acts on $\mathcal{K}\otimes\mathbb C^k$.

\section{Outline of the iteration procedure}\label{sec:outline}

The iteration method will use a suitable choice of decreasing spectral projections, as a theoretical tool. For any projection $P$ we define  
\begin{equation} \label{eq:sm-corner-maps}
        \mathcal{Q}_{P}(\rho)= P\rho P, \qquad \text{and} \qquad 
        \mathcal{Q}^{c}_{P}(\rho) = \rho-P\rho P,
\end{equation}
Thus $\mathcal{Q}_P$ compresses the density to the $P-P$ corner, is completely positive and trace reducing. Then satisfies 
 \[ \|\mathcal{Q}_P^c\|_{\diamond} \kl \|id\|_{\diamond}+\|\mathcal{Q}_P\|_{\diamond} \kl 2  \pl. \]
Note that $\mathcal{Q}_P^c(\rho)=(1-P)\rho (1-P)+P\rho(1-P)+(1-P)\rho P$ has three components. 

Now the iteration procedure consists of finding energy levels $N_j$ and jump operators 
  \[ a_j(\om_k) \lel \frac{1}{\sqrt{N_j}} \int_{\rz} f_j(t) e^{itH}g(\om_k) e^{-itH} dt \pl \]
for suitable function $N_j$. The Lindblad generator for a single jump operator is defined as
 \[ L_a(\rho) \lel 2 a\rho a^*-a^*a\rho-\rho a^*a 
 \lel 2 \Psi_a(\rho) + \mathcal{M}_{a}(\rho) \pl. \]
Here $\Psi_a$ is the positive part, compensated by the multiplication part $M_a$. In every step we consider 
 \[ \mathcal{L}_j \lel \frac{1}{M} \sum_{k} L_{a_j(\om_k)}\]
and empirical average approximating the ideal Lindbladian 
 \[ EL_j \lel \ez L_{a_j(\om)}   \pl. \]
The key inequality for \emph{conditional}  mixing we have to prove is 
 \begin{align}
     \|\mathcal{Q}_{P_{j-1}}^ce^{t_j\mathcal{L}_j}Q_{P_j}\|_{\diamond} <\delta_j \pl .
\end{align}
A similar estimate will first be obtained for the ideal Lindbladian and then derived for the empirical Lindbladian from concentration inequalities which unfortunately only apply to the multiplication part. Our goal is to estimate the diamond norm because it allows the energy reduction even for input densities entangled with an additional environment. By designing $f_j$ we find suitable algebraic relations so that
  \[ \mathcal{Q}_{P_{j-1}}^c\Psi_{a_j}\mathcal{Q}_{P_j} \lel 0 \] 
vanishes.   
Our  second method uses clustering of eigenvalues and modifies the structure of the Lindbladian, simplifying the random matrix estimates. A telescoping argument then shows that for the threshold  $\eta=\sum_j \delta_j$ we have  
 \[ \|\mathcal{Q}_{P_0}^ce^{t_1\mathcal{L}_1}\cdots e^{t_m\mathcal{L}_m}\|_{\diamond}   \kl 2\eta   \]
The mixing times $t_j$, the number of steps $m_j$ are determined by the eigenvalue distribution of the Hamiltonian $H$. The jump operators turn out to be conditionally bounded of order $O(1)$, meaning $\|a_jP_j\|\sim C(H)$ independent of the iteration step.

\section{Construction of jump operators for Hamiltonian with general spectrum}
\label{sec:first-method}
\subsection{Construction of the jump operators}
For the first channel $e^{t_m \mathcal{L}_m}$, we construct the jump
operator $a_m(\omega)$ as follows.  Let $\phi_m:\mathbb R\to\mathbb C$
be a filter function.
We shift $\mathcal{H}$ by $-\lambda_1 I$ and hence assume $\lambda_1=0$;
this leaves all energy differences and the filtered jump operators unchanged. 
Choose a smooth Fourier filter satisfying $0\leq\hat\phi_m\leq1$, with
$0<\alpha'_m<\alpha_m\leq\beta_m<\beta'_m$. For any $\xi \leq \alpha_m'$ or $\xi \geq \beta_m'$, we have
\begin{equation}
\hat{\phi}_m(\xi)=0,
\label{eq:first-filter-zero}
\end{equation}
while for $\alpha_m' <\xi < \beta_m'$,
\begin{equation}
	\hat{\phi}_m(\xi)>0,
\end{equation}
and for $ \alpha_m \leq \xi \leq \beta_m$,
\begin{equation}
\hat{\phi}_m(\xi)=1.
\label{eq:first-filter-plateau}
\end{equation}

We construct the random jump operator $a_m(\omega)$ as
\begin{equation}
a_m(\omega) :=\frac{1}{\sqrt{N}}\int_{-\infty}^{\infty}\phi_m(x)e^{ix\mathcal{H}}g(\omega) e^{-ix\mathcal{H}} \mathrm{d}x \qquad =
    \frac{1}{\sqrt{N}}\sum_{r,t=1}^{N} \hat{\phi}_m(\la_t-\la_r) E_{rr}g(\omega)E_{tt} \pl,
\end{equation}
where the entries of $g$ are independent standard real or complex
Gaussian variables.  Here $N_m=N$; at stage $j$ the factor is
$N_j^{-1/2}$, as in Proposition~\ref{prop:gammaass}. The empirical average of sufficiently many such
jump operators gives the preparation guarantee in
Theorem~\ref{thm:first-empirical-preparation}.  With this Fourier
convention, support of $\hat{\phi}$ on the positive real line makes
the jump operator lower energy.
Filters of this form that lower energy also appear
in~\cite{DingChenLin2024,Zhan2026Rapid}.
Write the filter coefficients as

\begin{equation}
    m^{m}_{rt}:=\hat{\phi}_m(\la_t-\la_r).
\end{equation}
We also use the spectral counting function
\begin{equation}
    r_{\eta}(s):=\#\{l:\lambda_l\leq \lambda_s-\eta\},
    \qquad \eta\in\mathbb{R}.
\end{equation}
For $\eta>0$, $r_{\eta}(s)$ counts the eigenvalues satisfying
$\lambda_l\leq\lambda_s-\eta$.  Negative values give enlarged upper cutoffs; for
example, $r_{-\beta_j'}(N_j)$ counts all eigenvalues at most
$\lambda_{N_j}+\beta_j'$.
The following support condition applies at every stage $j$, starting
with $j=m$.
\begin{prop}
\label{prop:boxass}
For a Hamiltonian $\mathcal{H}$ and a smooth Fourier filter
$0\leq\hat\phi_j\leq1$ satisfying
\eqref{eq:first-filter-zero}--\eqref{eq:first-filter-plateau} with
$m$ replaced by $j$ and $0<\alpha_j'<\alpha_j\leq\beta_j<\beta_j'$, suppose
there are positive integers \notationrev{$q_j,N_{j-1}$} such that

\begin{enumerate}
    \item $\la_{\notationrev{q_j}}-\lambda_1 \leq \alpha_j'$,
    \item $\lambda_{N_j}-\lambda_{q_j+1}\leq\alpha_j'$,
    \item \notationrev{$q_j\leq N_{j-1}<N_j$}.
\end{enumerate}
Then, on the active input corner $P_j$, the jump operator $a_j$ satisfies
\begin{equation}
    a_jP_j = Q_j a_j(P_j-\notationrev{Q_j})
\end{equation}
\notationrev{where $Q_j=\sum_{l=1}^{q_j}E_{ll}$.}
\end{prop}
\begin{proof}
Write $m^j_{rt}:=\hat\phi_j(\lambda_t-\lambda_r)$.  In
Figure~\ref{fig:jump}, $t$ is the horizontal input column index and $r$
the vertical output row index.  For $t\leq\notationrev{q_j}$,
\[
 \lambda_t-\lambda_r\leq\lambda_{q_j}-\lambda_1\leq\alpha'_j,
\]
so all these input columns vanish.  For $r\geq q_j+1$ and
$t\leq N_j$,
\[
 \lambda_t-\lambda_r\leq\lambda_{N_j}-\lambda_{q_j+1}
 \leq\alpha'_j,
\]
so every output row outside $Q_j$ vanishes on the active corner.  Thus
\begin{equation}
 \{(r,t):m^j_{rt}\ne0,\ 1\leq t\leq N_j\}
 \subseteq\{1,\ldots,q_j\}\times\{q_j+1,\ldots,N_j\},
\end{equation}
which proves $a_jP_j=Q_j a_j(P_j-Q_j)$.  \notationrev{Finally, $q_j\leq N_{j-1}$ gives $Q_j\leq P_{j-1}$.}
The filter conditions are imposed at each stage; no filter is needed
on the final corner $P_0$.

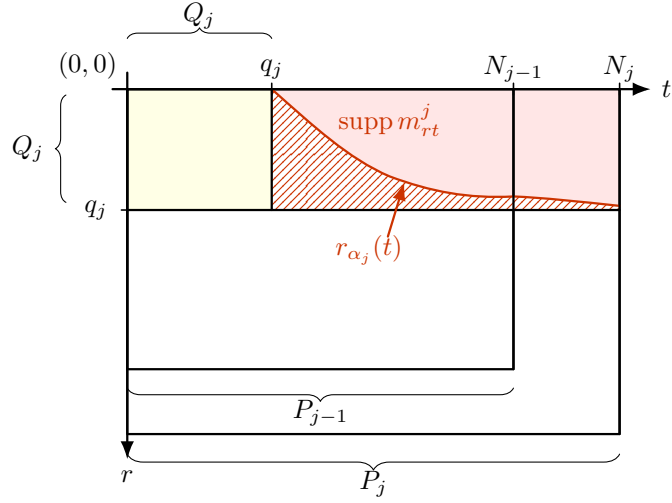
\begin{figure}
    \centering
    \begin{tikzpicture}[
        x=0.78cm,
        y=-0.78cm,
        >=Latex,
        line cap=round,
        line join=round,
        every node/.style={font=\small},
        tick/.style={line width=0.6pt},
        guide/.style={gray, dashed, line width=0.55pt},
        mainbox/.style={black, line width=0.85pt},
        redmark/.style={DarkRed, line width=0.85pt},
        bluemark/.style={DarkBlue, line width=0.85pt}
    ]
        \def\xT{2.45}
        \def\xNm{6.55}
        \def\xN{8.35}
        \def\yT{2.05}
        \def\yNm{4.75}
        \def\yN{5.85}

        \path[fill=VeryLightYellow] (0,0) rectangle (\xT,\yT);
        \path[fill=VeryLightRed] (\xT,0) rectangle (\xN,\yT);
        \draw[mainbox] (0,0) rectangle (\xN,\yN);
        \draw[mainbox] (0,0) rectangle (\xNm,\yNm);

        \path[fill=white]
            (\xT,0)
            .. controls (3.15,0.62) and (3.75,1.18) .. (4.45,1.48)
            .. controls (5.15,1.74) and (5.78,1.86) .. (\xNm,1.82)
            .. controls (7.05,1.84) and (7.70,1.91) .. (\xN,1.98)
            -- (\xN,\yT) -- (\xT,\yT) -- cycle;
        \path[pattern=north east lines, pattern color=DarkRed]
            (\xT,0)
            .. controls (3.15,0.62) and (3.75,1.18) .. (4.45,1.48)
            .. controls (5.15,1.74) and (5.78,1.86) .. (\xNm,1.82)
            .. controls (7.05,1.84) and (7.70,1.91) .. (\xN,1.98)
            -- (\xN,\yT) -- (\xT,\yT) -- cycle;
        \draw[redmark]
            (\xT,0)
            .. controls (3.15,0.62) and (3.75,1.18) .. (4.45,1.48)
            .. controls (5.15,1.74) and (5.78,1.86) .. (\xNm,1.82)
            .. controls (7.05,1.84) and (7.70,1.91) .. (\xN,1.98);

        \draw[mainbox] (\xT,0) -- (\xT,\yT);
        \draw[mainbox] (\xNm,0) -- (\xNm,\yNm);
        \draw[mainbox] (\xN,0) -- (\xN,\yN);
        \draw[mainbox] (0,\yT) -- (\xN,\yT);

        \draw[->, mainbox] (-0.25,0) -- (8.88,0) node[right] {$t$};
        \draw[->, mainbox] (0,-0.28) -- (0,6.28) node[below] {$r$};

        \draw[tick] (\xT,-0.1) -- (\xT,0.1);
        \node at (\xT,-0.34) {$\notationrev{q_j}$};
        \draw[tick] (\xNm,-0.1) -- (\xNm,0.1);
        \node at (\xNm,-0.34) {$\notationrev{N_{j-1}}$};
        \draw[tick] (\xN,-0.1) -- (\xN,0.1);
        \node at (\xN,-0.34) {$N_j$};
        \draw[tick] (-0.1,\yT) -- (0.1,\yT);
        \node[left] at (-0.16,\yT) {$q_j$};
        \node[above left] at (0,0) {$(0,0)$};

        \node[DarkRed] at (4.45,0.52) {$\supp m^j_{rt}$};
        \node[DarkRed] at (4.10,2.72) {$r_{\alpha_j}(t)$};
        \draw[->, redmark] (4.42,2.50) -- (4.70,1.58);

        \draw[decorate,decoration={brace,amplitude=4pt,mirror}]
            (-1.04,0.10) -- (-1.04,1.92);
        \node[left] at (-1.24,1.01) {$Q_j$};

        \draw[decorate,decoration={brace,amplitude=4pt}]
            (0,-0.86) -- (\xT,-0.86);
        \node at (1.22,-1.24) {$\notationrev{Q_j}$};

        \draw[decorate,decoration={brace,amplitude=5pt,mirror}]
            (0,\yNm+0.32) -- (\xNm,\yNm+0.32);
        \node at (3.28,\yNm+0.74) {$P_{j-1}$};

        \draw[decorate,decoration={brace,amplitude=6pt,mirror}]
            (0,\yN+0.34) -- (\xN,\yN+0.34);
        \node at (4.18,\yN+0.84) {$P_j$};
    \end{tikzpicture}
    \caption{Support region for the relation on the active corner
    $aP_j=Q_j a(P_j-\notationrev{Q_j})$, with horizontal input index $t$, vertical
    output index $r$, and thresholds $\notationrev{q_j}$, $\notationrev{N_{j-1}}$ and
    $N_j$.}
    \label{fig:jump}
\end{figure}

\end{proof}

\subsection{Convergence of ideal Lindbladian dynamics}
\label{subsec:fm-ideal}

At stage $j$, the ideal dynamics is generated by
$\mathbb{E}[\mathcal{L}_{a_j}]$.

\begin{prop}
\label{prop:gammaass}
Under the hypotheses of Proposition~\ref{prop:boxass}, suppose at stage $j$ that
$\alpha_j\leq\lambda_{N_j}-\lambda_1\leq\beta_j$ and
\notationrev{$\gamma_j>0$ satisfies}
$r_{\alpha_j}(s) \geq \gamma_j N_{j}$ for every \notationrev{$N_{j-1}<s\leq N_j$}.
Let $g_{rt}$ be independent standard real or complex Gaussian entries,
normalized by $\mathbb E\overline{g_{rt}}g_{r't'}=\delta_{rr'}\delta_{tt'}$,
and define

\begin{equation}
    a_j(\omega) \lel \sum_{r,t=1}^{N} \hat{\phi}_j(\la_t-\la_r) g_{rt}(\omega)/\sqrt{N_j}E_{rt} \pl. 
\end{equation}
Then
\begin{equation}
    \notationrev{(P_j-P_{j-1})\mathbb{E}[a^{*}_j a_j](P_j-P_{j-1})\geq \gamma_j (P_j-P_{j-1}).}
\end{equation}

\end{prop}
\begin{proof}
Independence and the entry covariance give the diagonal operator
\begin{equation}
 B_j:=\mathbb E[a_j^*a_j]
 =\frac1{N_j}\sum_{t=1}^N
     \left(\sum_{r=1}^N|\hat\phi_j(\lambda_t-\lambda_r)|^2\right)E_{tt}.
\end{equation}
For \notationrev{$N_{j-1}<t\leq N_j$} and $\lambda_r\leq\lambda_t-\alpha_j$,
the gap is at least $\alpha_j$ and at most
$\lambda_{N_j}-\lambda_1\leq\beta_j$.  These $r_{\alpha_j}(t)$ entries
therefore lie on the plateau, giving
\[
 \sum_r|\hat\phi_j(\lambda_t-\lambda_r)|^2
 \geq r_{\alpha_j}(t)\geq\gamma_j N_j.
\]
This proves the bound on \notationrev{$P_j-P_{j-1}$}.  The operator $B_j$ and its
lower bound must be determined separately at each stage.
\end{proof}

\begin{prop}
\label{prop:first-unitary-mass}
    Under the assumptions of Proposition~\ref{prop:gammaass}, let $u(\omega)$ be drawn
    from an exact unitary $t$-design with $t\geq1$ on
    $V_j$~\cite{Dankert2009}.  Each entry has second moment $1/N_j$,
    so no further normalization is needed. A design on $V_j$ presupposes
    access to this spectral subspace; for a unitary on a larger coordinate
    subspace of dimension $D$, the factor $\sqrt{D/N_j}$ of
    Lemma~\ref{lem:corner-support-current} restores the variance. Define
    \begin{equation}
    a_j(\omega) \lel \sum_{r,t=1}^{N_j} \hat{\phi}_j(\la_t-\la_r) u_{rt}(\omega)E_{rt} \pl. 
\end{equation}
Then
\begin{equation}
    \notationrev{(P_j-P_{j-1})\mathbb{E}[a^{*}_j a_j](P_j-P_{j-1})\geq \gamma_j (P_j-P_{j-1}).}
\end{equation}
\end{prop}
\begin{proof}

Expanding the product gives

\begin{equation}
    a_j^{*}(\omega)a_j(\omega) \lel \sum_{r',t'=1}^{N_j}\sum_{r,t=1}^{N_j} \hat{\phi}_j^{*}(\la_{t'}-\la_{r'}) \hat{\phi}_j(\la_t-\la_r) u^{*}_{r't'}(\omega)u_{rt}(\omega)E_{t'r'}E_{rt}.
\end{equation}
Using $E_{t'r'}E_{rt}=\delta_{r'r}E_{t't}$, we obtain
\begin{equation}
    a_j^{*}(\omega)a_j(\omega) \lel \sum_{t'=1}^{N_j}\sum_{r,t=1}^{N_j} \hat{\phi}_j^{*}(\la_{t'}-\la_{r}) \hat{\phi}_j(\la_t-\la_r) u^{*}_{rt'}(\omega)u_{rt}(\omega)E_{t't}.
\end{equation}
The identity
\begin{equation}
    \mathbb{E}_{U}[ u^{*}_{rt}(\omega)u_{rt'}(\omega)]= \frac{1}{N_j}\delta_{tt'}
\end{equation}
gives
\begin{equation}
    \mathbb{E}[a_j^{*}a_j]=\frac{1}{N_j} \sum_{r,t=1}^{N_j} \hat{\phi}_j^{*}(\la_{t}-\la_{r}) \hat{\phi}_j(\la_t-\la_r) E_{tt}.
\end{equation}
The same counting argument as in the Gaussian case proves the bound.
\end{proof}

\begin{lemma}[Block action of a corner jump operator]
\label{lem:block-corner-current}
Let $E\leq F\leq G$ be orthogonal projections onto subspaces spanned by
vectors in the fixed Hamiltonian eigenbasis $\{\psi_l\}_{l=1}^N$, and let
$a(\omega)$ be a family of jump operators satisfying
\begin{equation}
    a(\omega)G=E a(\omega)(G-E)
\end{equation}
for every $\omega$.  Put
\begin{equation}
    B:=\mathbb{E}[a(\omega)^*a(\omega)].
\end{equation}
Assume that $B$ commutes with $F$ and $G$.  Then
\begin{equation}
    GBG=(G-E)B(G-E),
\end{equation}
and, for every operator $x$ with $x=GxG$,
\begin{equation}
    \mathcal{Q}^{c}_{F}\mathbb{E}[\mathcal{L}_{a}]\mathcal{Q}_{G}(x)
    =
    -\mathcal{Q}^{c}_{F}(Bx+xB).
    \label{eq:block-corner-generator}
\end{equation}
With $R:=G-F$, the part of $x$ outside the $F$ corner decomposes as
\begin{equation}
    \mathcal{Q}^{c}_{F}\mathcal{Q}_{G}(x)=FxR+RxF+RxR.
\end{equation}
Moreover, if
\begin{equation}
    RBR\geq \gamma R,
\end{equation}
then
\begin{equation}
    \|\mathcal{Q}^{c}_{F}e^{t\mathbb{E}[\mathcal{L}_{a}]}
    \mathcal{Q}_{G}\|_{\diamond}
    \leq 2e^{-\gamma t}+e^{-2\gamma t}.
\end{equation}
\end{lemma}
\begin{proof}
By $aG=Ea(G-E)$, the operator $a$ maps $\mathrm{range}(G)$
into $\mathrm{range}(E)$, and $a$ annihilates $\mathrm{range}(E)$.
Hence
\begin{equation}
    EaG=aG,\qquad aE=0.
\end{equation}
Taking adjoints of $EaG=aG$ gives
\begin{equation}
    Ga^*E=Ga^*.
\end{equation}
Therefore
\begin{equation}
    BE=\mathbb{E}[a^*a]E=0,\qquad
    EB=0.
\end{equation}
Since $B$ commutes with $G$, the restriction of $B$ to $\mathrm{range}(G)$
is represented by $GBG$.  The two identities $BE=EB=0$ imply
$GBG=(G-E)B(G-E)$.

Decompose $x=GxG$ relative to
\begin{equation}
    \mathrm{range}(G)
    =
    \mathrm{range}(F)\oplus\mathrm{range}(R),
    \qquad R=G-F.
\end{equation}
Thus
\begin{equation}
    x=FxF+FxR+RxF+RxR.
\end{equation}
Since $B$ commutes with $F$ and $G$, it also commutes with $R$.  The
anticommutator part $-Bx-xB$ respects this block decomposition: the block
$FxF$ is sent back into the $F$ corner, while $FxR$, $RxF$, and $RxR$ remain in
the complement of the $F$ corner.  The term $axa^*$ has range in the $E$ corner:
\begin{equation}
    a x a^*=aGxGa^*=E(aGxGa^*)E.
\end{equation}
Since $E\leq F$, this term is killed by $\mathcal{Q}^{c}_{F}$.  Since also
$x=GxG$, we have $\mathcal{Q}_{G}(x)=x$, and hence
\begin{equation}
\begin{split}
    \mathcal{Q}^{c}_{F}\mathbb{E}[\mathcal{L}_{a}]\mathcal{Q}_{G}(x)
    &=
    \mathcal{Q}^{c}_{F}\mathbb{E}[\mathcal{L}_{a}](x)\\
    &=
    \mathcal{Q}^{c}_{F}\mathbb{E}[2axa^*-a^*ax-xa^*a]\\
    &=
    -\mathcal{Q}^{c}_{F}(Bx+xB),
\end{split}
\end{equation}
which proves \eqref{eq:block-corner-generator}.

The preceding computation also gives the evolution on the three complementary
blocks.  If
\begin{equation}
    y(t):=\mathcal{Q}^{c}_{F}e^{t\mathbb{E}[\mathcal{L}_{a}]}
    \mathcal{Q}_{G}(x),
\end{equation}
then $y(t)$ solves
\begin{equation}
    \frac{d}{dt}y(t)=-By(t)-y(t)B,\qquad
    y(0)=FxR+RxF+RxR.
\end{equation}
Therefore
\begin{equation}
    y(t)=e^{-tB}(FxR)e^{-tB}
    +e^{-tB}(RxF)e^{-tB}
    +e^{-tB}(RxR)e^{-tB}.
\end{equation}
Because $R$ commutes with $B$ and $RBR\geq \gamma R$, functional calculus gives
\begin{equation}
    \|e^{-tB}R\|_{\infty}
    =
    \|e^{-tRBR}R\|_{\infty}
    \leq e^{-\gamma t},
    \qquad
    \|e^{-tB}F\|_{\infty}\leq 1.
\end{equation}
Left and right multiplication by a matrix have completely bounded norm equal to
the operator norm of that matrix.  Thus the two off-diagonal blocks $FxR$ and
$RxF$ each contribute at most $e^{-\gamma t}$ in diamond norm, and the $RxR$
block contributes at most $e^{-2\gamma t}$.  The triangle inequality gives
\begin{equation}
    \|\mathcal{Q}^{c}_{F}e^{t\mathbb{E}[\mathcal{L}_{a}]}
    \mathcal{Q}_{G}\|_{\diamond}
    \leq 2e^{-\gamma t}+e^{-2\gamma t}.
\end{equation}
\end{proof}

\begin{theorem}
\label{thm:ideal-corner}
Fix a stage $j$. Assume the hypotheses of Proposition~\ref{prop:boxass} and
Proposition~\ref{prop:gammaass}.  Then \notationrev{$Q_j\leq P_{j-1}\leq P_j$}
by construction.  Assume in addition that
\begin{equation}
    a_j(\omega)P_j=Q_j a_j(\omega)(P_j-Q_j)
\end{equation}
for every realization $\omega$. Let
\begin{equation}
    B_j:=\mathbb{E}[a_j^{*}a_j].
\end{equation}
Then
\begin{equation}
    \|\mathcal{Q}^{c}_{\notationrev{P_{j-1}}} e^{t\mathbb{E}[\mathcal{L}_{a_j}]}
    \mathcal{Q}_{P_j}\|_{\diamond}
    \leq 2e^{-\gamma_j t }+e^{-2\gamma_j t}.
\end{equation}
In particular, when $j=m$ and $P_m=I$,
\begin{equation}
    \|\mathcal{Q}^{c}_{\notationrev{P_{m-1}}} e^{t\mathbb{E}[\mathcal{L}_{a_m}]}\|_{\diamond}
    \leq 2e^{-\gamma_m t }+e^{-2\gamma_m t}.
\end{equation}
\end{theorem}

\begin{proof}
Apply Lemma~\ref{lem:block-corner-current} with $E=Q_j$, $G=P_j$ and
$a=a_j$; the assumption $a_jP_j=Q_ja_j(P_j-Q_j)$ is its first hypothesis,
and \notationrev{$Q_j\leq P_{j-1}\leq P_j$}.  By the computations in the proofs of
Propositions~\ref{prop:gammaass} and~\ref{prop:first-unitary-mass},
$B_j$ is diagonal in the eigenbasis, so it commutes with \notationrev{$P_{j-1}$}
and $P_j$, and Proposition~\ref{prop:gammaass} gives
\notationrev{$(P_j-P_{j-1})B_j(P_j-P_{j-1})\geq\gamma_j(P_j-P_{j-1})$}.
\notationrev{The lemma with $F=P_{j-1}$ gives the stated bound.}
For $j=m$, $\mathcal Q_{P_m}$ is the identity.
\end{proof}

The following lemma bounds the accumulated error under composition.
It also applies when the empirical channels do not preserve the input
corners.

\begin{lemma}[Composition of quantum channels]
\label{lem:corner-composition}
Let $P_0\leq\cdots\leq P_m=I$ be orthogonal projections and let
$\Phi_1,\ldots,\Phi_m$ be quantum channels satisfying
\[
 \|\mathcal Q_{P_{j-1}}^c\Phi_j\mathcal Q_{P_j}\|_\diamond
 \leq\delta_j,\qquad
 \eta:=\sum_{j=1}^m\delta_j<1.
\]
Then
\begin{equation}
 \|\mathcal Q_{P_0}^c\Phi_1\cdots\Phi_m\|_\diamond\leq2\eta.
 \label{eq:corner-composition-error}
\end{equation}
For every density matrix $\rho$, there is a density matrix $\sigma$
supported in $P_0$ such that
\begin{equation}
 \|\Phi_1\cdots\Phi_m(\rho)-\sigma\|_1\leq2\eta.
 \label{eq:corner-composition-state}
\end{equation}
\end{lemma}
\begin{proof}
Set $C_j:=\mathcal Q_{P_{j-1}}\Phi_j\mathcal Q_{P_j}$.
Each $C_j$ is completely positive and trace-nonincreasing and has
diamond norm at most one.  Telescoping gives
\[
 \Phi_1\cdots\Phi_m-C_1\cdots C_m
 =\sum_{j=1}^m\Phi_1\cdots\Phi_{j-1}(\Phi_j-C_j)C_{j+1}\cdots C_m.
\]
The range of $C_{j+1}\cdots C_m$ lies in the $P_j$ corner
(and $P_m=I$ for $j=m$).  Thus each summand is bounded by $\delta_j$,
and the whole difference has diamond norm at most $\eta$.
Since $C_1\cdots C_m$ has range in $P_0$ and
$\|\mathcal Q_{P_0}^c\|_\diamond\leq2$, this proves
\eqref{eq:corner-composition-error}.

For a state $\rho$, put $\tau=C_1\cdots C_m(\rho)$.
It is positive, supported in $P_0$, and satisfies
$1-\eta\leq\operatorname{Tr}\tau\leq1$ and
$\|\Phi_1\cdots\Phi_m(\rho)-\tau\|_1\leq\eta$.
The normalized state $\sigma=\tau/\operatorname{Tr}\tau$ obeys
$\|\sigma-\tau\|_1=1-\operatorname{Tr}\tau\leq\eta$, proving
\eqref{eq:corner-composition-state}.  The compressed maps $C_j$
are used only in the proof; the physical evolution applies the $\Phi_j$.
\end{proof}

\begin{theorem}[State preparation with ideal channels]
\label{thm:ideal-iteration}
Suppose that every stage $1\leq j\leq m$ satisfies
Theorem~\ref{thm:ideal-corner}, with $P_m=I$ and decay constants
$\gamma_j>0$.  For $0<\eps<1$, choose
\[
 t_j=\gamma_j^{-1}\log(6m/\eps).
\]
Then
\[
 \|\mathcal Q_{P_0}^c e^{t_1\mathbb E[\mathcal L_{a_1}]}\cdots
    e^{t_m\mathbb E[\mathcal L_{a_m}]}\|_\diamond\leq\eps.
\]
For every initial state, the output is within trace norm $\eps$ of
a normalized state supported in $P_0$.
\end{theorem}
\begin{proof}
Apply Theorem~\ref{thm:ideal-corner} directly to the target $P_{j-1}$.
Its error is at most $3e^{-\gamma_jt_j}=\eps/(2m)$.
Lemma~\ref{lem:corner-composition} gives both conclusions.
\end{proof}

\subsection{Empirical generators}

We compare the empirical and ideal generators at stage $j$.
With the row and column convention of Figure~\ref{fig:jump}, the
support condition $\supp\hat{\phi}_j\subseteq[\alpha_j',\beta_j']$
implies that $a_{rt}$ vanishes unless the energy difference between
the input $t$ and output $r$ satisfies
\begin{equation}
    \alpha_j'\leq \lambda_t-\lambda_r\leq\beta_j'.
\end{equation}
The localization of $P_j a_j^*a_j$ in
Lemma~\ref{lem:corner-support-current} uses the upper bound
$\lambda_t-\lambda_r\leq\beta_j'$.
When we form $a_j^*a_j$, two input columns can be coupled only through a common
output row.  Therefore a column in the $P_j$ corner can only couple to columns
with energy at most $\lambda_{N_j}+\beta_j'$.
We write
\begin{equation}\label{eq:fm-enlarged-support}
    N_j^{\beta_j'}:=r_{-\beta_j'}(N_j),
    \qquad
    P_j^{\beta_j'}:=\sum_{l=1}^{N_j^{\beta_j'}}E_{ll}.
\end{equation}
Thus $P_j^{\beta_j'}$ is the orthogonal projection onto
\begin{equation}
    \operatorname{span}\{\psi_l:\lambda_l\leq\lambda_{N_j}+\beta_j'\}.
\end{equation}
In Figure~\ref{fig:jump}, the enlargement includes all columns that
can couple to $P_j$ through a common output row under the cutoff
$\beta_j'$.  It contains $V_j=\mathrm{range}(P_j)$, so
$P_j\leq P_j^{\beta_j'}$.  The projection $P_j^{\beta_j'}$ describes
the possible couplings in $a_j^*a_j$; \notationrev{$P_{j-1}$} still determines the target subspace.
Thus
\begin{equation}
    \notationrev{P_{j-1}}\leq P_j\leq P_j^{\beta_j'}.
\end{equation}
The enlargement to $P_j^{\beta_j'}$ accounts for the couplings created
by $a_j^*a_j$ from the $P_j$ corner.

\begin{lemma}
\label{lem:corner-support-current}
\label{lem:empiricalaverage1}
At the $j$-th stage, $Q_j$ is the output projection, $P_j$ is the active input
projection, and $P_j^{\beta_j'}$ is the enlarged spectral projection determined
by the outer support of the filter.
Let $a(\omega)$ be a random corner jump operator,
\begin{equation}
    a(\omega)=\sum_{r,t}a_{rt}(\omega)E_{rt},
\end{equation}
such that, on the active input corner,
\begin{equation}
    a(\omega)P_j=Q_j a(\omega)(P_j-Q_j)
\end{equation}
and whose entries satisfy the support condition for $\hat{\phi}_j$:
with the convention of Figure~\ref{fig:jump}, $a_{rt}(\omega)=0$ unless
$\alpha_j'\leq \lambda_t-\lambda_r\leq \beta_j'$.  Set
\begin{equation}
    X(\omega):=P_j a(\omega)^*a(\omega)P_j^{\beta_j'}
    =(P_j a(\omega)P_j)^*(P_j a(\omega)P_j^{\beta_j'}).
\end{equation}
Then, for every operator $\rho$,
\begin{equation}
    \mathcal{Q}^{c}_{Q_j}\mathcal{L}_{a(\omega)}\mathcal{Q}_{P_j}(\rho)
    =
    -\mathcal{Q}^{c}_{Q_j}\bigl(
    X(\omega)^*\mathcal{Q}_{P_j}(\rho)
    +\mathcal{Q}_{P_j}(\rho)X(\omega)\bigr),
\end{equation}
equivalently, the expression is unchanged after inserting the corner projection
$\mathcal{Q}_{P_j^{\beta_j'}}$:
\begin{equation}
    \mathcal{Q}^{c}_{Q_j}\mathcal{L}_{a(\omega)}\mathcal{Q}_{P_j}
    =
    \mathcal{Q}^{c}_{Q_j}\mathcal{Q}_{P_j^{\beta_j'}}
    \mathcal{L}_{a(\omega)}\mathcal{Q}_{P_j}.
    \label{eq:first-generator-localization}
\end{equation}
The same identities remain valid after left multiplication by
$\mathcal{Q}^{c}_{R}$ for any spectral projection
$R$ with $Q_j\leq R\leq P_j$, in particular for $R=P_{j-1}$.
Moreover, let $|m_{rt}|\leq1$ be the deterministic filter coefficients and
use either random model with $aP_j=Q_ja(P_j-Q_j)$ and $m_{rt}=0$ unless
$\alpha_j'\leq\lambda_t-\lambda_r\leq\beta_j'$:
\[
 a=N_j^{-1/2}\sum_{r,t}m_{rt}g_{rt}E_{rt},
 \qquad\text{or}\qquad
 a=\sqrt{D/N_j}\sum_{r,t\in\mathcal V}m_{rt}u_{rt}E_{rt}.
\]
In the first model, the entries of $g$ are independent standard real or
complex Gaussians.  In the second, $\mathcal V$ is a coordinate subspace
of dimension $D\leq d$ containing $P_j\mathcal{K}$, and $u$ is a unitary
on $\mathcal V$, extended by zero on its orthogonal complement.  The
choice $\mathcal V=P_j\mathcal{K}$ gives $D=N_j$ and the unitary
normalization in Proposition~\ref{prop:first-unitary-mass}.
Let $a^{(1)},\ldots,a^{(M)}$ be independent
copies and define
\[
 \mathcal L_{j,M}=\frac1M\sum_{\nu=1}^M\mathcal L_{a^{(\nu)}}.
\]
Fix $0<\eps,\delta<1$ and an even integer
\begin{equation}
 p\geq\max\{2,\log d,\log(1/\delta)\}.
 \label{eq:fm-moment-order}
\end{equation}
For the unitary model, take Haar samples or an exact unitary $t$-design
with $t\geq2p$.  There is a universal constant $C$ such that
\begin{equation}
 M\geq C\eps^{-2}\frac{N_j({\beta_j'})}{N_j}\,p^4
 \label{eq:fm-corrected-matrix-count}
\end{equation}
implies, with probability at least $1-\delta$,
\begin{equation}
 \left\|\mathcal Q_R^c
       (\mathcal L_{j,M}-\mathbb E\mathcal L_a)
                         \mathcal Q_{P_j}\right\|_\diamond\leq\eps
 \qquad(Q_j\leq R\leq P_j).
 \label{eq:fm-empirical-generator-error}
\end{equation}
These estimates hold on a single event for all such $R$.
An approximate design requires a separate bound on its moment error.
\end{lemma}
\begin{proof}
To localize $a^*a$ on the input corner, write its entries as sums
over common output rows.  The entry of $a^*a$ coupling columns $t$
and $t'$ is a sum over $r$:
\begin{equation}
    (a^*a)_{t't}=\sum_r \overline{a_{rt'}}a_{rt}.
\end{equation}
Hence $P_j a^*a$ has the expansion
\begin{equation}
    P_j a^*a
    =
    \sum_{t=1}^{N_j}\sum_{t'=1}^{d}
    \left(\sum_r\overline{a_{rt}}a_{rt'}\right)E_{tt'}.
\end{equation}
If $t\leq N_j$ and the coefficient of $E_{tt'}$ has a nonzero summand, then
both $a_{rt}$ and $a_{rt'}$ are nonzero for the same row $r$.  The upper support
bound gives
\begin{equation}
    \lambda_{t'}\leq \lambda_r+\beta_j'
    \leq \lambda_t+\beta_j'
    \leq \lambda_{N_j}+\beta_j'.
\end{equation}
Thus $\psi_{t'}\in\mathrm{range}(P_j^{\beta_j'})$, which means
\begin{equation}
    P_j a^*a=P_j a^*aP_j^{\beta_j'}.
\end{equation}
Moreover, because
$a(\omega)P_j=Q_j a(\omega)(P_j-Q_j)$ and $Q_j\leq P_j$, the output of
$a(\omega)P_j$ lies in the $P_j$ row space.  Hence we may insert $P_j$ between
$a^*$ and $a$:
\begin{equation}
    P_j a^*a
    =
    P_j a^*P_j aP_j^{\beta_j'}
    =
    (P_j aP_j)^*(P_j aP_j^{\beta_j'})
    =
    X(\omega).
\end{equation}
Taking adjoints gives the corresponding identity for the support on the right
\begin{equation}
    a^*aP_j=P_j^{\beta_j'}a^*aP_j=X(\omega)^*.
\end{equation}

Now let $\rho_j=\mathcal{Q}_{P_j}(\rho)=P_j\rho P_j$.  Expanding the Lindbladian
gives
\begin{equation}
    \mathcal{L}_{a(\omega)}(\rho_j)
    =
    2a(\omega)\rho_j a(\omega)^*
    -a(\omega)^*a(\omega)\rho_j-\rho_j a(\omega)^*a(\omega).
\end{equation}
The term $2a(\omega)\rho_ja(\omega)^*$ has range in $Q_j$, because
$a(\omega)P_j=Q_j a(\omega)(P_j-Q_j)$.  Hence
\begin{equation}
    \mathcal{Q}^{c}_{Q_j}\bigl(2a(\omega)\rho_j a(\omega)^*\bigr)=0.
\end{equation}
For the anticommutator part, $P_j a^*a=X(\omega)$ and
$a^*aP_j=X(\omega)^*$ give
\begin{equation}
    a(\omega)^*a(\omega)\rho_j=X(\omega)^*\rho_j,
    \qquad
    \rho_j a(\omega)^*a(\omega)=\rho_j X(\omega).
\end{equation}
This proves the first identity.  Since $X(\omega)$ has left support in $P_j$
and right support in $P_j^{\beta_j'}$, while $X(\omega)^*$ has left support in
$P_j^{\beta_j'}$ and right support in $P_j$, the two terms
$X(\omega)^*\rho_j$ and $\rho_jX(\omega)$ both lie in the
$P_j^{\beta_j'}$ corner.  This proves the localization identity with
$\mathcal{Q}_{P_j^{\beta_j'}}$.  If $Q_j\leq R\leq P_j$, then
$\mathcal{Q}^{c}_{R}=\mathcal{Q}^{c}_{R}\mathcal{Q}^{c}_{Q_j}$ on the left of
these expressions, so the same identities hold with $\mathcal{Q}^{c}_{R}$.

The probabilistic assertion is proved in Appendix~\ref{app:first-loss-proof},
using \eqref{eq:first-generator-localization}.
\end{proof}

\begin{lemma}
\label{lem:semigroup-current}
Assume $Q_j\leq P_{j-1}\leq P_j$,
\begin{equation}
    a(\omega)P_j=Q_j a(\omega)(P_j-Q_j),
\end{equation}
and that $\mathbb{E}[a(\omega)^*a(\omega)]$ commutes with $P_j$.  Let
$\mathcal{L}_{j,M}$ denote the empirical generator defined in
Lemma~\ref{lem:corner-support-current}.  Then, for every $t\geq 0$,
\begin{equation}
    \|\mathcal{Q}^{c}_{Q_j}
    (e^{t\mathbb{E}[\mathcal{L}_a]}-e^{t\mathcal{L}_{j,M}})
    \mathcal{Q}_{P_j}\|_{\diamond}
    \leq
    2t\|\mathcal{Q}^{c}_{Q_j}
    (\mathbb{E}[\mathcal{L}_a]-\mathcal{L}_{j,M})
    \mathcal{Q}_{P_j}\|_{\diamond}.
\end{equation}
Consequently,
\begin{equation}
    \|\mathcal{Q}^{c}_{P_{j-1}}
    (e^{t\mathbb{E}[\mathcal{L}_a]}-e^{t\mathcal{L}_{j,M}})
    \mathcal{Q}_{P_j}\|_{\diamond}
    \leq
    2t\|\mathcal{Q}^{c}_{Q_j}
    (\mathbb{E}[\mathcal{L}_a]-\mathcal{L}_{j,M})
    \mathcal{Q}_{P_j}\|_{\diamond}.
\end{equation}
\end{lemma}
\begin{proof}
The corner relation $aP_j=Q_j a(P_j-Q_j)$ gives
\begin{equation}
    aQ_j=aP_jQ_j=Q_j a(P_j-Q_j)Q_j=0.
\end{equation}
Consequently, if $x=Q_jxQ_j$, then
\begin{equation}
    axa^*=0,\qquad a^*ax=0,\qquad xa^*a=0,
\end{equation}
and therefore $\mathcal{L}_a\mathcal{Q}_{Q_j}=0$.  The same identity holds
after taking expectation and after empirical averaging.  Hence
\begin{equation}
    \mathcal{Q}^{c}_{Q_j}e^{t\mathbb{E}[\mathcal{L}_a]}
    =
    \mathcal{Q}^{c}_{Q_j}e^{t\mathbb{E}[\mathcal{L}_a]}
    \mathcal{Q}^{c}_{Q_j},
    \qquad
    \mathcal{Q}^{c}_{Q_j}e^{t\mathcal{L}_{j,M}}
    =
    \mathcal{Q}^{c}_{Q_j}e^{t\mathcal{L}_{j,M}}
    \mathcal{Q}^{c}_{Q_j}.
    \label{eq:first-complement-invariance}
\end{equation}
By assumption, $\mathbb{E}[a(\omega)^*a(\omega)]$ commutes with $P_j$.

The support relation $a(\omega)P_j=Q_j a(\omega)(P_j-Q_j)$ gives
$aP_j=P_jaP_j$.  For arbitrary $x$ and $B=\mathbb{E}[a^*a]$, we obtain
\begin{align}
    \mathbb{E}[\mathcal{L}_a]\mathcal{Q}_{P_j}(x)
    &=
    2\mathbb{E}[aP_jxP_ja^*]-BP_jxP_j-P_jxP_jB \notag\\
    &=
    2P_j\mathbb{E}[aP_jxP_ja^*]P_j
    -P_jBP_jxP_j-P_jxP_jBP_j \notag\\
    &=
    \mathcal{Q}_{P_j}\mathbb{E}[\mathcal{L}_a]\mathcal{Q}_{P_j}(x).
    \label{eq:Pj-invariance-generator}
\end{align}
Here the first term is in the $P_j$ corner because $aP_j=P_jaP_j$, and the
last two terms use $BP_j=P_jB$.  Hence
\begin{equation}
    \mathbb{E}[\mathcal{L}_a]\mathcal{Q}_{P_j}
    =
    \mathcal{Q}_{P_j}\mathbb{E}[\mathcal{L}_a]\mathcal{Q}_{P_j}.
\end{equation}
Thus the $P_j$ corner is invariant under
$e^{t\mathbb{E}[\mathcal{L}_a]}$:
\begin{equation}
    e^{t\mathbb{E}[\mathcal{L}_a]}\mathcal{Q}_{P_j}
    =
    \mathcal{Q}_{P_j}e^{t\mathbb{E}[\mathcal{L}_a]}\mathcal{Q}_{P_j}.
    \label{eq:first-active-invariance}
\end{equation}

By Duhamel's formula,
\begin{equation}
    e^{t\mathbb{E}[\mathcal{L}_a]}-e^{t\mathcal{L}_{j,M}}
    =
    \int_0^t e^{(t-\tau)\mathcal{L}_{j,M}}
    (\mathbb{E}[\mathcal{L}_a]-\mathcal{L}_{j,M})
    e^{\tau\mathbb{E}[\mathcal{L}_a]}\,d\tau.
\end{equation}
After multiplying by $\mathcal{Q}^{c}_{Q_j}$ on the left and
$\mathcal{Q}_{P_j}$ on the right,
\eqref{eq:first-complement-invariance} and~\eqref{eq:first-active-invariance} allow us to
insert the two corner maps next to the middle factor:
\begin{equation}
\begin{split}
    &\mathcal{Q}^{c}_{Q_j}
    (e^{t\mathbb{E}[\mathcal{L}_a]}-e^{t\mathcal{L}_{j,M}})
    \mathcal{Q}_{P_j}\\
    &\quad =
    \int_0^t
    \mathcal{Q}^{c}_{Q_j}e^{(t-\tau)\mathcal{L}_{j,M}}
    \mathcal{Q}^{c}_{Q_j}
    (\mathbb{E}[\mathcal{L}_a]-\mathcal{L}_{j,M})
    \mathcal{Q}_{P_j}
    e^{\tau\mathbb{E}[\mathcal{L}_a]}\mathcal{Q}_{P_j}\,d\tau.
\end{split}
\end{equation}
The maps $e^{s\mathbb{E}[\mathcal{L}_a]}$, $e^{s\mathcal{L}_{j,M}}$,
$\mathcal{Q}_{P_j}$ have diamond norm at most one, while
$\|\mathcal{Q}^{c}_{Q_j}\|_\diamond\leq2$.  Taking norms inside the
integral therefore gives
\begin{equation}
    \|\mathcal{Q}^{c}_{Q_j}
    (e^{t\mathbb{E}[\mathcal{L}_a]}-e^{t\mathcal{L}_{j,M}})
    \mathcal{Q}_{P_j}\|_{\diamond}
    \leq
    2t\|\mathcal{Q}^{c}_{Q_j}
    (\mathbb{E}[\mathcal{L}_a]-\mathcal{L}_{j,M})
    \mathcal{Q}_{P_j}\|_{\diamond},
\end{equation}
which is the first estimate.  Since $Q_j\leq P_{j-1}$, the complementary corner
maps satisfy
\begin{equation}
    \mathcal{Q}^{c}_{P_{j-1}}
    =
    \mathcal{Q}^{c}_{P_{j-1}}\mathcal{Q}^{c}_{Q_j}.
\end{equation}
For the second estimate, use the same Duhamel identity with
$\mathcal Q^c_{P_{j-1}}$ as its leftmost map.  Since the evolution fixes
the $Q_j$ corner and $Q_j\leq P_{j-1}$,
\[
 \mathcal Q^c_{P_{j-1}}e^{s\mathcal L_{j,M}}
  =\mathcal Q^c_{P_{j-1}}e^{s\mathcal L_{j,M}}\mathcal Q^c_{Q_j}.
\]
The left prefactor has diamond norm at most two, giving the same $2t$
bound directly.
\end{proof}

\begin{theorem}[State preparation by empirical channels]
\label{thm:first-empirical-preparation}
Let $P_0\leq\cdots\leq P_m=I$ and suppose that every stage satisfies
the support and active mass assumptions of
Propositions~\ref{prop:boxass} and~\ref{prop:gammaass}, with
$0\leq\hat\phi_j\leq1$.  Use either of the random models in
Lemma~\ref{lem:corner-support-current}, and let $d=N$ be the ambient
dimension.  Fix $0<\eps,\zeta<1$ and set
\begin{equation}
 t_j=\gamma_j^{-1}\log(12m/\eps),\qquad
 \eta_j=\frac{\eps}{8mt_j},\qquad
 \chi_j=\frac{N_j^{\beta_j'}}{N_j}.
 \label{eq:first-empirical-parameters}
\end{equation}
At stage $j$, choose an even integer
\[
 p_j\geq\max\{2,\log d,\log(m/\zeta)\}
\]
and draw $M_j$ independent matrices, where
\begin{equation}
 M_j\geq C\eta_j^{-2}\chi_j p_j^4
       =64C m^2t_j^2\eps^{-2}\chi_j p_j^4.
 \label{eq:first-preparation-matrix-count}
\end{equation}
For unitary samples, use Haar measure or an exact design of order at
least $2p_j$.  Let
\[
 \mathcal L_{j,M_j}=\frac1{M_j}\sum_{\nu=1}^{M_j}
                         \mathcal L_{a_j^{(\nu)}}.
\]
With probability at least $1-\zeta$, simultaneously for all initial
states,
\begin{equation}
 \|\mathcal Q_{P_0}^c e^{t_1\mathcal L_{1,M_1}}\cdots
    e^{t_m\mathcal L_{m,M_m}}\|_\diamond\leq\eps,
 \label{eq:first-empirical-preparation}
\end{equation}
and each output is within trace norm $\eps$ of a normalized state
supported in $P_0$.
\end{theorem}
\begin{proof}
The active mass condition implies $\gamma_j\leq1$, so
$0<\eta_j<1$.  In both random models $B_j=\mathbb E[a_j^*a_j]$ is diagonal, and
its active block equals the one in Proposition~\ref{prop:gammaass}.
For a unitary on a larger coordinate subspace of dimension $D_j$,
the normalization factor $\sqrt{D_j/N_j}$ compensates for the entry
variance $1/D_j$; all nonzero output rows of active input columns lie in
$Q_j\leq P_j$.  Thus Theorem~\ref{thm:ideal-corner} applies to the
ideal average at every stage, including these models on larger subspaces.

Apply Lemma~\ref{lem:corner-support-current} with error $\eta_j$ and
failure probability $\zeta/m$.  A union bound gives all the required
generator estimates with probability at least $1-\zeta$.
On this event, Lemma~\ref{lem:semigroup-current} and the direct
$P_{j-1}$ version of Theorem~\ref{thm:ideal-corner} give
\begin{align*}
 \|\mathcal Q_{P_{j-1}}^c e^{t_j\mathcal L_{j,M_j}}
    \mathcal Q_{P_j}\|_\diamond
 &\leq 2e^{-\gamma_jt_j}+e^{-2\gamma_jt_j}+2t_j\eta_j\\
 &\leq\frac{\eps}{4m}+\frac{\eps}{4m}
 =\frac{\eps}{2m}.
\end{align*}
Lemma~\ref{lem:corner-composition} completes the proof.
\end{proof}

\begin{rem}
The independence assumption concerns the $M_j$ copies within each
stage.  The union bound does not require independence between stages.
Drawing separate families uses $\sum_j M_j$ independent matrices;
compatible models may reuse randomness across stages with the same
stagewise marginals.  For the generators $\mathcal L_{j,M_j}$ in
Theorem~\ref{thm:first-empirical-preparation}, the total evolution time
is $\sum_jt_j$.  The resulting bounds depend on $m$,
$\gamma_j^{-1}$, $\chi_j$ and $\eps^{-1}$, which must be estimated
for the Hamiltonian under consideration.  The empirical dynamics need not preserve
$P_j$ exactly: the comparison uses invariance of the ideal corner,
and Lemma~\ref{lem:corner-composition} accounts for the propagated error.
\end{rem}

\subsection{Choice of spectral cutoffs}

Let
\begin{equation}
    N_{\mathcal{H}}(E):=\#\{l:\lambda_l\leq E\}
\end{equation}
be the cumulative eigenvalue counting function.  For an energy cutoff $E$, the
corresponding spectral projection is
\begin{equation}
    P(E):=\sum_{\lambda_l\leq E}E_{ll}.
\end{equation}
At stage $j$, the hypotheses of Propositions~\ref{prop:boxass}
and~\ref{prop:gammaass} can be expressed through four indices
$r_j\leq q_j\leq N_{j-1}<N_j$: the active corner $P_j$, the target
\notationrev{$P_{j-1}$}, the output projection $Q_j$ of rank $q_j$, and a lower index
$r_j$ that measures the mass.  They hold if
\begin{equation}
\begin{gathered}
    \max\{\lambda_{q_j},\ \lambda_{N_j}-\lambda_{q_j}\}\leq\alpha_j'
      <\alpha_j<\lambda_{N_{j-1}}-\lambda_{r_j},
    \qquad \lambda_{N_j}\leq\beta_j<\beta_j',\\
    \frac{N_j^{\beta_j'}}{r_j}\leq\gamma^{-1}.
\end{gathered}
    \label{eq:stage-condition-iv}
\end{equation}
The first condition removes the input columns of $Q_j$ and the output rows
outside $Q_j$.  It also gives
$r_{\alpha_j}(s)=N_{\mathcal H}(\lambda_s-\alpha_j)\geq r_j$ for every
active index $s>N_{j-1}$, so Proposition~\ref{prop:gammaass} holds with
$\gamma_j=r_j/N_j$.  The last condition bounds both $\gamma_j^{-1}$ and the
enlargement ratio $\chi_j=N_j^{\beta_j'}/N_j$, with $N_j^{\beta_j'}$ as in
\eqref{eq:fm-enlarged-support}, by $\gamma^{-1}$.  A polynomial bound on
$\gamma^{-1}$ therefore controls both quantities in
Theorem~\ref{thm:first-empirical-preparation}.  The next theorem gives a
sufficient condition in terms of a sequence of eigenvalues.

\begin{theorem}[Regular values give admissible stages]
\label{thm:regular-values-current}
Let $x_0<x_1<\cdots <x_m$ be spectral indices whose eigenvalues are
strictly increasing, and set
\begin{equation}
    \Lambda_l:=\lambda_{x_l},\qquad
    0\leq\Lambda_0<\Lambda_1<\cdots<\Lambda_m.
\end{equation}
Assume that there are fixed integers $1\leq m_0<k_0$ such that, for every
$j>k_0$,
\begin{equation}
    \Theta_j:=
    \max\{\Lambda_{j-m_0},\Lambda_j-\Lambda_{j-m_0}\}
    <
    D_j:=\Lambda_{j-1}-\Lambda_{j-k_0}.
    \label{eq:regular-values-current}
\end{equation}
Choose $\delta_j>0$ so that
\begin{equation}
    \Theta_j+2\delta_j<D_j.
\end{equation}
At stage $j$, define
\begin{align}
    N_j&=N_{\mathcal{H}}(\Lambda_j),&
    \notationrev{N_{j-1}}&=N_{\mathcal{H}}(\Lambda_{j-1}),\\
    \notationrev{q_j}&=N_{\mathcal{H}}(\Lambda_{j-m_0}),&
    r_j&=N_{\mathcal{H}}(\Lambda_{j-k_0}).
\end{align}
Set
\begin{equation}
    \alpha_j'=\Theta_j+\delta_j,\qquad
    \alpha_j=\Theta_j+2\delta_j,\qquad
    \beta_j=\Lambda_j,\qquad
    \beta_j'=\Lambda_j+\delta_j.
    \label{eq:regular-stage-parameters}
\end{equation}
Then the hypotheses of Proposition~\ref{prop:boxass} hold at stage $j$,
and those of Proposition~\ref{prop:gammaass} hold
with
\begin{equation}
    \gamma_j=
    \frac{N_{\mathcal{H}}(\Lambda_{j-k_0})}
    {N_{\mathcal{H}}(\Lambda_j)}.
    \label{eq:regular-gamma}
\end{equation}
\end{theorem}
\begin{proof}
We first check the support conditions.  Since
$q_j=N_{\mathcal{H}}(\Lambda_{j-m_0})$, we have
\begin{equation}
    \lambda_{q_j}=\Lambda_{j-m_0}.
\end{equation}
Hence
\begin{equation}
    \lambda_{\notationrev{q_j}}-\lambda_1
    \leq \Lambda_{j-m_0}
    \leq \Theta_j<\alpha_j',
\end{equation}
and
\begin{equation}
    \lambda_{N_j}-\lambda_{q_j}
    \leq \Lambda_j-\Lambda_{j-m_0}
    \leq \Theta_j<\alpha_j'.
\end{equation}
Finally $m_0\geq 1$ implies \notationrev{$q_j\leq N_{j-1}$}, and therefore
\notationrev{$Q_j\leq P_{j-1}\leq P_j$}.

It remains to check the hypotheses of Proposition~\ref{prop:gammaass}.
Since $\alpha_j=\Theta_j+2\delta_j<D_j<\Lambda_j=\lambda_{N_j}$ and
$\beta_j=\Lambda_j$,
\begin{equation}
    \alpha_j\leq \lambda_{N_j}\leq \beta_j.
\end{equation}
For the mass condition, if \notationrev{$s>N_{j-1}$}, then
\begin{equation}
    \lambda_s\geq \Lambda_{j-1}.
\end{equation}
Since $\alpha_j<D_j=\Lambda_{j-1}-\Lambda_{j-k_0}$,
\begin{equation}
    \lambda_s-\alpha_j
    >
    \Lambda_{j-k_0}.
\end{equation}
Therefore
\begin{equation}
    r_{\alpha_j}(s)
    =
    N_{\mathcal{H}}(\lambda_s-\alpha_j)
    \geq
    N_{\mathcal{H}}(\Lambda_{j-k_0})
    =
    \gamma_jN_j,
\end{equation}
with $\gamma_j$ as in \eqref{eq:regular-gamma}, proving the mass
condition in Proposition~\ref{prop:gammaass}.
\end{proof}

Geometrically spaced eigenvalues with ratio sufficiently close to
one satisfy this criterion; see Section~\ref{subsec:ex-geometric}.

\section{Construction of Jump operators for Clustered spectrum}
\label{sec:second-method}

In this section every eigenvalue of $\mathcal H$ lies in one of the disjoint
intervals $[m_j-\delta,m_j+\delta]$, $0\leq j\leq J$, with centers
$m_0<\cdots<m_J$.  Let $\Pi_j$ be the spectral projection onto the $j$-th
cluster, $P_j=\Pi_0+\cdots+\Pi_j$ and $N_j=\dim(P_j\mathcal K)$.  As for the
first method, stage $j$ drives states supported in $P_j$ into $P_{j-1}$, and
the stages are applied in the order $j=J,\ldots,1$.  The jump operators are
the frequency components
\[
    a_w(g)=\sum_{\substack{0\leq i<j\leq J\\ m_j-m_i=w}}\Pi_ig\Pi_j,
    \qquad w\in W=\{m_j-m_i:0\leq i<j\leq J\},
\]
of a Gaussian matrix $g$, extracted from $e^{it\mathcal H}ge^{-it\mathcal H}$
by a filter around each $w$; this requires the elements of $W\cup\{0\}$ to be
more than $6\delta$ apart.  Each component is a jump operator of its own.
For every sample, $a_w(g)$ maps $P_j$ into $P_{j-1}$, and
$a_w(g)^*a_w(g)$ commutes with every $\Pi_j$.  The second identity fails for
the combined operator $\sum_wa_w(g)$, whose square contains cross terms such
as $\Pi_1g^*\Pi_0g\Pi_2$ for three equally spaced clusters.  Randomness is
needed only to bound the loss on the top cluster from below.  For equally
spaced centers $m_j=m_0+j\Delta$, $W=\{\Delta,\ldots,J\Delta\}$ and
$a_{k\Delta}(g)$ lowers every cluster index by exactly $k$.

In contrast with the first method, one family of $M$ Gaussian matrices and
one set of filters serve all stages, which differ only by the normalization
$1/N_j$.  The price is that all center differences must be resolved, and
$|W|$ can be as large as $J(J+1)/2$.  Section~\ref{subsec:sm-clusters}
specifies the cluster and separation assumptions and constructs the
corresponding filters.  Sections~\ref{subsec:sm-algebra}
and~\ref{subsec:sm-generator} establish the block identities and define
the empirical generator, respectively.  Section~\ref{subsec:sm-sampling}
provides sufficient bounds on the number of random matrices, with the
Gaussian estimate proved in Appendix~\ref{app:sm-gaussian}.
Section~\ref{subsec:sm-convergence} combines the estimates for individual
stages to establish convergence of the preparation procedure.  Finally,
Section~\ref{subsec:sm-stage-dependent} extends the analysis to generators
that vary between stages.

\subsection{Clusters and filtered jump operators}
\label{subsec:sm-clusters}

Let $J\geq1$, and let $\Pi_0,\ldots,\Pi_J$ be nonzero, mutually orthogonal
spectral projections of $\mathcal{H}$ whose sum is $I$.  Thus there are
$J+1$ clusters and $J$ reduction stages.  Suppose the eigenvalues in
$\Pi_j\mathcal{K}$ lie in the interval
\begin{equation}
    I_j=[m_j-\delta,m_j+\delta],\qquad
    m_0<m_1<\cdots<m_J,\qquad \delta>0.
    \label{eq:sm-clusters}
\end{equation}
Here $m_j$ is a cluster center and $\delta$ is a common bound on
the cluster half-widths.  A cluster may contain a number of distinct energies.
Assumption~\ref{as:specluster-separation} will
ensure that these intervals are disjoint and ordered by energy.

Define the cluster dimensions and cumulative projections by
\begin{equation}
\begin{split}
    d_j&=\dim(\Pi_j\mathcal K),\qquad
    P_j=\sum_{i=0}^j\Pi_i,\\
    N_j&=\dim(P_j\mathcal K)=\sum_{i=0}^jd_i,
    \qquad P_J=I,\quad N_J=N.
\end{split}
    \label{eq:sm-counts}
\end{equation}
Here $d_j$ is the cluster dimension and $N_j$ is the cumulative dimension
through cluster $j$; both count eigenstates with
multiplicity.  Stage $j$ reduces the active subspace $P_j\mathcal{K}$ to
$P_{j-1}\mathcal{K}$.  The stages are applied in the order $J,J-1,\ldots,1$.

\begin{assumption}[Separation of Spectrum]
\label{as:specluster-separation}
Let $W$ be the set of distinct positive differences of cluster centers,
and let $W_0$ be $W$ together with zero:
\[
    W=\{m_j-m_i:0\leq i<j\leq J\},\qquad W_0=W\cup\{0\}.
\]
We assume that distinct elements of $W_0$ are more than $6\delta$ apart:
\begin{equation}
    \min_{\substack{w,w'\in W_0\\w\ne w'}}|w-w'|>6\delta.
    \label{eq:sm-difference-separation}
\end{equation}
\end{assumption}

A value occurring for several pairs of clusters is included in $W$ only
once.  Including $0$ ensures that the positive frequencies are also
separated from intracluster transitions.

For each $w\in W$, choose a smooth, compactly supported frequency filter
$\widehat{\phi}_w$ satisfying
\begin{equation}
\begin{split}
    &0\leq\widehat{\phi}_w\leq1,\qquad
      \widehat{\phi}_w(\xi)=1\quad\text{for }|\xi-w|\leq2\delta,\\
    &\supp\widehat{\phi}_w\subset(w-3\delta,w+3\delta).
\end{split}
    \label{eq:sm-filter}
\end{equation}
The explicit construction of the filter function can be found in~\cite[Chapter~5]{SteinShakarchi2003}. The plateau of half-width $2\delta$ passes all energy differences between
two clusters whose centers differ by $w$; the additional width $\delta$
on each side allows a smooth transition to zero.  Define the corresponding
function of time by the Fourier convention 
\begin{equation}
    \phi_w(t)=\frac{1}{2\pi}\int_{\mathbb R}
                  \widehat{\phi}_w(\xi)e^{it\xi}\,d\xi,
    \qquad
    \widehat{\phi}_w(\xi)=\int_{\mathbb R}\phi_w(t)e^{-it\xi}\,dt.
    \label{eq:sm-fourier-convention}
\end{equation}
The smooth compact support makes $\phi_w$ rapidly decaying, so the following
integral is well defined.  For any matrix $g$ on $\mathcal{K}$, set
\begin{equation}
\begin{split}
    a_w(g)&=\int_{\mathbb R}\phi_w(t)
                  e^{it\mathcal{H}}g e^{-it\mathcal{H}}\,dt\\
          &=\sum_{r,s=1}^N
             \widehat{\phi}_w(\lambda_s-\lambda_r)E_{rr}gE_{ss}.
\end{split}
    \label{eq:sm-filtered-jump}
\end{equation}
Indeed, the $(r,s)$ matrix entry of the conjugated operator acquires the
phase $e^{-it(\lambda_s-\lambda_r)}$ and this gives rise to the fourier transform $\hat{\phi}_w$.  Thus, a positive filter frequency selects a transition from a higher-energy state $s$ to a lower-energy state $r$.

We now to illustrate that based on the Assumption~\ref{as:specluster-separation}, the jump operator we introduced in \ref{eq:sm-filtered-jump} can be represented by the projection operators corresponding to each spectrum cluster.
If $\lambda_s\in I_j$ and $\lambda_r\in I_i$, then
$|(\lambda_s-\lambda_r)-(m_j-m_i)|\leq2\delta$.
When $m_j-m_i=w$, the filter is therefore exactly one.  When the positive
center difference is a different value $w'$, its entire interval of width
$4\delta$ lies outside the support of $\widehat{\phi}_w$, by
\eqref{eq:sm-difference-separation}.  Intracluster differences lie in
$[-2\delta,2\delta]$ and are also excluded; transitions that increase energy
have negative frequency and are excluded as well.  Consequently,
\begin{equation}
    a_w(g)=\sum_{\substack{0\leq i<j\leq J\\m_j-m_i=w}}\Pi_i g \Pi_j.
    \label{eq:sm-block-jump}
\end{equation}
Several pairs of clusters may have the same center separation $w$.
The individual eigenvalue differences within these pairs can vary, but
all lie in $[w-2\delta,w+2\delta]$, where $\widehat{\phi}_w=1$.
Thus each corresponding block $\Pi_i g\Pi_j$ is included exactly.
Every stage uses all frequencies
in $W$, including transitions between nonadjacent clusters.  The filters
and unscaled jump operators $a_w(g)$ are fixed throughout the reduction.

\begin{figure}[H]
\centering
\begin{tikzpicture}[x=0.48cm,y=1.1cm,>=Stealth,font=\small]
  \fill[blue!12] (-2,0) rectangle (2,1);
  \draw[->] (-3.8,0)--(10.1,0) node[right] {$\xi$};
  \draw[thick,blue!70!black] (-3,0)--(-2.8,0)
    .. controls (-2.35,0) and (-2.45,1) .. (-2,1)
    --(2,1) .. controls (2.45,1) and (2.35,0) .. (2.8,0)
    --(3,0);
  \node[above] at (0,1) {$\widehat\phi_w=1$};
  \foreach \x in {-3,-2,0,2,3,4.5,6.5,8.5}
    \draw (\x,0.04)--(\x,-0.04);
  \node[below] at (-3,-0.08) {$w-3\delta$};
  \node[below] at (0,-0.08) {$w$};
  \node[below] at (3,-0.08) {$w+3\delta$};
  \draw[very thick,blue!70!black] (-2,0.15)--(2,0.15);
  \draw[very thick,orange!80!black] (4.5,0.15)--(8.5,0.15);
  \node[above,align=center] at (6.5,0.25)
    {unwanted differences\\$[w'-2\delta,w'+2\delta]$};
  \node[below] at (6.5,-0.08) {$w'$};
  \draw[<->] (0,-0.9)--(6.5,-0.9)
    node[midway,below] {$|w'-w|>6\delta$};
\end{tikzpicture}
\caption{The selected Bohr differences lie in $[w-2\delta,w+2\delta]$
(blue).  The filter vanishes outside $(w-3\delta,w+3\delta)$, before
the next possible difference window (orange).  The same exclusion is
required for $w'=0$, which represents intracluster transitions.}
\label{fig:sm-frequency-filter}
\end{figure}
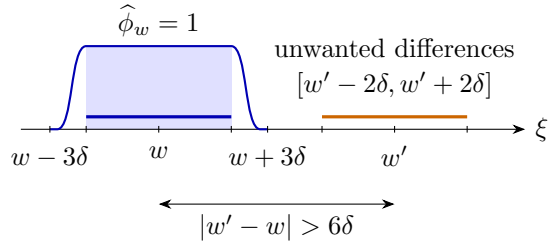

\subsection{Algebraic reduction}
\label{subsec:sm-algebra}

For any operator $\rho$ on $\mathcal{K}$, recall that
\begin{equation}
\begin{split}
    \Psi_a(\rho)&=a\rho a^*,\qquad
    \mathcal{M}_a(\rho)=-a^*a\rho-\rho a^*a,\\
    \mathcal{L}_a&=2\Psi_a+\mathcal{M}_a.
\end{split}
    \label{eq:sm-lindblad-convention}
\end{equation}
For positive $\rho$, $2\Psi_a(\rho)$ is positive.
The anticommutator term $\mathcal M_a(\rho)$ has the opposite trace, so
$\operatorname{Tr}\mathcal L_a(\rho)=0$.  The empirical generator
$\mathcal B_M$ in \eqref{eq:sm-empirical-generator} is a sum of these dissipators.

\begin{lemma}
\label{lem:sm-block-shifts}
Recall from \eqref{eq:sm-counts} that $P_j=\sum_{i=0}^j\Pi_i$.
Fix $w\in W$ and a matrix $g$, we use $a_w$ to denote $a_w(g)$.  Let
$D_w\subseteq\{1,\ldots,J\}$ be the set of indices $j$ for which
$m_j-m_i=w$ for some $i<j$.  Denote this unique target index by $f_w(j)$.
Then $f_w:D_w\to\{0,\ldots,J-1\}$ is injective, and
\begin{equation}
    a_w\Pi_j=\Pi_{f_w(j)}a_w\quad(j\in D_w),\qquad
    a_w\Pi_j=0\quad(j\notin D_w).
    \label{eq:sm-block-shift}
\end{equation}
Moreover,
\begin{equation}
\begin{split}
    &[a_w^*a_w,\Pi_j]=0\quad(0\leq j\leq J),\\
    &a_wP_j=P_{j-1}a_wP_j\quad(1\leq j\leq J),\qquad a_w\Pi_0=0,
\end{split}
    \label{eq:sm-loss-blocks}
\end{equation}
where $[A,B]=AB-BA$ is the commutator.  The corresponding superoperator
identities, for $1\leq j\leq J$, are
\begin{equation}
    \mathcal M_{a_w}\mathcal Q_{P_j}
       =\mathcal Q_{P_j}\mathcal M_{a_w},
    \label{eq:sm-loss-commutation}
\end{equation}
and
\begin{equation}
\begin{split}
    \Psi_{a_w}\mathcal Q_{P_j}
       &=\mathcal Q_{P_{j-1}}\Psi_{a_w}\mathcal Q_{P_j},\\
    \mathcal Q_{P_{j-1}}^c\Psi_{a_w}\mathcal Q_{P_j}&=0,\\
    \mathcal Q_{P_{j-1}}^c\Psi_{a_w}
       &=\mathcal Q_{P_{j-1}}^c\Psi_{a_w}\mathcal Q_{P_j}^c.
\end{split}
    \label{eq:sm-gain-lowering}
\end{equation}
In particular, each cumulative corner is invariant under the dissipator:
\begin{equation}
\begin{split}
    \mathcal{L}_{a_w}\mathcal{Q}_{P_j}
       &=\mathcal{Q}_{P_j}\mathcal{L}_{a_w}\mathcal{Q}_{P_j},\\
    \mathcal{Q}_{P_j}^c\mathcal{L}_{a_w}
       &=\mathcal{Q}_{P_j}^c\mathcal{L}_{a_w}\mathcal{Q}_{P_j}^c.
\end{split}
    \label{eq:sm-invariant-corners}
\end{equation}
\end{lemma}
\begin{proof}
For a fixed source center $j$, the target center must equal $m_j-w$.  Since the
centers are distinct, there is at most one such target.  Conversely, if
$f_w(j)=f_w(k)$, then $m_j-w=m_k-w$, hence $j=k$.  This proves injectivity
and, using \eqref{eq:sm-block-jump}, the identities
\eqref{eq:sm-block-shift}.

In the product $a_w^*a_w$, a cross term from sources $j$ and $k$ contains
$\Pi_{f_w(j)}\Pi_{f_w(k)}$.  It vanishes unless $j=k$, because the target
projections are orthogonal and $f_w$ is injective.  Therefore
\begin{equation}
    a_w^*a_w=
       \sum_{j\in D_w}\Pi_jg^*\Pi_{f_w(j)}g\Pi_j.
    \label{eq:sm-loss-formula}
\end{equation}
Each summand is supported in one source cluster, so $a_w^*a_w$ commutes
with $\Pi_j$.  Summing gives
\begin{equation}
    [a_w^*a_w,P_j]=\sum_{i=0}^j[a_w^*a_w,\Pi_i]=0.
    \label{eq:sm-cumulative-commutation}
\end{equation}
Both left and right multiplication by $a_w^*a_w$ therefore commute with
$\mathcal Q_{P_j}$, proving \eqref{eq:sm-loss-commutation}.
Every target index is smaller than its source index because of Equation~\ref{eq:sm-block-jump}.  An input
in $P_j$ therefore has its image under $a_w$ in $P_{j-1}$, and an input in
$\Pi_0$ is annihilated.  This proves \eqref{eq:sm-loss-blocks}.

For every $\rho$, the identity $a_wP_j=P_{j-1}a_wP_j$ gives
\[
    a_wP_j\rho P_ja_w^*
       =P_{j-1}\bigl(a_wP_j\rho P_ja_w^*\bigr)P_{j-1}.
\]
This proves the first two statements of \eqref{eq:sm-gain-lowering};
the last follows by splitting the input into
$\mathcal Q_{P_j}+\mathcal Q_{P_j}^c$.
If $\rho=P_j\rho P_j$, then $\Psi_{a_w}(\rho)$ is supported in $P_{j-1}\leq P_j$,
while $\mathcal{M}_{a_w}(\rho)$ remains supported in $P_j$ because
$a_w^*a_w$ commutes with $P_j$.
Thus $\mathcal{L}_{a_w}(\rho)$ is supported in $P_j$, which is the first
identity in \eqref{eq:sm-invariant-corners}.  The second follows from
$\mathcal{Q}_{P_j}^c\mathcal{L}_{a_w}\mathcal{Q}_{P_j}=0$ and
\eqref{eq:sm-corner-maps}.  Here $P_{j-1}\leq P_j$ means inclusion of the
corresponding ranges.
\end{proof}

\begin{rem}
\label{rem:sm-strict-lowering}

The lemma separates two properties: the \emph{operator $a_w^*a_w$ commutes} with
every cluster projection, whereas the \emph{dissipator preserves} every
cumulative corner.  The latter is invariance, not full commutation:
 \[  \mathcal{L}_{a_w}\mathcal{Q}_{P_j} \stackrel{?}{=} 
 \mathcal{Q}_{P_j}\mathcal{L}_{a_w} \] 
need not hold.  Indeed, for
$\rho=(I-P_j)\rho(I-P_j)$ we have $\mathcal{Q}_{P_j}(\rho)=0$, hence
$\mathcal{L}_{a_w}\mathcal{Q}_{P_j}(\rho)=0$, whereas
$\mathcal{Q}_{P_j}\mathcal{L}_{a_w}(\rho)=2P_ja_w\rho a_w^*P_j$ because
$a_w^*a_w$ commutes with $P_j$, This part does not have to vanish because $a_w$ can map $(I-P_j)K$ into $P_jK$. However, the multiplication part is fully covered by the ideal Lindbladian and hence by the empirical mean. 
\end{rem}



Accordingly, the empirical generator $\mathcal B_M$ in
\eqref{eq:sm-empirical-generator} sums the frequency-resolved
dissipators, preserving the cumulative corners for every realization of $g$.

\subsection{The empirical generator and its normalization}
\label{subsec:sm-generator}

Let $M$ independent random matrices $g^{(1)},\ldots,g^{(M)}$ have
independent standard complex Gaussian entries in the chosen eigenbasis.
Each entry is distributed as $(u+iv)/\sqrt{2}$, where $u$ and $v$ are
independent standard real Gaussian variables.  In particular,
$\mathbb E g^{(\nu)}_{rs}=0$ and
$\mathbb E|g^{(\nu)}_{rs}|^2=1$.  Here $M$ is the number of random matrices,
and $N$ and $N_j$ are dimensions of Hilbert spaces.

Define the empirical generator, the operator $K_M$, and the stage generator
by
\begin{equation}
\begin{split}
    \mathcal{B}_M&=\frac1M\sum_{\nu=1}^M\sum_{w\in W}
                    \mathcal{L}_{a_w(g^{(\nu)})},\\
    K_M&=\frac1M\sum_{\nu=1}^M\sum_{w\in W}
                    a_w(g^{(\nu)})^*a_w(g^{(\nu)}),\\
    \mathcal{L}_{j,M}&=\frac{1}{N_j}\mathcal{B}_M
                         \qquad(1\leq j\leq J).
\end{split}
    \label{eq:sm-empirical-generator}
\end{equation}
For every realization, $K_M$ is positive and commutes with every $\Pi_j$.
The factor $1/M$ averages over samples.

The second factor, $1/N_j$, normalizes the generator by the dimension
$N_j=\dim(P_j\mathcal{K})$ of the subspace that carries the state at
stage $j$.  Since $\mathcal{L}_{ca}=|c|^2\mathcal{L}_a$, it is
equivalent to scaling each jump operator by $N_j^{-1/2}$:
\begin{equation}
    \mathcal{L}_{j,M}
      =\frac1M\sum_{\nu=1}^M\sum_{w\in W}
          \mathcal{L}_{a_w(g^{(\nu)})/\sqrt{N_j}}.
    \label{eq:sm-normalized-jumps}
\end{equation}
On $P_j\mathcal{K}$, the jump operators $a_w(g^{(\nu)})/\sqrt{N_j}$ have
norm of order one, and \eqref{eq:sm-expected-loss} gives
$\mathbb{E}(P_jK_MP_j)/N_j\leq P_j$.  Only the factor $1/N_j$ depends on
$j$; the samples, filters, and normalization of the density matrix remain
fixed.  The generator acts on
the full Hilbert space, and at stage $j$ we estimate its action on inputs
supported in $P_j$.

At stage $j$, the loss from the highest cluster in $P_j\mathcal{K}$ is
determined by $\Pi_jK_M\Pi_j$.  To compute this block, fix $j\geq1$.
For each lower target cluster $i<j$, the transition from
cluster $j$ to cluster $i$ belongs to the frequency
$w=m_j-m_i\in W$.  These frequencies are distinct as $i$ varies, even
though the same frequency can occur for transitions from other source
clusters.  Thus \eqref{eq:sm-loss-formula} gives, for any matrix $g$,
\[
    \sum_{w\in W}\Pi_j a_w(g)^*a_w(g)\Pi_j
      =\sum_{i=0}^{j-1}\Pi_jg^*\Pi_i g\Pi_j
      =\Pi_jg^*P_{j-1}g\Pi_j.
\]
Averaging this identity over the $M$ samples gives
\begin{equation}
\begin{split}
    \Pi_jK_M\Pi_j&=\frac1M\sum_{\nu=1}^M
                   \Pi_j(g^{(\nu)})^*P_{j-1}g^{(\nu)}\Pi_j,\\
    \mathbb E(\Pi_jK_M\Pi_j)&=N_{j-1}\Pi_j\qquad(j\geq1),
    \qquad K_M\Pi_0=0.
\end{split}
    \label{eq:sm-expected-loss}
\end{equation}
To verify the expectation, let $\psi_s,\psi_t$ be basis vectors in
$\Pi_j\mathcal{K}$, and write $g_{rs}=\langle\psi_r|g|\psi_s\rangle$.
Independence of the entries and $\mathbb E|g_{rs}|^2=1$ give
$\mathbb E(\overline{g_{rs}}g_{rt})=\delta_{st}$, where $\delta_{st}$ is the
Kronecker delta.  Hence
\[
    \langle\psi_s|\,\mathbb E(\Pi_jK_M\Pi_j)\,|\psi_t\rangle
      =\frac1M\sum_{\nu=1}^M\ \sum_{r:\,\psi_r\in P_{j-1}\mathcal{K}}
         \mathbb E\bigl(\overline{g^{(\nu)}_{rs}}\,g^{(\nu)}_{rt}\bigr)
      =N_{j-1}\,\delta_{st},
\]
because the inner sum has $N_{j-1}$ terms.  Since $K_M$ is not normalized
and the entries of $g^{(\nu)}$ have unit variance, each of the $N_{j-1}$
basis vectors of $P_{j-1}\mathcal{K}$ contributes one to this mean.  The
normalization enters only through $\mathcal{L}_{j,M}=\mathcal{B}_M/N_j$, for
which the corresponding mean is $\gamma_j\Pi_j$, with
$\gamma_j=N_{j-1}/N_j\leq1$ as defined in \eqref{eq:sm-gamma}.

Accordingly, define
\begin{equation}
    \gamma_j=\frac{N_{j-1}}{N_j},\qquad
    \gamma=\min_{1\leq j\leq J}\gamma_j>0.
    \label{eq:sm-gamma}
\end{equation}
For the \emph{normalized} generator, the block $\mathbb{E}(\Pi_jK_M\Pi_j)/N_j$ is
$\gamma_j\Pi_j$.  Thus $\gamma_j$ is the fraction of the dimension of $P_j\mathcal{K}$
available as targets of lower energy.  Under the ideal generator
$\mathbb E\mathcal{L}_{j,M}$, as in Section~\ref{subsec:fm-ideal}, an input $\rho$ supported in $P_j$ evolves
to $\rho_t$ with $\operatorname{Tr}(\Pi_j\rho_t)
=e^{-2\gamma_jt}\operatorname{Tr}(\Pi_j\rho)$; the factor two comes from
\eqref{eq:sm-lindblad-convention}.

For comparison, if $A=\sum_{w\in W}a_w(g)$ for a single Gaussian matrix,
then Gaussian covariance gives
\begin{equation}
    \mathbb E(A^*A)=\sum_{w\in W}\mathbb E(a_w^*a_w).
    \label{eq:sm-expected-sum}
\end{equation}
The cross terms vanish in expectation, but need not vanish for a fixed
sample, as the example with three clusters shows.

\subsection{Number of random matrices}
\label{subsec:sm-sampling}

It suffices to prove $\Pi_jK_M\Pi_j\geq cN_{j-1}\Pi_j$ for a constant
$c>0$ independent of $j$ and $N$; Proposition~\ref{prop:sm-sampling} gives such a bound
for Gaussian samples.  On this event, the preparation error decreases with
the evolution time, with no further sampling.  Let $\zeta$ be the allowed
failure probability.

\begin{prop}[Simultaneous bounds on $\Pi_jK_M\Pi_j$]
\label{prop:sm-sampling}
Let $0<\zeta<1$, and let the integer $M$ satisfy
\begin{equation}
    M\geq128\left(2\gamma^{-1}\log9+
                         \log\frac{2J}{\zeta}\right).
    \label{eq:sm-sample-count}
\end{equation}
Then, with probability at least $1-\zeta$, simultaneously for every
$1\leq j\leq J$,
\begin{equation}
    \frac{\gamma}{2}\Pi_j
          \leq\frac{\gamma_j}{2}\Pi_j
          \leq\frac{1}{N_j}\Pi_jK_M\Pi_j
          \leq\frac{3\gamma_j}{2}\Pi_j.
    \label{eq:sm-empirical-lower}
\end{equation}
Here $A\leq B$ means that $B-A$ is positive semidefinite.  The bounds hold
uniformly over vectors in each cluster.
\end{prop}
In Appendix~\ref{app:sm-gaussian}, the proof stacks the $M$ Gaussian
blocks $P_{j-1}g^{(\nu)}\Pi_j$ into a matrix $X_j$.  Then
$\Pi_jK_M\Pi_j=X_j^*X_j/M$, where $X_j^*X_j$ is the Gram matrix of the
columns of $X_j$: its entries are their pairwise inner products.  A union
bound gives the result for all clusters.  Theorem~\ref{thm:sm-preparation}
uses only the lower bound; the upper bound shows that each block
$\Pi_jK_M\Pi_j$ is comparable to its mean $N_{j-1}\Pi_j$.

Consequently, $M=O(\gamma^{-1}+\log(J/\zeta))$ independent dense Gaussian matrices are sufficient, regardless of the final preparation error $\epsilon$. This count refers to matrices; the cost of generating and implementing them is a separate consideration. Polynomial bounds in the system size necessitate polynomial bounds on $\gamma^{-1}$ and on the number of clusters and frequencies.



 

The proof of Proposition~\ref{prop:sm-sampling} uses the independence of
the Gaussian entries.  The expectation identities \eqref{eq:sm-expected-loss}
and \eqref{eq:sm-expected-sum} also hold for $g=\sqrt N\,u$ when $u$ is
sampled uniformly from the unitary group, according to Haar measure.
The lower bound for these samples needs a different estimate.  An exact
unitary $t$-design is a distribution on unitaries whose averages agree
with Haar averages for every product of $k$ matrix entries and $k$
conjugate entries, for all $k\leq t$.

\begin{cor}[An alternative estimate for Haar matrices and exact designs]
\label{cor:sm-design-sampling}
Keep the same clusters, filters, and normalization, but draw independent
matrices $g^{(\nu)}=\sqrt N\,u^{(\nu)}$, where the $u^{(\nu)}$ are Haar
unitaries on $\mathcal{K}$.  Choose an even integer
\[
 p\geq\max\{2,\log N,\log(J/\zeta)\}.
\]
There is a universal constant $C$ such that
\begin{equation}
 M\geq C\gamma^{-2}p^4
 \label{eq:sm-design-count}
\end{equation}
implies the simultaneous bounds \eqref{eq:sm-empirical-lower} with
probability at least $1-\zeta$.  The same assertion holds for independent
samples from an exact unitary $t$-design with $t\geq2p$.
\end{cor}
\begin{proof}
Apply Proposition~\ref{prop:decomposed-gaussian-loss} and
Corollary~\ref{cor:decomposed-haar-loss} in Appendix~\ref{app:probability}
to the cumulative cutoffs $q=P_j$, with tolerance $\eta=\gamma/2$.
The mean block $\mathbb{E}(\Pi_jK_M\Pi_j)/N_j$ is $\gamma_j\Pi_j$ by covariance.
The empirical block is therefore within $\gamma/2$ of $\gamma_j\Pi_j$ in
operator norm.  Since $\gamma\leq\gamma_j$, it lies between
$(\gamma_j-\gamma/2)\Pi_j\geq(\gamma_j/2)\Pi_j$ and
$(\gamma_j+\gamma/2)\Pi_j\leq(3\gamma_j/2)\Pi_j$, simultaneously for
all $j$.
\end{proof}

Both sampling results use one jump operator $a_w(g^{(\nu)})$ for each
sample $\nu$ and frequency $w\in W$, giving $M|W|$ jump operators.  For
Haar or exact-design samples, the sufficient matrix count
\eqref{eq:sm-design-count} scales as $\gamma^{-2}p^4$; for Gaussian
samples, \eqref{eq:sm-sample-count} scales as
$\gamma^{-1}+\log(J/\zeta)$.  The first is a weaker sufficient bound,
not a claim that unitary samples need more matrices.  The proof for exact
designs uses averages of polynomials with $p$ entries of $u$ and $p$
conjugate entries.  An approximate design matches these Haar averages
only up to an error, which would need to be bounded to obtain a
corresponding sampling result.

\subsection{State preparation for lowest cluster}
\label{subsec:sm-convergence}

The compression $\mathcal{Q}_P$ is completely positive
and trace-nonincreasing, so its diamond norm is at most one, whereas we
use only $\|\mathcal{Q}_P^c\|_\diamond\leq2$.  The complement map is not
in general completely positive or contractive.

\begin{lemma}[Reduction by one cluster]
\label{lem:sm-one-step}
Fix $1\leq j\leq J$ and a sample family.  Suppose
$\Pi_jK_M\Pi_j/N_j\geq\kappa_j\Pi_j$ for some $\kappa_j>0$.  Then, for every
$t\geq0$,
\begin{equation}
    \left\|\mathcal{Q}_{P_{j-1}}^c
       e^{t\mathcal{L}_{j,M}}\mathcal{Q}_{P_j}\right\|_\diamond
          \leq2e^{-\kappa_jt}.
    \label{eq:sm-one-step}
\end{equation}
For any density matrix $\rho=P_j\rho P_j$, the diagonal block
$\Pi_je^{t\mathcal{L}_{j,M}}(\rho)\Pi_j$ satisfies the sharper estimate
\begin{equation}
    \operatorname{Tr}\left(\Pi_j e^{t\mathcal{L}_{j,M}}(\rho)\right)
       \leq e^{-2\kappa_jt}\operatorname{Tr}(\Pi_j\rho).
    \label{eq:sm-population}
\end{equation}
\end{lemma}
\begin{proof}
Write $\mathcal{L}_{j,M}=\mathcal{D}_j+\mathcal{G}_j$, where
\begin{equation}
\begin{split}
    \mathcal{D}_j&=\frac{1}{MN_j}\sum_{\nu=1}^M\sum_{w\in W}
                           \mathcal M_{a_w(g^{(\nu)})},\\
    \mathcal{G}_j&=\frac{2}{MN_j}\sum_{\nu=1}^M\sum_{w\in W}
                           \Psi_{a_w(g^{(\nu)})},\qquad
    \mathcal D_j(x)=-\frac{K_Mx+xK_M}{N_j}.
\end{split}
    \label{eq:sm-gain-loss-split}
\end{equation}
The evolution generated by $\mathcal D_j$ is
$e^{t\mathcal{D}_j}(x)=e^{-tK_M/N_j}x e^{-tK_M/N_j}$.
Duhamel's formula gives
\begin{equation}
    e^{t\mathcal{L}_{j,M}}
       =e^{t\mathcal{D}_j}
        +\int_0^t e^{(t-s)\mathcal{L}_{j,M}}
                         \mathcal{G}_j e^{s\mathcal{D}_j}\,ds.
    \label{eq:sm-duhamel-loss}
\end{equation}
For an input supported in $P_j$, $e^{t\mathcal D_j}$ preserves that support.
Each term of $\mathcal{G}_j$ then sends it into $P_{j-1}$ by
\eqref{eq:sm-gain-lowering}.  The subsequent full evolution remains in
$P_{j-1}$ by corner invariance.  Applying $\mathcal{Q}_{P_{j-1}}^c$
therefore annihilates the integral in \eqref{eq:sm-duhamel-loss}.

Set $A_t=e^{-tK_M/N_j}P_j$.  Since $K_M$ commutes with $P_j$, the operator
$A_t$ is self-adjoint, and for an arbitrary input operator $x$ we obtain
\begin{equation}
    \mathcal{Q}_{P_{j-1}}^c e^{t\mathcal{L}_{j,M}}
          \mathcal{Q}_{P_j}(x)
       =A_txA_t-P_{j-1}A_txA_tP_{j-1}.
    \label{eq:sm-loss-only-remainder}
\end{equation}
Using $P_j=P_{j-1}+\Pi_j$ and the commutation with these projections, the
expression on the right is
\begin{equation}
    (\Pi_jA_t)xA_t+(P_{j-1}A_t)x(A_t\Pi_j).
    \label{eq:sm-two-block-bound}
\end{equation}
Positivity of $K_M$ gives $\|A_t\|_\infty\leq1$, and the assumed lower
bound on its $\Pi_j$ block gives
$\|\Pi_jA_t\|_\infty\leq e^{-\kappa_jt}$.
The map $x\mapsto bxc$ has diamond norm at most
$\|b\|_\infty\|c\|_\infty$, because the same multiplication bound for the
trace norm holds after tensoring with any ancillary identity.
Applying it to the two terms in \eqref{eq:sm-two-block-bound} proves
\eqref{eq:sm-one-step}.

Finally, the $\Pi_j$ diagonal block of the evolution is exactly
$\Pi_jA_t\rho A_t\Pi_j$.  For positive $\rho$, its trace is at most
$\|\Pi_jA_t\|_\infty^2\operatorname{Tr}(\Pi_j\rho)$, proving
\eqref{eq:sm-population}.  In the diagonal and the off-diagonal block,
\[
    \Pi_je^{t\mathcal{L}_{j,M}}(\rho)\Pi_j=(\Pi_jA_t)\rho(\Pi_jA_t)^*,
    \qquad
    P_{j-1}e^{t\mathcal{L}_{j,M}}(\rho)\Pi_j=(P_{j-1}A_t)\rho(\Pi_jA_t)^*,
\]
the factor $\Pi_jA_t$ appears twice and once, respectively.  This is why
\eqref{eq:sm-population} decays at twice the rate in
\eqref{eq:sm-one-step}.
\end{proof}

By Proposition~\ref{prop:sm-sampling}, with probability at least
$1-\zeta$ the hypothesis of Lemma~\ref{lem:sm-one-step} holds at every
stage with $\kappa_j=\gamma_j/2$.  We run stage $j$ until the bound
\eqref{eq:sm-one-step} has dropped to $\epsilon/(2J)$, that is, for the
time $t_j$ with $2e^{-\gamma_jt_j/2}=\epsilon/(2J)$.  Adding up the
errors of the $J$ stages gives the main result of this section.

\begin{theorem}[State preparation for the lowest cluster]
\label{thm:sm-preparation}
Assume the cluster structure \eqref{eq:sm-clusters},
Assumption~\ref{as:specluster-separation}, and filters satisfying
\eqref{eq:sm-filter}.  Let $0<\epsilon,\zeta<1$, and choose $M$ according
to \eqref{eq:sm-sample-count}.  Set
\begin{equation}
    t_j=\frac{2}{\gamma_j}\log\frac{4J}{\epsilon},\qquad
    \mathcal{T}=
       e^{t_1\mathcal{L}_{1,M}}\cdots e^{t_J\mathcal{L}_{J,M}}.
    \label{eq:sm-stage-times}
\end{equation}
The rightmost factor acts first.  With probability at least $1-\zeta$,
for every initial density matrix $\rho$ there is a density matrix
$\sigma_\rho=P_0\sigma_\rho P_0$ such that
\begin{equation}
    \|\mathcal{T}(\rho)-\sigma_\rho\|_1\leq\epsilon,
    \label{eq:sm-preparation}
\end{equation}
and moreover $\|\mathcal{Q}_{P_0}^c\mathcal{T}\|_\diamond\leq\epsilon$.
If $P_0=|\Omega\rangle\langle\Omega|$, then
$\sigma_\rho=|\Omega\rangle\langle\Omega|$ for every $\rho$.
The same conclusions, with the same times, hold for the sampling models
of Corollary~\ref{cor:sm-design-sampling} when
\eqref{eq:sm-design-count} is used instead.
\end{theorem}
\begin{proof}
Assume that the bounds \eqref{eq:sm-empirical-lower} hold for every $j$.
By Proposition~\ref{prop:sm-sampling}, or by
Corollary~\ref{cor:sm-design-sampling} for the other sampling models, this
happens with probability at least $1-\zeta$.  Define
\[
    C_j=\mathcal{Q}_{P_{j-1}}e^{t_j\mathcal{L}_{j,M}}
             \mathcal{Q}_{P_j}.
\]
The maps $e^{t_j\mathcal{L}_{j,M}}$ are channels.  The $C_j$ are completely positive,
trace-nonincreasing maps, with outputs supported in $P_{j-1}$.
All these maps have diamond norm at most one.  Since the output of
$C_{j+1}\cdots C_J$ is supported in $P_j$,
\[
    \bigl(e^{t_j\mathcal{L}_{j,M}}-C_j\bigr)\mathcal{Q}_{P_j}
       =\mathcal{Q}_{P_{j-1}}^c e^{t_j\mathcal{L}_{j,M}}
          \mathcal{Q}_{P_j}.
\]
Use the telescoping identity
\begin{equation}
\begin{split}
    &\mathcal{T}-C_1\cdots C_J\\
    &\quad=\sum_{j=1}^J
       e^{t_1\mathcal{L}_{1,M}}\cdots
       e^{t_{j-1}\mathcal{L}_{j-1,M}}
       \bigl(e^{t_j\mathcal{L}_{j,M}}-C_j\bigr)C_{j+1}\cdots C_J,
\end{split}
    \label{eq:sm-telescoping}
\end{equation}
where an empty product is the identity map.  For $j=J$, the required input
projection is simply $P_J=I$.  Lemma~\ref{lem:sm-one-step} and the
contraction of the other factors give
\begin{equation}
    \|\mathcal{T}-C_1\cdots C_J\|_\diamond
       \leq\sum_{j=1}^J2e^{-\gamma_jt_j/2}
       =\frac{\epsilon}{2}.
    \label{eq:sm-telescoping-error}
\end{equation}
The maps $C_j$ are used only for comparison; the evolving state need not
be exactly supported in a smaller corner after a finite stage.

Set $\tau_\rho=C_1\cdots C_J(\rho)$.  It is positive, supported in $P_0$,
and has trace at most one.  Since $\mathcal{T}(\rho)$ has trace one,
\eqref{eq:sm-telescoping-error} also gives
$1-\operatorname{Tr}\tau_\rho\leq\epsilon/2$.
Thus $\tau_\rho$ has positive trace and can be normalized:
$\sigma_\rho=\tau_\rho/\operatorname{Tr}\tau_\rho$.  Then
\begin{equation}
\begin{split}
    \|\mathcal{T}(\rho)-\sigma_\rho\|_1
       &\leq\|\mathcal{T}(\rho)-\tau_\rho\|_1
                  +\|\tau_\rho-\sigma_\rho\|_1\\
       &=\|\mathcal{T}(\rho)-\tau_\rho\|_1
                  +1-\operatorname{Tr}\tau_\rho
        \leq\epsilon.
\end{split}
\end{equation}
Finally, $\mathcal{Q}_{P_0}^c(C_1\cdots C_J)=0$, so
\eqref{eq:sm-telescoping-error} and
$\|\mathcal{Q}_{P_0}^c\|_\diamond\leq2$ give the stated bound in diamond
norm.  The compressed maps and the normalization of $\tau_\rho$ establish the
existence of the target state.  The evolution $\mathcal{T}$ requires no
intermediate measurements or postselection.
\end{proof}

\begin{rem}[Time normalization and sample reuse]
\label{rem:sm-time-normalization}
For the Hamiltonian with clustered spectrum, the stage generators are scalar multiples of the same $\mathcal{B}_M$.
Consequently,
\begin{equation}
\begin{split}
    \mathcal{T}
       &=\exp\left(\left[\sum_{j=1}^J\frac{t_j}{N_j}\right]
                           \mathcal{B}_M\right),\\
    \sum_{j=1}^Jt_j
       &\leq\frac{2J}{\gamma}\log\frac{4J}{\epsilon}.
\end{split}
    \label{eq:sm-same-generator}
\end{equation}
The same sampled generator is used at every stage.  The bound on
$\sum_jt_j$ corresponds to the normalization
$\mathcal{L}_{j,M}=\mathcal{B}_M/N_j$ at stage $j$.  For the fixed
normalization $\mathcal{L}_{\mathrm{fix},M}=\mathcal{B}_M/N$, one has
\begin{equation}
    e^{t_j\mathcal{L}_{j,M}}
       =e^{(Nt_j/N_j)\mathcal{L}_{\mathrm{fix},M}}.
    \label{eq:sm-fixed-normalization}
\end{equation}
The corresponding total time is $\sum_j Nt_j/N_j$.  Normalizing by $N_j$
controls $P_j(\mathbb{E}K_M)P_j/N_j$ on the active corner; it does not give a
uniform bound on the norms of the jump operators over the full Hilbert space.  Any implementation cost must
account for the specified normalization. In practice, one may has to stop at a certain $N_j$ as a function of the available running time (for simulation).  
\end{rem}


\begin{rem}[More on the estimates]
\label{rem:sm-scope}
The results of this section concern exact filtered generators,
with Gaussian samples or samples from Haar measure or exact designs as in
Corollary~\ref{cor:sm-design-sampling}.  Gate complexity also depends on
the representation of the random operators, the accuracy of the frequency
filters, and simulation at the prescribed normalizations.  The dense
$N\times N$ Gaussian model alone gives no construction of circuits of
polynomial size.  For Lindblad simulation using primitives for simulating
the Hamiltonian, see~\cite{DingLiLin2024}.  The stage times refer to the
normalization $1/N_j$, under which the jump operators have norm of order one on
the active corner.  On the whole space their norm is of order
$\sqrt{N/N_j}$, here and for the jump operators of
Section~\ref{sec:first-method}, and simulation costs that depend on the
normalization of a block encoding~\cite{CleveWang2017} reflect this factor.
A gate count polynomial in $n$ therefore requires access to jump operators
normalized on the active corner, which structured random operators can
provide under ETH-type hypotheses. If the lowest cluster contains
several energies, preparing its ground state requires a further
construction.
\end{rem}

\subsection{Generators that vary by stage}
\label{subsec:sm-stage-dependent}

The proof of Theorem~\ref{thm:sm-preparation} uses the clusters only
through three properties of each stage: the jump operators map
$P_j\mathcal K$ into $P_{j-1}\mathcal K$, the operator $K_M$ commutes with
$P_j$ and $P_{j-1}$, and the block $\Pi_jK_M\Pi_j$ has a positive lower
bound.  The next proposition keeps only these properties, for arbitrary
nested projections and for jump operators that may change from stage to
stage.  In particular, the projections need not come from clusters
satisfying Assumption~\ref{as:specluster-separation}.

\begin{prop}[Reduction with generators that vary by stage]
\label{prop:sm-stage-dependent}
Let $P_0<P_1<\cdots<P_J=I$ be nested orthogonal projections and put
$\Pi_j=P_j-P_{j-1}$.  At stage $j$, let $b_{j,h}$ be finitely many jump operators,
including their scalar normalizations, and set
\[
 \mathcal A_j=\sum_h\mathcal L_{b_{j,h}},\qquad
 R_j=\sum_h b_{j,h}^*b_{j,h}.
\]
Suppose, for every $j$ and $h$, that
\begin{equation}
\begin{split}
 b_{j,h}P_j&=P_{j-1}b_{j,h}P_j,\\
 [R_j,P_j]&=[R_j,P_{j-1}]=0,\qquad
 \Pi_jR_j\Pi_j\geq\kappa_j\Pi_j\quad(\kappa_j>0).
\end{split}
\label{eq:sm-stage-dependent-assumptions}
\end{equation}
Then
\[
 \|\mathcal Q_{P_{j-1}}^c e^{t\mathcal A_j}\mathcal Q_{P_j}\|_\diamond
       \leq2e^{-\kappa_jt}.
\]
For $0<\epsilon<1$, choose
$t_j=\kappa_j^{-1}\log(4J/\epsilon)$.  The channel
$\mathcal T=e^{t_1\mathcal A_1}\cdots e^{t_J\mathcal A_J}$ has
$\|\mathcal Q_{P_0}^c\mathcal T\|_\diamond\leq\epsilon$, and every
$\mathcal T(\rho)$ is within trace norm $\epsilon$ of a normalized state
supported in $P_0$.
\end{prop}
\begin{proof}
The evolution generated by $\sum_h\mathcal M_{b_{j,h}}$ is
$x\mapsto e^{-tR_j}xe^{-tR_j}$ and preserves both corners.  Each
$\Psi_{b_{j,h}}$ maps operators supported in $P_j$ to operators supported
in $P_{j-1}$ by the first assumption.  Both $\Psi_{b_{j,h}}$ and
$\mathcal M_{b_{j,h}}$ preserve $P_{j-1}$, so the full evolution does
too.  Duhamel's formula therefore removes all terms containing some
$\Psi_{b_{j,h}}$ after applying $\mathcal Q_{P_{j-1}}^c$ to an input in $P_j$.
The remaining map decomposes into the two terms in
\eqref{eq:sm-two-block-bound}, with $A_t=e^{-tR_j}P_j$.
This proves the bound for a single stage.  The compressed comparison maps and
telescoping argument in the proof of Theorem~\ref{thm:sm-preparation}
then give the two conclusions.  No commutation or independence between
different stages is used.
\end{proof}

For the clusters of Section~\ref{subsec:sm-clusters}, take the jump
operators $b_{j,(w,\nu)}=a_w(g^{(\nu)})/\sqrt{MN_j}$ at stage $j$.  Then
$\mathcal A_j=\mathcal L_{j,M}$ and $R_j=K_M/N_j$, and
Proposition~\ref{prop:sm-sampling} gives $\kappa_j=\gamma_j/2$, so the
proposition recovers Theorem~\ref{thm:sm-preparation}.  
For a refinement using exact gaps ($\kappa$), the selected gaps must lower the
active input corner, and the sample average of $\sum_\kappa a_\kappa^*a_\kappa$
must have a lower bound on the removed block.  The number of selected gaps must also be counted.
For the Hamiltonian $H$ in \eqref{eq:ex-weighted-hamiltonian} with $b>2$,
we use one broad filter at each stage and compress its jump operator on both sides
by the active spectral projection.  Access to this compressed operator is
an additional assumption.  The compression is part of the jump operator;
it requires no measurement or postselection of the state.

\section{Examples}
\label{sec:examples}

For each Hamiltonian $H$, we compute its spectrum and multiplicities, then
verify the hypotheses of the corresponding preparation method.
We recall the conditions to be checked.  Write
$\lambda_1\leq\cdots\leq\lambda_N$ for the eigenvalues, counted with
multiplicity, and let
\begin{equation}
 F(E)=\operatorname{Tr}\mathbf1_{(-\infty,E]}(H).
 \label{eq:ex-cumulative-count}
\end{equation}

\medskip\noindent\textbf{Method for Hamiltonian with General Spectrum.}
Choose ranks $1\leq N_0<\cdots<N_m=N$, with cumulative projections
$P_j=\sum_{l\leq N_j}E_{ll}$.
At every stage $j$, Propositions~\ref{prop:boxass} and~\ref{prop:gammaass}
require integers \notationrev{$1\leq q_j\leq N_{j-1}$},
cutoffs $0<\alpha_j'<\alpha_j\leq\beta_j<\beta_j'$, and $\gamma_j>0$
such that
\begin{equation}
\begin{split}
 \max\{\lambda_{q_j}-\lambda_1,
             \lambda_{N_j}-\lambda_{q_j+1}\}&\leq\alpha_j',\\
 \alpha_j\leq\lambda_{N_j}-\lambda_1&\leq\beta_j,\\
 F(\lambda_t-\alpha_j)&\geq\gamma_jN_j
                 \qquad(\notationrev{N_{j-1}}<t\leq N_j).
\end{split}
\label{eq:ex-first-method-assumptions}
\end{equation}
The smooth filter satisfies $0\leq\hat\phi_j\leq1$, is zero outside
$(\alpha_j',\beta_j')$, is positive inside this interval, and equals
one on $[\alpha_j,\beta_j]$, as in
\eqref{eq:first-filter-zero}--\eqref{eq:first-filter-plateau}.
With the sampling choices of
Theorem~\ref{thm:first-empirical-preparation}, these conditions give
the error bound \eqref{eq:first-empirical-preparation} with target
$\mathrm{range}(P_0)$.  To bound the time and matrix count, we also
estimate the number $m$ of stages, $\gamma_j^{-1}$,
and the enlargement ratio
\[
 \chi_j=\frac{F(\lambda_{N_j}+\beta_j')}{N_j}.
\]

\medskip\noindent\textbf{Method for Hamiltonian with Clustered Spectrum.}
For Theorem~\ref{thm:sm-preparation}, choose $J\geq1$, nonzero, mutually
orthogonal spectral projections $\Pi_0,\ldots,\Pi_J$ with
$\sum_{j=0}^J\Pi_j=I$, centers $m_0<\cdots<m_J$, and $\delta>0$.
The cluster and separation conditions are
\begin{equation}
\begin{split}
 \|(H-m_jI)\Pi_j\|_\infty&\leq\delta
                      \qquad(0\leq j\leq J),\\
 \min_{\substack{w,w'\in W\cup\{0\}\\w\ne w'}}|w-w'|&>6\delta,
 \qquad W=\{m_j-m_i:0\leq i<j\leq J\}.
\end{split}
\label{eq:ex-second-method-assumptions}
\end{equation}
These are \eqref{eq:sm-clusters} and
Assumption~\ref{as:specluster-separation}.  Set
\[
 P_j=\sum_{i=0}^j\Pi_i,\qquad N_j=\operatorname{rank}P_j,\qquad
 \gamma_j=\frac{N_{j-1}}{N_j},\qquad
 \gamma=\min_{1\leq j\leq J}\gamma_j.
\]
Filters satisfying \eqref{eq:sm-filter} and the independent standard
complex Gaussian samples of Section~\ref{subsec:sm-generator}, with
count \eqref{eq:sm-sample-count}, give the error bound
\eqref{eq:sm-preparation} with target $\mathrm{range}(\Pi_0)$.
We estimate $J$, $|W|$, and $\gamma^{-1}$
to bound the matrix and jump counts and the times in
\eqref{eq:sm-stage-times}.

For nested orthogonal projections
$P_0<\cdots<P_J=I$, normalized jumps $b_{j,h}$, and
$R_j=\sum_hb_{j,h}^*b_{j,h}$, its conditions are
\[
\begin{gathered}
 b_{j,h}P_j=P_{j-1}b_{j,h}P_j,\qquad
 [R_j,P_j]=[R_j,P_{j-1}]=0,\\
 \Pi_jR_j\Pi_j\geq\kappa_j\Pi_j,\qquad
 \Pi_j=P_j-P_{j-1},\quad\kappa_j>0.
\end{gathered}
\]
These are \eqref{eq:sm-stage-dependent-assumptions}.

Each example ends with numerical results, with the numerical approach described in
Appendix~\ref{app:numerics}. The figures follow one realization from a
random pure state. Trajectory figures show the energy of $10$ unraveled quantum
trajectories and of their mean, with the population of each eigenvalue on
a logarithmic scale and the boundaries of the corners of the current stage.

\subsection{Number operators with geometric coefficients: first method}
\label{subsec:ex-weighted}

For $n\geq1$, on $(\mathbb C^2)^{\otimes n}$ with orthonormal basis
$\{|z\rangle:z\in\{0,1\}^n\}$, define the linear combination of
single-site number operators
\begin{equation}
 H=\sum_{s=0}^{n-1}b^s n_s,\qquad
 n_s|z\rangle=z_s|z\rangle,\qquad b>1.
 \label{eq:ex-weighted-hamiltonian}
\end{equation}
The eigenvalues are
\[
 E(z)=\sum_{s=0}^{n-1}b^s z_s,\qquad z\in\{0,1\}^n,
\]
counted with multiplicity.  The unique ground state is
$\Omega=|0,\ldots,0\rangle$, with energy zero, and
$\lambda_{\max}=(b^n-1)/(b-1)$.  We first establish the eigenvalue
bounds used to construct the cutoffs for the first method.

\begin{lemma}[Eigenvalue bounds for $H$]
\label{lem:ex-subset-sums}
Let $\lambda_1^+\leq\cdots\leq\lambda_{2^n-1}^+$ be the strictly
positive eigenvalues of $H$, counted with multiplicity, and set
$p=\log b/\log2$.  Then
\begin{equation}
 \frac1b j^p\leq\lambda_j^+\leq\frac{b}{b-1}j^p.
 \label{eq:ex-subset-regularity}
\end{equation}
For $d>\max\{1,b-1\}$ and
$\mu_0=(d-\max\{1,b-1\})^{-1}$,
\[
 \min\{\lambda\in\sigma(H):\lambda\geq\mu\}\leq d\mu
 \qquad(\max\{1,\mu_0\}\leq\mu\leq\lambda_{\max}).
\]
\end{lemma}
\begin{proof}
Let $\sigma(j)\subseteq\{0,\ldots,n-1\}$ be defined by
$j=\sum_{s\in\sigma(j)}2^s$, and put
$f_b(j)=\sum_{s\in\sigma(j)}b^s$.  If $2^k\leq j<2^{k+1}$, then
\[
 b^k\leq f_b(j)\leq\frac{b^{k+1}}{b-1},\qquad
 b^k\leq j^p<b^{k+1}.
\]
Thus $b^{-1}j^p\leq f_b(j)\leq b(b-1)^{-1}j^p$.
Increasing rearrangement preserves these bounds, proving
\eqref{eq:ex-subset-regularity}.

For $1<b\leq2$, let
$\Sigma_k=\left \{\sum_{s=0}^{k-1}b^s z_s:z_s\in\{0,1\}\right \}$.
The identity $\Sigma_{k+1}=\Sigma_k\cup(b^k+\Sigma_k)$ and
$b^k-\sum_{s<k}b^s\leq1$ show inductively that consecutive distinct
values have gap at most one.  For $b>2$, $f_b(j)$ is strictly
increasing.  If the binary expansion of $j$ ends in $k$ ones, then
the consecutive values have the form
$x=A+\sum_{s<k}b^s$ and $y=A+b^k$, with $A\geq0$.  Hence
\[
 y-x=(b-2)\sum_{s<k}b^s+1\leq(b-2)x+1.
\]
In both cases, $y\leq\max\{1,b-1\}x+1$ for consecutive eigenvalues.
For $\max\{1,\mu_0\}\leq\mu\leq\lambda_{\max}$,
the least eigenvalue at least $\mu$ is therefore at most
$\max\{1,b-1\}\mu+1\leq d\mu$.
\end{proof}

\begin{prop}[First method for $H$]
\label{regdense}
Fix $1<b<(3+\sqrt5)/2$ and let $H$ be given by
\eqref{eq:ex-weighted-hamiltonian}.  There are constants
$C_b,\gamma_b,\chi_b>0$, depending only on $b$ and not on $n$,
such that the first method uses at most $C_b n$ stages and has
target $\operatorname{span}\{\Omega\}$.
More precisely, one can choose ranks
$1=N_0<\cdots<N_m=2^n$, with $1\leq m\leq C_b n$,
and, at each stage $1\leq j\leq m$, an integer
$1\leq q_j\leq N_{j-1}$ and cutoffs
$0<\alpha_j'<\alpha_j\leq\beta_j<\beta_j'$ satisfying the support
inequalities
\[
 \max\{\lambda_{q_j},\lambda_{N_j}-\lambda_{q_j+1}\}
       \leq\alpha_j',\qquad
 \alpha_j\leq\lambda_{N_j}\leq\beta_j,
\]
the mass bound
\[
 F(\lambda_t-\alpha_j)\geq\gamma_bN_j
          \qquad(N_{j-1}<t\leq N_j),
\]
and the enlargement bound
\[
 \frac{F(\lambda_{N_j}+\beta_j')}{N_j}\leq\chi_b.
\]
\notationrev{With smooth filters}
$0\leq\hat\phi_j\leq1$ satisfying
\eqref{eq:first-filter-zero}--\eqref{eq:first-filter-plateau}, these
are the hypotheses of Propositions~\ref{prop:boxass} and~\ref{prop:gammaass}.
The final projection is $P_0=|\Omega\rangle\langle\Omega|$.
\end{prop}
\begin{proof}
Choose $d,u,v$ depending only on $b$ such that
\[
 \max\{1,b-1\}<d<\frac{1+\sqrt5}{2},\qquad
 \frac{d}{1+d}<u<v<\frac1d.
\]
Set $\eta=(v-u)/4$ and $p=\log b/\log2$, and let
$\mu_0=(d-\max\{1,b-1\})^{-1}$ as in
Lemma~\ref{lem:ex-subset-sums}.
At a spectral cutoff $L$ with $(1-u)L\geq\max\{1,\mu_0\}$,
Lemma~\ref{lem:ex-subset-sums} gives eigenvalues $Q,S$ satisfying
\[
 (1-u)L\leq Q<uL,\qquad vL\leq S\leq dvL<L.
\]
Set
\[
 \begin{gathered}
 N=F(L),\quad q=F(Q),\quad s=F(S),\quad r=F(\eta L),\\
 \alpha'=(u+\eta)L,\quad\alpha=(u+2\eta)L,\quad
 \beta=L,\quad\beta'=2L.
 \end{gathered}
\]
Then $r\leq q\leq s<N$, $Q<\alpha'$, and $L-Q<\alpha'$.
Moreover, for $1\leq l\leq r$ and $s<t\leq N$,
\[
 \alpha<(v-\eta)L\leq\lambda_t-\lambda_l\leq L=\beta.
\]
Since $\lambda_q=Q$ and $\lambda_N=L$, the support inequalities hold.
The $r$ eigenvalues with $\lambda_l\leq\eta L$ give
$F(\lambda_t-\alpha)\geq r$ for every $s<t\leq N$.
Taking $N_j=N$ and $N_{j-1}=s$ therefore gives the required mass
bound with $\gamma_j=r/N$, while $\chi_j=F(3L)/N$.

For $L\geq1/b$, \eqref{eq:ex-subset-regularity} gives
\[
 \frac rN\geq\frac12\left(\frac{(b-1)\eta}{b^2}\right)^{1/p},\qquad
 \frac{F(3L)}N\leq\frac{F(3L)}r
 \leq2\left(\frac{3b^2}{(b-1)\eta}\right)^{1/p}.
\]
This cutoff construction applies for $L\geq L_*$, where
\[
 L_*:=\frac{\max\{1,\mu_0\}}{1-u}>1.
\]
Starting from $L=\lambda_{\max}$, replace $L$ by $S\leq dvL$ at
each stage until $L\leq L_*$.
Since $dv<1$ and $\lambda_{\max}\leq(1+b)^n$, at most
\[
 \left\lceil
 \frac{\max\{0,\log(\lambda_{\max}/L_*)\}}
      {\log(1/(dv))}\right\rceil
 \leq\left(1+\frac{\log(1+b)}{\log(1/(dv))}\right)n
\]
stages are needed to reach a subspace of dimension at most
\[
 F(L_*)\leq1+(bL_*)^{1/p}=:D_b.
\]

To reach the ground state, continue with the remaining distinct
positive eigenvalues in decreasing order.  At each stage, let
$L\leq L_*$ be the largest active eigenvalue and $L_-$ the largest
eigenvalue strictly less than $L$, and choose
\[
\begin{gathered}
 N_j=F(L),\qquad q_j=\notationrev{N_{j-1}}=F(L_-),\\
 \alpha_j'=\frac{2L_-+L}{3},\qquad
 \alpha_j=\frac{L_-+2L}{3},\qquad
 \beta_j=L,\quad\beta_j'=2L.
\end{gathered}
\]
The identities $\lambda_{q_j}=L_-$ and
$\lambda_{q_j+1}=\lambda_{N_j}=L$ verify the support inequalities.
For every $N_{j-1}<t\leq N_j$, the ground state contributes to the
mass bound:
\[
 F(\lambda_t-\alpha_j)
 =F\left(\frac{L-L_-}{3}\right)\geq1\geq\frac{N_j}{D_b}.
\]
The enlargement ratio satisfies
\[
 \frac{F(\lambda_{N_j}+\beta_j')}{N_j}
 \leq F(3L_*)\leq1+(3bL_*)^{1/p}.
\]
Each such stage has active projection $P_j=\mathbf1_{[0,L]}(H)$ and
target projection $P_{j-1}=\mathbf1_{[0,L_-]}(H)$.
Thus at most $D_b-1$ additional stages reach
$P_0=|\Omega\rangle\langle\Omega|$.  The total number is at most
$C_b n$, where $C_b=D_b+\log(1+b)/\log(1/(dv))$.
Both parts of the construction satisfy the proposition with
\[
\begin{split}
 \gamma_b&=\min\left\{
 \frac12\left(\frac{(b-1)\eta}{b^2}\right)^{1/p},\frac1{D_b}\right\}>0,\\
 \chi_b&=\max\left\{
 2\left(\frac{3b^2}{(b-1)\eta}\right)^{1/p},
 1+(3bL_*)^{1/p}\right\}.
\end{split}
\]
The constants $C_b,\gamma_b,\chi_b$ depend only on $b$.
\end{proof}

With the sample counts and stage times of
Theorem~\ref{thm:first-empirical-preparation}, Proposition~\ref{regdense}
therefore gives, with probability at least $1-\zeta$,
\[
 \|\mathcal T(\rho)-|\Omega\rangle\langle\Omega|\|_1\leq\epsilon
\]
for every initial density matrix $\rho$, where $\mathcal T$ is the
composition of the $m$ stage channels.

\paragraph{Numerical results.}
The first method reaches the ground state of $H_b$ for every $b$ tested,
including $b>2$ (Table~\ref{tab:ex-weighted-first}).  The largest
bottleneck mass is $1/2$ for $b\geq2.3$ and $1/3$ for $b=2$; it decreases as
$b$ approaches one, where the subset sums crowd, and for $b=1.2$ it
decreases slowly with $n$.  The masses guaranteed by Lemma~\ref{regdense}
are smaller than the exact ones by one to three orders of magnitude.  For
$H_2$ with $n=8$, whose spectrum is $\{0,1,\ldots,255\}$, the fastest
schedule has $15$ stages and reaches the error $5.9\cdot10^{-7}$ in total
time $593$ (Figure~\ref{fig:ex-traj-weighted}).  Each output projection
$Q_j$ is the lower half of the active corner, so every jump lands below
half of the largest active eigenvalue, and the trajectories make one to
three jumps.

\begin{table}[ht]
\centering\footnotesize\setlength{\tabcolsep}{3.5pt}
\begin{tabular}{rrcrrrcc}
\hline
$b$ & greedy & $\gamma^*$ & stages & stop rank & rank bound & $\min\gamma$, lemma & $\min\gamma$, exact \\
\hline
1.2 & 18 & 0.092 & 9 & 691 & 89689 & $2.4\cdot10^{-4}$ & 0.011 \\
1.4 & 14 & 0.202 & 47 & 22 & 202 & $3.3\cdot10^{-4}$ & 0.086 \\
1.6 & 13 & 0.300 & 496 & 8 & 36 & $2.4\cdot10^{-4}$ & 0.214 \\
2.0 & 15 & 0.333 & 63 & 9 & 19 & $1.6\cdot10^{-2}$ & 0.248 \\
2.3 & 12 & 0.500 & 81 & 12 & 22 & $1.6\cdot10^{-2}$ & 0.333 \\
2.6 & 12 & 0.500 & 364 & 64 & 123 & $2.8\cdot10^{-3}$ & 0.471 \\
3.0 & 12 & 0.500 & 0 & -- & -- & -- & -- \\
\hline
\end{tabular}
\caption{First method for $H_b$, $n=12$: number of greedy stages, largest
bottleneck mass $\gamma^*$, and the stages of Lemma~\ref{regdense} (number
of stages, stopping rank and its bound, least mass $r/N$ from the lemma,
least exact mass).  For $b=3$ no admissible density parameter exists.}
\label{tab:ex-weighted-first}
\end{table}

\begin{figure}[ht]
\centering
\includegraphics[width=\linewidth]{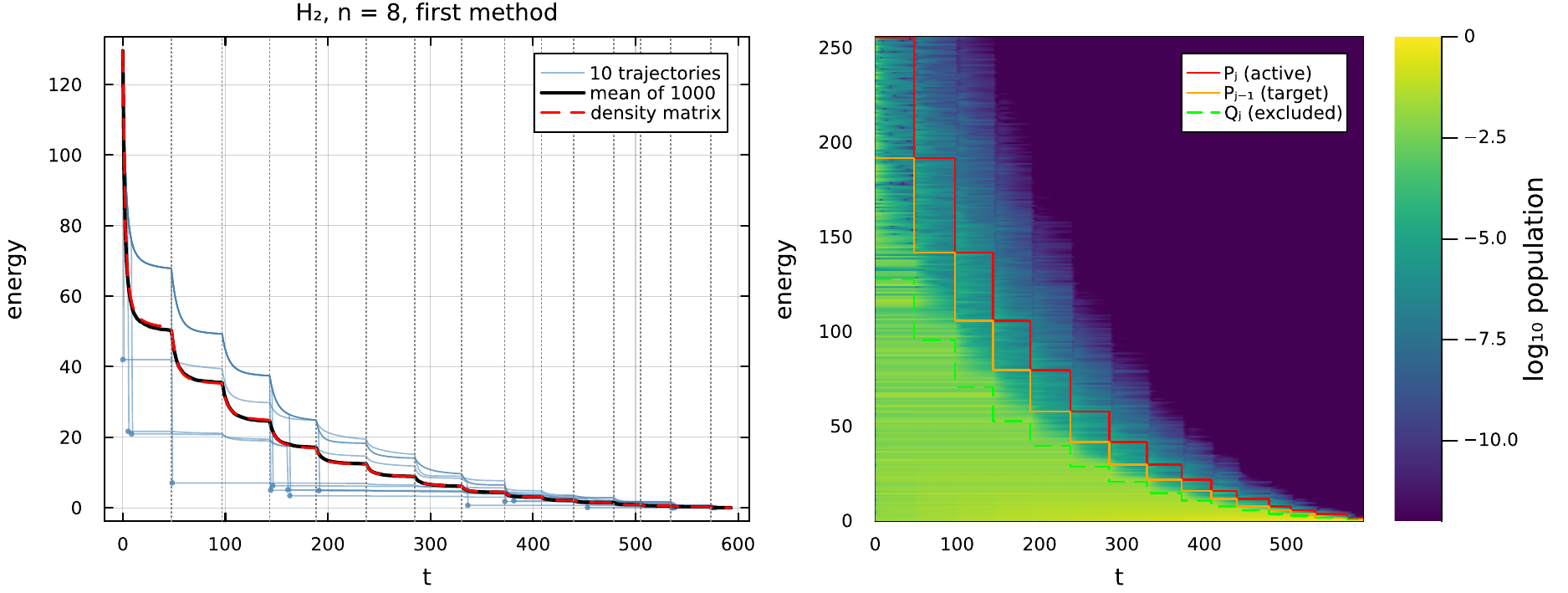}
\caption{First method for $H_2$, $n=8$, fastest schedule ($15$ stages,
$20$ jump operators each), from a random initial state.  Left: energy of
quantum trajectories (jumps marked) and of the density matrix.  Right:
populations of the eigenvalues ($\log_{10}$), with the boundaries of
$P_j$, $P_{j-1}$ and $Q_j$.}
\label{fig:ex-traj-weighted}
\end{figure}

\begin{rem}
\label{rem:ex-large-base}
For $b>2$ the subset sums in \eqref{eq:ex-weighted-hamiltonian} are
distinct, and in all cases we computed, a common cluster width compatible
with Assumption~\ref{as:specluster-separation} leaves a single eigenvalue in
each cluster, so that Theorem~\ref{thm:sm-preparation} needs $J=2^n-1$
stages.  Jump operators that vary by stage avoid this.  Compressing one
broad filtered jump operator on both sides by the active spectral
projection gives stages that satisfy the conditions of
Proposition~\ref{prop:sm-stage-dependent}, with $n$ stages, of order
$2^r+\log(n/\zeta)$ Gaussian matrices and total time of order
$n2^r\log(n/\epsilon)$, where $r$ depends only on $b$.  This requires access
to the compressed operators.  For $2<b<(3+\sqrt5)/2$ the first method
applies without this assumption by Proposition~\ref{regdense}, and in our
computations it reaches the ground state for every $b\leq3$ tested
(Table~\ref{tab:ex-weighted-first}).
\end{rem}

\subsection{Geometric eigenvalues}
\label{subsec:ex-geometric}

Let $R_0,\ldots,R_m$ be nonzero, mutually orthogonal projections with
$\sum_{\ell=0}^mR_\ell=I$.  For $C>0$ and $2/3<x<1$, define
\[
 H=\sum_{\ell=0}^m\Lambda_\ell R_\ell,
 \qquad \Lambda_\ell=Cx^{-\ell},\qquad d_\ell=\operatorname{rank}R_\ell.
\]
Its distinct eigenvalues are $\Lambda_0<\cdots<\Lambda_m$, with
eigenspaces $\mathrm{range}(R_\ell)$ and multiplicities $d_\ell$.
The cumulative spectral projections and their ranks are
\[
 P_j=\sum_{\ell=0}^jR_\ell,\qquad
 N_j=\operatorname{rank}P_j=\sum_{\ell=0}^jd_\ell.
\]

\begin{lemma}[Cutoffs for geometric values]
\label{lem:geometric-values-current}
There are integers $1\leq m_0<k_0$, depending only on $x$, such that
\begin{equation}
 \max\{\Lambda_{j-m_0},\Lambda_j-\Lambda_{j-m_0}\}
 <\Lambda_{j-1}-\Lambda_{j-k_0}
 \qquad(k_0<j\leq m).
 \label{eq:geometric-lambda-regular}
\end{equation}
Thus Theorem~\ref{thm:regular-values-current} gives admissible stages
$j=m,m-1,\ldots,k_0+1$ whenever $m>k_0$.
\end{lemma}
\begin{proof}
Choose $m_0$ minimal with $x^{m_0}\leq1/2$.  Then
$x^{m_0}>x/2>1-x$, so one can choose $k_0>m_0$ with
$x^{k_0}<x+x^{m_0}-1$.  Consequently,
\begin{equation}
 \max\{x^{m_0},1-x^{m_0}\}
 =1-x^{m_0}<x-x^{k_0}.
 \label{eq:geometric-regular-constant}
\end{equation}
Multiplying by $\Lambda_j$ proves
\eqref{eq:geometric-lambda-regular}.
\end{proof}

Assume henceforth that $m>k_0$.  With $\Theta_j,D_j$ as in
\eqref{eq:regular-values-current}, choose
$\delta_j=(D_j-\Theta_j)/3$ and the filter parameters
\eqref{eq:regular-stage-parameters}.  They satisfy
\[
 \Theta_j<\alpha_j'<\alpha_j<D_j,\qquad
 \beta_j=\Lambda_j<\beta_j'=\Lambda_j+\delta_j<2\Lambda_j.
\]
The support and mass conditions therefore hold with
\[
 q_j=N_{j-m_0},\qquad r_j=N_{j-k_0},\qquad
 \gamma_j=\frac{N_{j-k_0}}{N_j}.
\]
For $s=\lceil\log3/\log(1/x)\rceil$, the enlarged projection satisfies
\[
 N_j^{\beta_j'}\leq N_{\min\{m,j+s\}},\qquad
 \frac{N_j^{\beta_j'}}{r_j}
 \leq\frac{N_{\min\{m,j+s\}}}{N_{j-k_0}},
\]
because $\lambda_{N_j}+\beta_j'<3\Lambda_j\leq\Lambda_{j+s}$
(using $\Lambda_\ell=Cx^{-\ell}$ also for $\ell>m$).
In particular, if $N_{\ell-1}\geq\gamma_{\mathrm{mult}}N_\ell$
for $1\leq\ell\leq m$, with fixed $0<\gamma_{\mathrm{mult}}\leq1$, then
\[
 \gamma_j\geq\gamma_{\mathrm{mult}}^{k_0},\qquad
 \frac{N_j^{\beta_j'}}{r_j}
 \leq\gamma_{\mathrm{mult}}^{-(s+k_0)}.
\]
This verifies \eqref{eq:stage-condition-iv} uniformly in $m$.

\begin{cor}[Equal multiplicities]
\label{cor:exponential-rescaling-current}
If $d_\ell=d$ for every $\ell$, then, for $j>k_0$,
\[
 \gamma_j=\frac{j-k_0+1}{j+1}\geq\frac{2}{k_0+2},\qquad
 \frac{N_j^{\beta_j'}}{r_j}
 \leq1+\frac{s+k_0}{2}.
\]
Both bounds depend only on $x$, so
Theorem~\ref{thm:first-empirical-preparation} applies.
The $m-k_0$ stages have target
$\mathrm{range}(P_{k_0})$, of dimension $d(k_0+1)$.
\end{cor}
\begin{proof}
Substitute $N_j=d(j+1)$ into the preceding ratios and use
$j-k_0+1\geq2$.
\end{proof}

For general multiplicities, the same iteration has target
$\mathrm{range}(P_{k_0})$ of dimension $N_{k_0}$; preparing a state in
a proper subspace requires an additional step.

\subsection{A finite spectrum at semicircular quantiles}
\label{subsec:ex-semicircle}

We give a finite-dimensional example whose empirical eigenvalue
distribution converges to the semicircular law.  Shift the radius-one
semicircle to $[0,2]$ and write its cumulative distribution as
\begin{equation}
 G_{\mathrm{sc}}(E)
 =\frac{2}{\pi}\int_0^{\min\{E,2\}}\sqrt{x(2-x)}\,dx,
 \qquad E\geq0.
 \label{eq:ex-semicircle-cdf}
\end{equation}
Thus $G_{\mathrm{sc}}(0)=0$ and $G_{\mathrm{sc}}(E)=1$ for $E\geq2$.
Let $N=2^n$, with $n\geq1$, and choose any orthonormal basis
$\{\psi_k\}_{k=1}^N$.  Define
\begin{equation}
 H=\sum_{k=1}^N\lambda_k|\psi_k\rangle\langle\psi_k|,
 \qquad
 \lambda_k=G_{\mathrm{sc}}^{-1}\left(\frac{k-1}{N-1}\right).
 \label{eq:ex-semicircle-spectrum}
\end{equation}
Here the inverse is taken on $[0,2]$, so
$0=\lambda_1<\cdots<\lambda_N=2$ and $\psi_1$ is the unique ground
state.  The counting function in \eqref{eq:ex-cumulative-count} is
\begin{equation}
 F(E)=1+\left\lfloor(N-1)G_{\mathrm{sc}}(E)\right\rfloor,
 \qquad E\geq0.
 \label{eq:ex-semicircle-count}
\end{equation}
In particular, $|F(E)/N-G_{\mathrm{sc}}(E)|\leq N^{-1}$.
The quantile prescription gives control of the individual low-lying
eigenvalues as well as of the limiting distribution.

\begin{prop}[First method for semicircular quantiles]
\label{prop:ex-semicircle}
For the Hamiltonian \eqref{eq:ex-semicircle-spectrum}, the first method
reaches $P_0=|\psi_1\rangle\langle\psi_1|$ in $m=O(n)$ stages.
The cutoffs can be chosen to satisfy
\eqref{eq:ex-first-method-assumptions} with the uniform bounds
\begin{equation}
 \gamma_j\geq\gamma_*:=\frac{1}{1+8^{3/2}},\qquad
 \chi_j=\frac{F(\lambda_{N_j}+\beta_j')}{N_j}
       \leq\chi_*:=3^{3/2}.
 \label{eq:ex-semicircle-ratios}
\end{equation}
Consequently, Theorem~\ref{thm:first-empirical-preparation} gives,
with probability at least $1-\zeta$, simultaneously for every initial
density matrix $\rho$,
\[
 \|\mathcal T(\rho)-|\psi_1\rangle\langle\psi_1|\|_1
 \leq\epsilon.
\]
Under the normalization of that theorem, the total evolution time is
$O(n\log(2n/\epsilon))$, and the total number of random matrix samples
is polynomial in $n$, $\epsilon^{-1}$, and $\log(1/\zeta)$.
\end{prop}
\begin{proof}
For $0<E\leq2$, the change of variables $x=Eu$ gives
\begin{equation}
 \frac{G_{\mathrm{sc}}(E)}{E^{3/2}}
 =\frac{2}{\pi}\int_0^1\sqrt{u(2-Eu)}\,du.
 \label{eq:ex-semicircle-scaling}
\end{equation}
This quotient is decreasing in $E$.  It remains decreasing for
$E\geq2$, where $G_{\mathrm{sc}}(E)=1$.  In particular,
\begin{equation}
 G_{\mathrm{sc}}(L/8)\geq8^{-3/2}G_{\mathrm{sc}}(L),\qquad
 G_{\mathrm{sc}}(3L)\leq3^{3/2}G_{\mathrm{sc}}(L)
 \quad(0<L\leq2).
 \label{eq:ex-semicircle-dilation}
\end{equation}
The endpoint values of the quotient also give
\begin{equation}
 2^{-3/2}E^{3/2}\leq G_{\mathrm{sc}}(E)
 \leq\frac{4\sqrt2}{3\pi}E^{3/2},\qquad 0\leq E\leq2.
 \label{eq:ex-semicircle-edge}
\end{equation}
Thus the cumulative mass near the lower edge has order $E^{3/2}$.

At an active cutoff $L=\lambda_{N_j}>0$, choose
\begin{equation}
 \begin{gathered}
 \notationrev{q_j}=F(L/2),\qquad
 \notationrev{N_{j-1}}=F(3L/4),\\
 \alpha_j'=L/2,\qquad \alpha_j=5L/8,\qquad
 \beta_j=L,\qquad\beta_j'=2L.
 \end{gathered}
 \label{eq:ex-semicircle-cutoffs}
\end{equation}
Since $\lambda_{q_j}\leq L/2<\lambda_{q_j+1}$, both support
inequalities in \eqref{eq:ex-first-method-assumptions} hold.
Moreover, \notationrev{$1\leq q_j\leq N_{j-1}<N_j$}.  If \notationrev{$N_{j-1}<t\leq N_j$}, then
$\lambda_t>3L/4$, and hence
\[
 F(\lambda_t-\alpha_j)\geq F(L/8).
\]
Take $\gamma_j=F(L/8)/F(L)$.  Writing $k=N_j\geq2$, the quantile
identity gives $(N-1)G_{\mathrm{sc}}(L)=k-1$.  Equations
\eqref{eq:ex-semicircle-count} and~\eqref{eq:ex-semicircle-dilation}
therefore imply
\begin{align*}
 F(L/8)&\geq\max\{1,8^{-3/2}(k-1)\}
          \geq\frac{k}{1+8^{3/2}},\\
 F(3L)&\leq1+3^{3/2}(k-1)\leq3^{3/2}k.
\end{align*}
These are the mass and enlargement bounds
\eqref{eq:ex-semicircle-ratios}.  Smooth filters with these cutoffs
can be chosen as in
\eqref{eq:first-filter-zero}--\eqref{eq:first-filter-plateau}.

Start with $L=2$ and repeat \eqref{eq:ex-semicircle-cutoffs} until
$N_{j-1}=1$.  Each positive next cutoff is at most $3L/4$.
Since $G_{\mathrm{sc}}(\lambda_2)=1/(N-1)$,
\eqref{eq:ex-semicircle-edge} gives
$\lambda_2\asymp(N-1)^{-2/3}$.  Thus the number of stages satisfies
\[
 m\leq1+\left\lceil
       \frac{\log(2/\lambda_2)}{\log(4/3)}\right\rceil
 =O(\log N)=O(n).
\]
The same cutoff rule applies to the final stage with \notationrev{$N_{j-1}=1$}, so
the iteration reaches the ground space.  Applying
Theorem~\ref{thm:first-empirical-preparation} with the uniform
bounds \eqref{eq:ex-semicircle-ratios} proves the preparation and
resource estimates.
\end{proof}

\begin{rem}[Scope of the semicircular example]
This example uses deterministic semicircular quantiles.  Extending the
result to random Hamiltonians with a semicircular limiting distribution
requires estimates on their finite-dimensional eigenvalue counts,
including near the lower spectral edge.  The ground-state gap in
\eqref{eq:ex-semicircle-spectrum} is of order $N^{-2/3}=2^{-2n/3}$,
and the transition width $\alpha_j-\alpha_j'=L/8$ at the last stage
has the same order.  The polynomial bounds above concern the evolution
time of the normalized generators and the number of sampled matrices;
the cost of implementing the filters and random operators must be
estimated separately.
\end{rem}

The conditions \eqref{eq:ex-first-method-assumptions} can be checked
numerically when no proof is available.  For the eigenvalues of GUE
matrices with $64\leq d\leq2048$, they hold along computed schedules of $9$
to $19$ stages, a number that grows like $\log d$.  The least mass stays
above $0.08$; it is attained at the lower edge, where the gaps of a finite
spectrum fluctuate and the semicircle law gives no control.  For $d=1024$
(Figure~\ref{fig:intro}), $18$ stages with $20$ Gaussian matrices each
reach the error $5.2\cdot10^{-7}$.

\subsection{The spin Hamiltonian}
\label{subsec:sm-spin}

On $\mathcal K=(\mathbb C^2)^{\otimes n}$, consider
\begin{equation}
    H=\sum_{\ell=1}^n\frac{Z_\ell+I}{2},
    \qquad Z=\operatorname{diag}(1,-1),
    \label{eq:sm-spin-hamiltonian}
\end{equation}
where $Z_\ell$ acts as $Z$ on the $\ell$-th tensor factor.
For $z\in\{0,1\}^n$,
$H|z\rangle=(n-\sum_{\ell=1}^n z_\ell)|z\rangle$.
Thus the spectrum is $\{0,1,\ldots,n\}$, and the spectral projections
$\Pi_j=\mathbf1_{\{j\}}(H)$ have dimensions $\binom nj$.
For the second method, take
\begin{equation}
\begin{split}
    &m_j=j,\qquad J=n,\qquad N=2^n,\qquad d_j=\binom nj,\\
    &N_j=\sum_{r=0}^j\binom nr,\qquad W=\{1,2,\ldots,n\}.
\end{split}
    \label{eq:sm-spin-clusters}
\end{equation}
Each cluster has width zero, and distinct elements of $W\cup\{0\}$
are at least one apart.  Hence any $0<\delta<1/6$ satisfies
\eqref{eq:sm-clusters} and \eqref{eq:sm-difference-separation}.
The filtered jumps are therefore
\begin{equation}
    a_k(g)=\sum_{j=k}^n\Pi_{j-k}g\Pi_j,\qquad 1\leq k\leq n,
    \label{eq:sm-spin-jumps}
\end{equation}
with $g$ distributed as in Section~\ref{subsec:sm-generator}.
The remaining dimension bound follows from
\begin{equation}
\begin{split}
    d_j&=\frac{n-j+1}{j}\binom n{j-1}
          \leq\frac{n-j+1}{j}N_{j-1},\\
    \gamma_j&=\frac{N_{j-1}}{N_{j-1}+d_j}
          \geq\frac{j}{n+1}\qquad(1\leq j\leq n).
\end{split}
    \label{eq:sm-binomial-ratio}
\end{equation}
Since $\gamma_1=1/(n+1)$, we have $\gamma=1/(n+1)$.
Proposition~\ref{prop:sm-sampling} and
Theorem~\ref{thm:sm-preparation} therefore apply with $n$ stages,
$M=O(n+\log(n/\zeta))$ Gaussian matrices, and $nM$ jump operators.
With probability at least $1-\zeta$, the resulting channel satisfies
\[
 \|\mathcal T(\rho)-|\Omega_Z\rangle\langle\Omega_Z|\|_1\leq\epsilon,
 \qquad \Omega_Z=|1\rangle^{\otimes n},
\]
for every density matrix $\rho$.  Under the stage normalization
\eqref{eq:sm-normalized-jumps}, \eqref{eq:sm-binomial-ratio} gives
\begin{equation}
 \sum_{j=1}^n t_j
 \leq2(n+1)\left(\sum_{j=1}^n\frac1j\right)\log\frac{4n}{\epsilon}
 \leq2(n+1)(1+\log n)\log\frac{4n}{\epsilon}.
 \label{eq:sm-number-time}
\end{equation}

\paragraph{Numerical results.}
The smallest mass is $\gamma_1=N_0/N_1=1/(n+1)$, at the last stage, which
therefore takes about a third of the total time.  For $n=18$, with $M=420$
Gaussian matrices, the total time is $\sim1034$ and the averaged dynamics ends
with the error $<\cdot10^{-9}$.  The bound \eqref{eq:sm-empirical-lower} holds
for every realization drawn up to $n=18$, and the least number of matrices
for which it holds grows linearly in $n$, close to the Marchenko-Pastur
estimate $n/(1-2^{-1/2})^2$.  Figures~\ref{fig:ex-traj-spin}
and~\ref{fig:ex-snap-spin} show $n=18$.  Since all frequencies act at every
stage, the upper clusters empty ahead of the schedule.

The first method also reaches the ground state of $\mathcal H_Z$, but every
admissible schedule has
$\min_j\gamma_j\leq\min_kN_{\lfloor(k+1)/2\rfloor}/N_k\approx e^{-0.19n}$, i.e. exhibits an exponential scaling,
because a jump from $\Pi_k$ cannot go below $\Pi_{\lfloor(k+1)/2\rfloor}$.

\begin{figure}[ht]
\centering
\includegraphics[width=\linewidth]{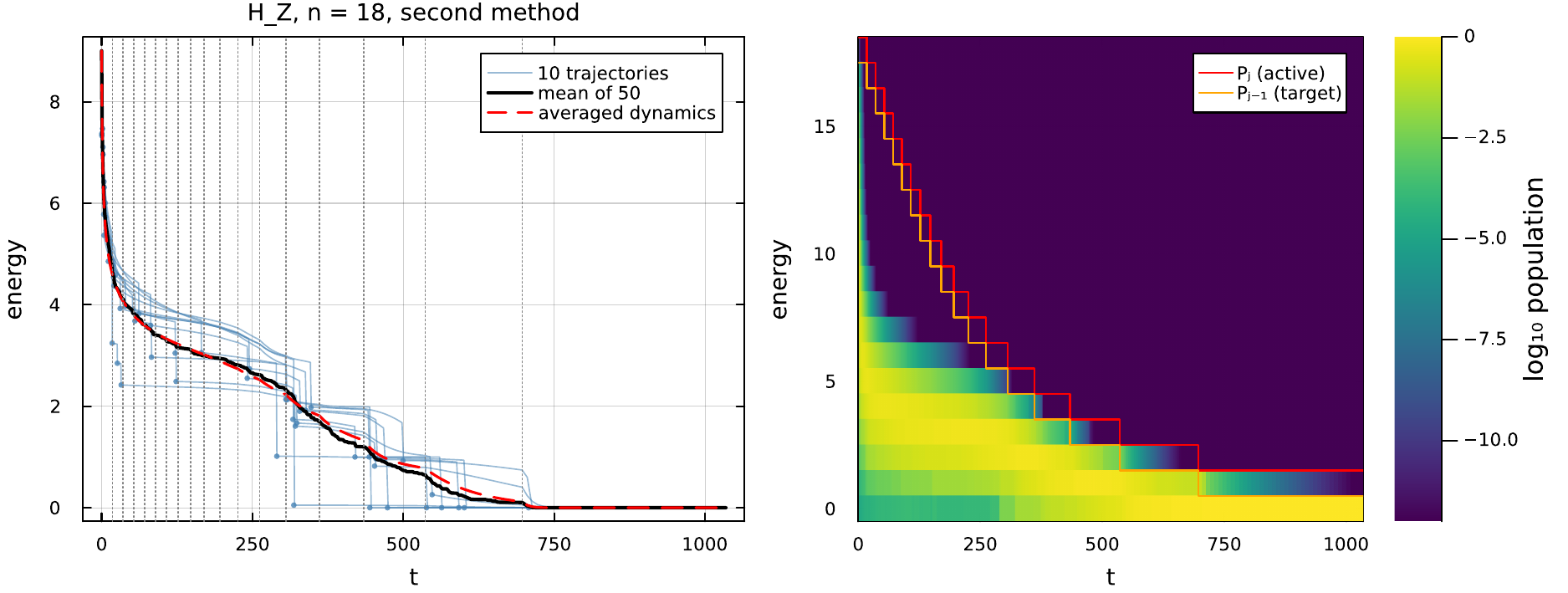}
\caption{Second method on $\mathcal H_Z$, $n=18$, $M=420$.  Left: energy of
quantum trajectories and of their mean over $50$ trajectories, against the
averaged dynamics.  Right: populations of the clusters.}
\label{fig:ex-traj-spin}
\end{figure}

\begin{figure}[ht]
\centering
\includegraphics[width=\linewidth]{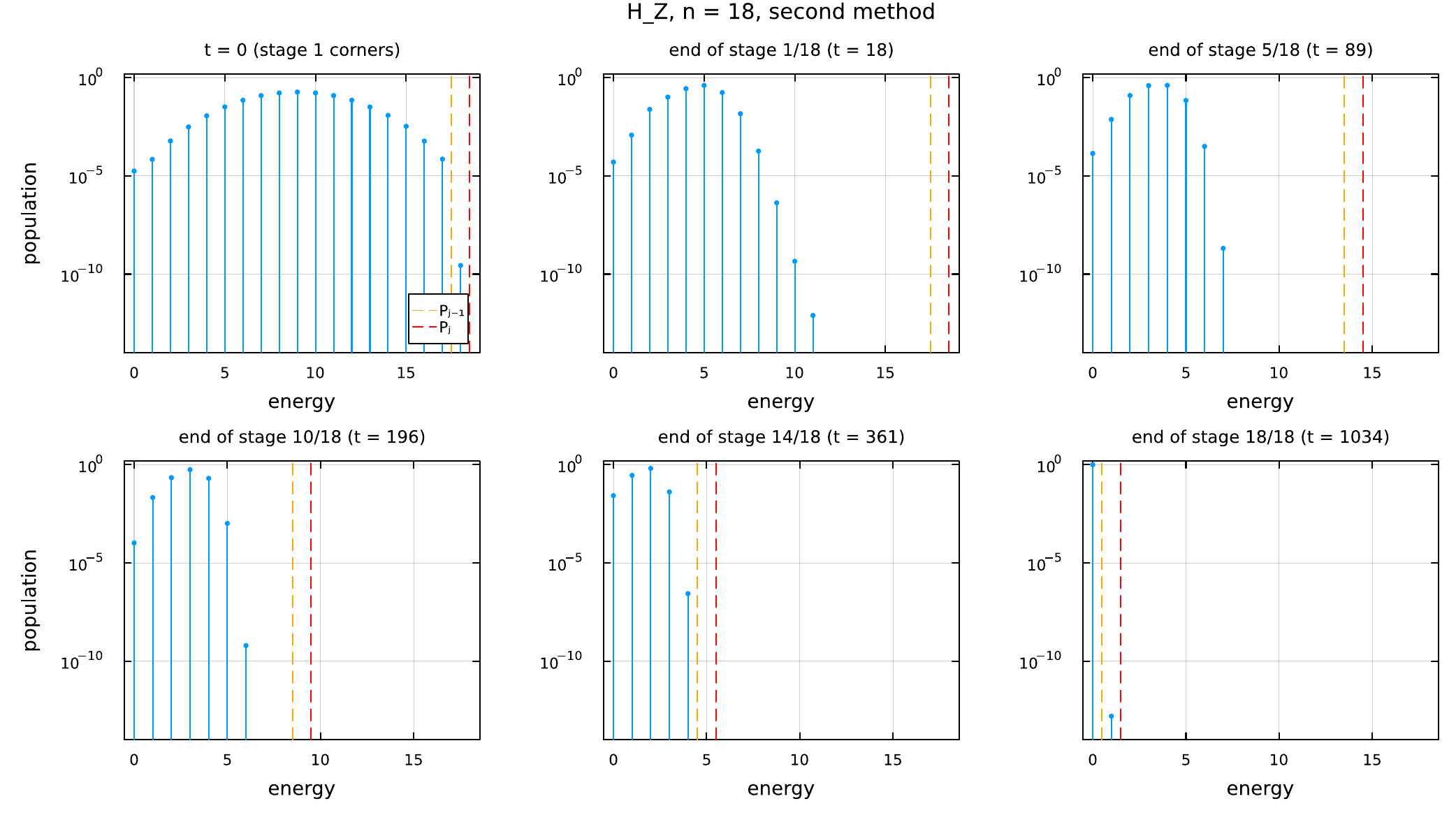}
\caption{The evolution of Figure~\ref{fig:ex-traj-spin}: populations of
the clusters at $t=0$ and at the end of five stages.}
\label{fig:ex-snap-spin}
\end{figure}

\subsection{Free fermions}
\label{subsec:sm-fermions}

On the $n$-mode fermionic Fock space, consider
\begin{equation}
 H=E_0I+\sum_{\ell=1}^n
              \varepsilon_\ell c_\ell^\dagger c_\ell,
 \qquad \varepsilon_\ell>0,
 \label{eq:sm-fermion-hamiltonian}
\end{equation}
where $c_\ell$ are the fermionic annihilation operators.  In the
orthonormal occupation basis $\{|u\rangle:u\in\{0,1\}^n\}$,
\[
 c_\ell^\dagger c_\ell|u\rangle=u_\ell|u\rangle,\qquad
 H|u\rangle=
 \left(E_0+\sum_{\ell=1}^n\varepsilon_\ell u_\ell\right)|u\rangle.
\]
These are the eigenvalues counted with multiplicity, and
$\Omega=|0,\ldots,0\rangle$ is the unique ground state.
For $\varepsilon_\ell=\Delta>0$, let $\Pi_j$ be the orthogonal
projection defined by
\[
 \mathrm{range}(\Pi_j)=\operatorname{span}
 \{|u\rangle:\textstyle\sum_{\ell=1}^n u_\ell=j\},\qquad 0\leq j\leq n.
\]
Then $H\Pi_j=(E_0+j\Delta)\Pi_j$, so
\begin{equation}
\begin{split}
    &m_j=E_0+j\Delta,\qquad J=n,\qquad d_j=\binom nj,\\
    &N_j=\sum_{r=0}^j\binom nr,\qquad
      W=\{\Delta,2\Delta,\ldots,n\Delta\}.
\end{split}
    \label{eq:sm-number-clusters}
\end{equation}
The clusters have width zero, and distinct elements of $W\cup\{0\}$
are at least $\Delta$ apart.  Thus any $0<\delta<\Delta/6$ satisfies
\eqref{eq:sm-clusters} and \eqref{eq:sm-difference-separation}, and
\[
 a_{k\Delta}(g)=\sum_{j=k}^n\Pi_{j-k}g\Pi_j,\qquad 1\leq k\leq n.
\]
The multiplicities give exactly the same $\gamma_j$ as in
\eqref{eq:sm-binomial-ratio}, hence $\gamma=1/(n+1)$.
Theorem~\ref{thm:sm-preparation} therefore prepares $\Omega$ to
trace-norm error $\epsilon$ from every initial state, with probability
at least $1-\zeta$, using the same stage count, matrix count, jump
count, and normalized time bound \eqref{eq:sm-number-time} as the
spin Hamiltonian in Section~\ref{subsec:sm-spin}.

\begin{rem}[Nearly equal coefficients]
Suppose
\begin{equation}
\varepsilon_\ell=\Delta+\eta_\ell,\qquad
 \sum_{\ell=1}^n|\eta_\ell|\leq\delta<\frac{\Delta}{6},\qquad\delta>0.
 \label{eq:sm-nearly-equal}
\end{equation}
For $\varepsilon_\ell=\Delta$, each cluster has one distinct eigenvalue;
the perturbations $\eta_\ell$ may produce several distinct eigenvalues
within the same cluster.  For $u\in\{0,1\}^n$ with $\sum_\ell u_\ell=j$,
\[
 \left|\sum_{\ell=1}^n\varepsilon_\ell u_\ell-j\Delta\right|
 =\left|\sum_{\ell=1}^n\eta_\ell u_\ell\right|\leq\delta.
\]
Thus $\Pi_j$ satisfy \eqref{eq:sm-clusters} with $m_j=E_0+j\Delta$,
and $\Delta>6\delta$ verifies Assumption~\ref{as:specluster-separation}.
The dimensions and $W$ in \eqref{eq:sm-number-clusters} are unchanged,
so Theorem~\ref{thm:sm-preparation} gives the same preparation bounds.
\end{rem}

\paragraph{Numerical results.}
The separation condition is sufficient with a wide margin.  For $n=6$ and
$\sum_\ell\eta_\ell$ up to four times the bound in
\eqref{eq:sm-nearly-equal}, the filtered jump operators still select the
sector blocks exactly, and the final error is unchanged.  The same clusters
arise for the open chain
$\mathcal H_Z+g\sum_{\ell<n}(X_\ell X_{\ell+1}+Y_\ell Y_{\ell+1})/2$, which
the Jordan-Wigner transformation~\cite{JordanWigner1928}
maps to free fermions~\cite{LiebSchultzMattis1961,OsborneNielsen2002} with
$\varepsilon_\ell=1+2g\cos(\pi\ell/(n+1))$.  For $n=7$ and $g=0.03$
(Figure~\ref{fig:ex-snap-xx}), each sector is a narrow band of eigenvalues, the masses and times are those of
$\mathcal H_Z$.

\begin{figure}[ht]
\centering
\includegraphics[width=\linewidth]{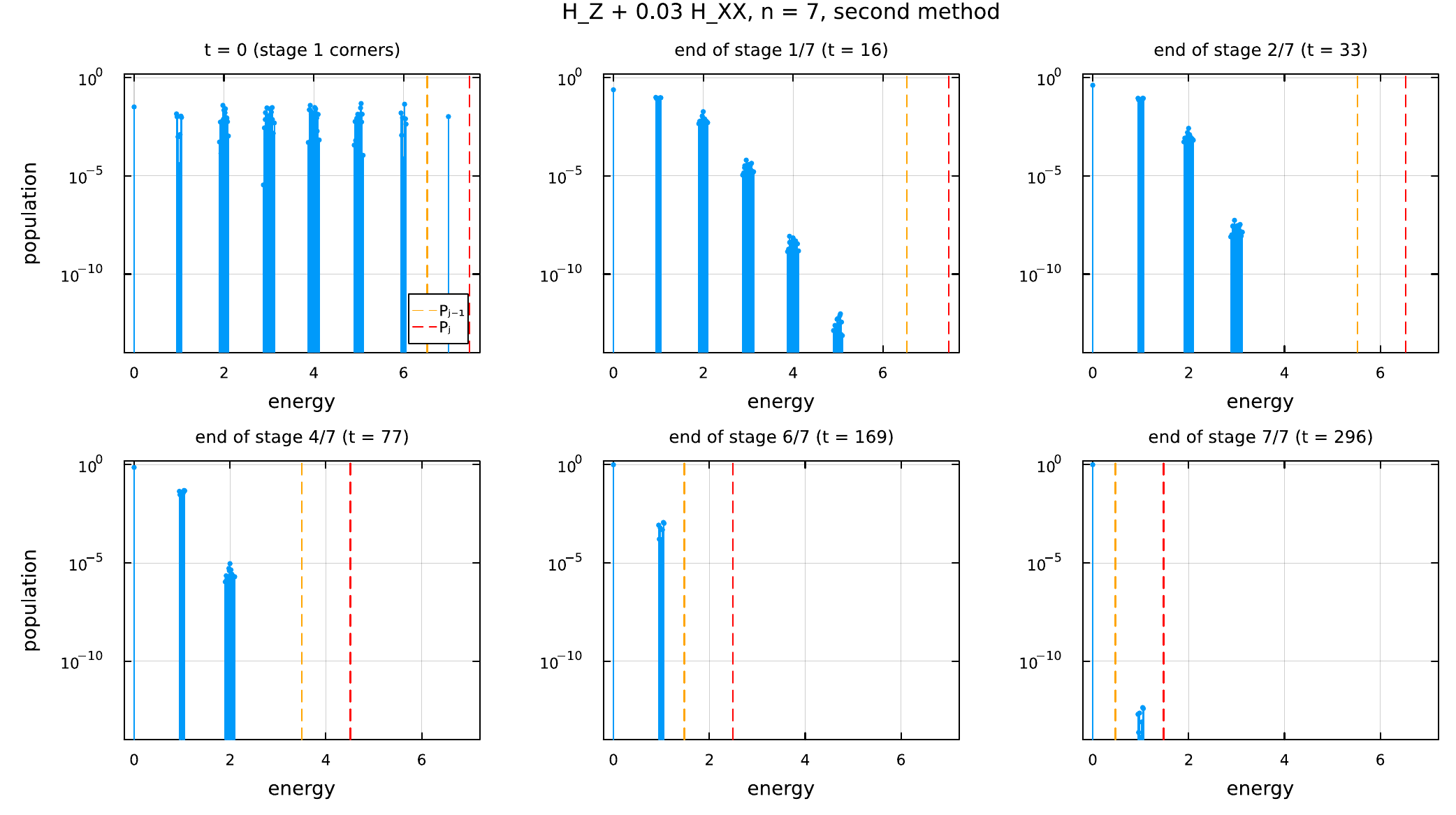}
\caption{Second method on the open chain with $g=0.03$, $n=7$, $M=164$,
from a random initial state: populations of the eigenvalues at $t=0$ and
at the end of the stages applied first, second, fourth, sixth and
seventh, with the boundaries of $P_{j-1}$ and $P_j$ drawn midway between
clusters.}
\label{fig:ex-snap-xx}
\end{figure}

\section*{Acknowledgements}
The authors thank Bert de Jong for valuable discussions. Part of the work was performed while PC, MJ and DL were visiting the Institute for Pure and Applied Mathematics (IPAM), which is supported by the National Science Foundation (Grant No. DMS-1925919).
P.C. acknowledges support from the Simons Foundation Travel Support for Mathematicians award MPS-TSM-00966604, and was partially supported by the National Science Foundation award DMS-2608677.
D.L was supported by the U.S. Department of Energy, Office of Science, Office of Advanced Scientific Computing Research through the Accelerated Research in Quantum Computing Program MACH-Q Project. 
MJ was partially supported by the National Science Foundation awards DMS-2247114 and NSF DMS-2556195.
  
  The authors acknowledge the use of GPT-6 Astra and Opus 5.5  for assistance with checking typos, inconsistencies, and numerical implementation of the examples in the initial research draft. The authors also acknowledge that the draft is still being actively polished. 

\appendix

\section{Probability estimates}
\label{app:probability}

We prove the estimates for the operators $\sum_ha_h^*a_h$ in
$\sum_h\mathcal M_{a_h}$ used in the two methods: the localized
estimate for the first method, and the Gaussian and decomposed
estimates for the second.  The normalizations and sampling assumptions
differ.  In particular,
Theorem~\ref{thm:sm-preparation} requires only a fixed lower bound on
$\Pi_jK_M\Pi_j$, rather than approximation to its mean with a prescribed small error.

\subsection{Common moment tools}
\label{app:moment-estimates}

For a random matrix $X$, write
\[
 \|X\|_{L_p(S_p)}=(\mathbb E\operatorname{Tr}|X|^p)^{1/p},
\]
where the trace is unnormalized.  We shall use the following two moment
inequalities in both methods.

\begin{lemma}[Gaussian coefficient estimate]
\label{Piq}
For deterministic rectangular matrices $B_i$, independent standard real
or complex Gaussian variables $g_i$, and $p\geq2$,
\[
 \left\|\sum_i g_iB_i\right\|_{L_p(S_p)}
 \leq C\sqrt p\max\left\{
  \left\|\left(\sum_iB_i^*B_i\right)^{1/2}\right\|_{S_p},
  \left\|\left(\sum_iB_iB_i^*\right)^{1/2}\right\|_{S_p}\right\}.
\]
\end{lemma}
This is the noncommutative Khintchine inequality
\cite{LustPiquard1986,LustPiquardPisier1991}; rectangular matrices can be
embedded as off-diagonal blocks of square matrices.

\begin{lemma}[Independent copies]
\label{Ro}
Let $X^{(1)},\ldots,X^{(M)}$ be independent copies of a random matrix $X$, and
let $p\geq2$.  Then
\begin{equation}
 \left\|\frac1M\sum_{\nu=1}^M(X^{(\nu)}-\mathbb E X)\right\|_{L_p(S_p)}
       \leq\frac{Cp}{\sqrt M}\|X\|_{L_p(S_p)}.
 \label{eq:common-copy-estimate}
\end{equation}
\end{lemma}
\begin{proof}
Set $Y=X-\mathbb E X$.  Jensen's inequality gives
$\|Y\|_{L_p(S_p)}\leq2\|X\|_{L_p(S_p)}$.  The triangle inequality
in $S_{p/2}$ and the scalar $L_p$--$L_2$ inequality imply
\[
 \|(\mathbb E Y^*Y)^{1/2}\|_{S_p},\quad
 \|(\mathbb E YY^*)^{1/2}\|_{S_p}
       \leq\|Y\|_{L_p(S_p)}.
\]
The noncommutative Rosenthal inequality~\cite{JungeZeng2013} therefore
bounds the norm of the centered sum by
$C(pM^{1/p}+\sqrt{pM})\|Y\|_{L_p(S_p)}$.
For a non-self-adjoint $Y$, apply the self-adjoint version to its
off-diagonal dilation; this yields both variance terms
$(\mathbb E Y^*Y)^{1/2}$ and $(\mathbb E YY^*)^{1/2}$.
Use $M^{1/p}\leq\sqrt M$ and divide by $M$.
\end{proof}

We also use the Gaussian--Haar comparison~\cite{MarcusPisier1981}.
For any seminorm $\mathfrak{s}$ on $D\times D$ complex matrices, a Haar unitary
$u$, an independent standard complex Gaussian matrix $g$, and $p\geq1$,
\begin{equation}
 (\mathbb E\mathfrak{s}(\sqrt D\,u)^p)^{1/p}
       \leq C(\mathbb E\mathfrak{s}(g)^p)^{1/p}.
 \label{MP}
\end{equation}
We apply this comparison to seminorms obtained by composing a linear
Schur multiplier with a Schatten norm.  For exact designs, the required
orders follow from the polynomial degrees in the moment estimates of
Sections~\ref{app:first-loss-proof} and~\ref{app:decomposed-gaussian-loss}.

\subsection{The localized estimate for the first method}
\label{app:first-loss-proof}

\begin{proof}[Proof of the probabilistic assertion of Lemma~\ref{lem:corner-support-current}]
Let $X^{(\nu)}$ be independent copies of $X$ and write
\[
 \Delta X_M=\frac1M\sum_{\nu=1}^M X^{(\nu)}-\mathbb E X.
\]
The deterministic part of the lemma applies to every $R$ satisfying
$Q_j\leq R\leq P_j$.  Since
$\|\mathcal Q_R^c\|_\diamond\leq2$,
$\|\mathcal Q_{P_j}\|_\diamond\leq1$, and left or right multiplication
by a matrix has diamond norm equal to its operator norm,
\begin{equation}
 \left\|\mathcal Q_R^c
  (\mathcal L_{j,M}-\mathbb E\mathcal L_a)\mathcal Q_{P_j}
                                                       \right\|_\diamond
       \leq4\|\Delta X_M\|_\infty.
 \label{eq:fm-four-times-loss-error}
\end{equation}
The factor $4$ accounts for the two terms of $\mathcal M_a$ and the bound $2$ on the
complementary corner map.  The estimate is uniform in $R$ on any event
on which $\Delta X_M$ is bounded.

Set $\chi=N_j^{\beta_j'}/N_j\geq1$. 
Introduce the two rectangular blocks
\[
 A=P_jaP_j,\qquad B=P_jaP_j^{\beta_j'},\qquad X=A^*B.
\]
For Gaussian samples and every $r\geq2$, the noncommutative Khintchine
inequality gives
\begin{equation}
 \|A\|_{L_r(S_r)}\leq C\sqrt r\,d^{1/r},
 \qquad
 \|B\|_{L_r(S_r)}\leq C\sqrt r\,d^{1/r}\sqrt\chi.
 \label{eq:fm-rectangular-moments}
\end{equation}
Each row and column variance of $A$ is at most one, since its
coefficients are bounded by $N_j^{-1/2}$.  For $B$, the column variances
are at most one and the row variances at most $N_j^{\beta_j'}/N_j=\chi$.
The covariance matrices have rank at most $d$, which gives the asserted
Schatten norm bounds.  The Schatten and scalar H\"older inequalities,
both with exponent $2p$, give
\begin{equation}
 \|X\|_{L_p(S_p)}
 \leq\|A\|_{L_{2p}(S_{2p})}\|B\|_{L_{2p}(S_{2p})}
 \leq Cp\,d^{1/p}\sqrt\chi.
 \label{eq:fm-product-moment}
\end{equation}
Independence of $A$ and $B$ is not required here.
Lemma~\ref{Ro} and \eqref{eq:fm-product-moment} give
\[
 \|\Delta X_M\|_{L_p(S_p)}\leq\frac{Cp}{\sqrt M}\|X\|_{L_p(S_p)}
   \leq Cp^2d^{1/p}\sqrt{\chi/M}.
\]
As $\|Y\|_\infty\leq\|Y\|_{S_p}$ for the unnormalized trace,
\begin{equation}
 \left(\mathbb E\|\Delta X_M\|_\infty^p\right)^{1/p}
      \leq Cp^2d^{1/p}\sqrt{\chi/M}.
 \label{eq:fm-normalized-loss-moment}
\end{equation}

The unitary case uses the same argument.  Apply the Gaussian--Haar
comparison \eqref{MP}, in dimension $D$, to the two seminorms
\[
 u\mapsto\|P_j\operatorname{Schur}_m(u)P_j\|_{S_{2p}},\qquad
 u\mapsto\|P_j\operatorname{Schur}_m(u)P_j^{\beta_j'}\|_{S_{2p}}.
\]
The factor $\sqrt{D/N_j}$ produces Gaussian coefficients normalized by
$N_j^{-1/2}$, so \eqref{eq:fm-rectangular-moments} at $r=2p$, and hence
\eqref{eq:fm-normalized-loss-moment}, hold up to a universal constant.
For integer $p$, the moments
$\operatorname{Tr}(A^*A)^p$ and $\operatorname{Tr}(B^*B)^p$ are
polynomials of degree $p$ in the entries of $u$ and degree $p$ in their
conjugates.  An exact design with $t\geq2p$ reproduces these moments and
the second moments determining $\mathbb E X$.  The same argument, with Lemma~\ref{Ro}, applies.

By \eqref{eq:fm-moment-order}, $d^{1/p}\leq e$.
Under \eqref{eq:fm-corrected-matrix-count}, increasing the universal
constant makes the right side of
\eqref{eq:fm-normalized-loss-moment} at most $\eps/(4e)$.
Markov's inequality gives
\[
 \mathbb P(\|\Delta X_M\|_\infty>\eps/4)
       \leq e^{-p}\leq\delta.
\]
Together with \eqref{eq:fm-four-times-loss-error}, this proves the
estimate simultaneously for all $Q_j\leq R\leq P_j$.

\end{proof}

\subsection{Bounds on \texorpdfstring{$\Pi_jK_M\Pi_j$}{the cluster blocks} for Gaussian samples}
\label{app:sm-gaussian}

\begin{proof}[Proof of Proposition~\ref{prop:sm-sampling}]

We use a covering argument; see
\cite[Lemmas~5.2 and~5.4 and the proof of Theorem~5.39]{Vershynin2012}.
Fix $j$.  In orthonormal bases of $P_{j-1}\mathcal{K}$ and
$\Pi_j\mathcal{K}$, each block $P_{j-1}g^{(\nu)}\Pi_j$ is an
$N_{j-1}\times d_j$ matrix.  Stack the $M$ such matrices vertically to form
$X_j$, with $R_j=MN_{j-1}$ rows and $d_j$ columns.  Its entries are
independent standard complex Gaussians.  Identifying the restriction to
$\Pi_j\mathcal{K}$ with a $d_j\times d_j$ matrix, \eqref{eq:sm-expected-loss}
becomes
\begin{equation}
    \frac{1}{N_j}\Pi_jK_M\Pi_j
       =\frac{X_j^*X_j}{MN_j}
       =\gamma_j\frac{X_j^*X_j}{R_j}.
    \label{eq:sm-gram-matrix}
\end{equation}
It therefore suffices to show that $X_j^*X_j/R_j$ is within $1/2$ of the
identity in operator norm.

For a fixed unit vector $v\in\mathbb C^{d_j}$, the entries of $X_jv$ are
independent standard complex Gaussians.  Thus $\|X_jv\|_2^2$ is a sum of
$R_j$ independent exponential random variables of mean one, where
$\|\cdot\|_2$ is the Euclidean vector norm.  For one such variable $Y$,
$\mathbb E e^{s(Y-1)}=e^{-s}/(1-s)$ for $s<1$.
Using $-s-\log(1-s)\leq s^2$ for $|s|\leq1/2$, exponential Markov bounds
on the upper and lower tails give the nonoptimal bound
\begin{equation}
    \mathbb P\left(\left|\frac{\|X_jv\|_2^2}{R_j}-1\right|>\frac14\right)
       \leq2e^{-R_j/128}.
    \label{eq:sm-fixed-vector-tail}
\end{equation}
For example, choosing $s=1/8$ for the upper tail and $s=-1/8$ for the lower
tail gives a smaller bound.

A $1/4$-net is a finite set of unit vectors such that every unit vector is
within Euclidean distance $1/4$ of a vector in the set.  Viewing
$\mathbb C^{d_j}$ as $\mathbb R^{2d_j}$, the usual volume comparison for
disjoint balls gives such a net $\mathcal{N}_j$ with
$|\mathcal{N}_j|\leq9^{2d_j}$.  For a Hermitian matrix $D$, a unit vector
$v$, and a net point $u$ with $\|v-u\|_2\leq1/4$,
\[
    |\langle v,Dv\rangle-\langle u,Du\rangle|
       \leq2\|v-u\|_2\|D\|_\infty
       \leq\tfrac12\|D\|_\infty.
\]
Taking the supremum over $v$ yields
$\|D\|_\infty\leq2\max_{u\in\mathcal{N}_j}|\langle u,Du\rangle|$.
Apply this with $D=X_j^*X_j/R_j-I_{d_j}$, where $I_{d_j}$ is the identity
on the source space.  A union bound over the net in
\eqref{eq:sm-fixed-vector-tail} gives
\begin{equation}
    \mathbb P\left(\left\|\frac{X_j^*X_j}{R_j}-I_{d_j}\right\|_\infty
                       >\frac12\right)
       \leq2\exp\left(2d_j\log9-\frac{MN_{j-1}}{128}\right).
    \label{eq:sm-gaussian-tail}
\end{equation}

The estimate is uniform in $j$: since
$d_j/N_{j-1}=\gamma_j^{-1}-1\leq\gamma^{-1}$ and $N_{j-1}\geq1$,
\eqref{eq:sm-sample-count} makes the right side of
\eqref{eq:sm-gaussian-tail} at most $\zeta/J$.  Another union bound, now
over the $J$ clusters, shows that, with probability at least $1-\zeta$,
$\tfrac12I_{d_j}\leq X_j^*X_j/R_j\leq\tfrac32I_{d_j}$ for every $j$.  By
\eqref{eq:sm-gram-matrix}, these are the two bounds in
\eqref{eq:sm-empirical-lower}.  Independence between stages is unnecessary;
the same samples are used at every stage.
\end{proof}

\subsection{Estimates for decomposed Gaussian jump operators}
\label{app:decomposed-gaussian-loss}

We now estimate the deviation of the empirical average of
$\sum_\ell a_\ell^*a_\ell$ from its mean.
Proposition~\ref{prop:sm-sampling} gives a sharper estimate in the Gaussian
case, where only a fixed lower bound on each cluster is required.
Here $M$ is the number of independent matrices and $\eta$ is the error
in this estimate, distinct from the final preparation error.
The traces in the moment norms are unnormalized.

Fix a projection $q$ of rank $s_0\geq1$, diagonal in the chosen basis.
For every allowed matrix entry $(r,s)$, let $\ell(r,s)$ be one label in a
finite set $\Lambda$, and let $|m_{rs}|\leq1$.  Define
\[
 a_\ell(g)=\sum_{\ell(r,s)=\ell}m_{rs}g_{rs}E_{rs},
 \qquad \ell\in\Lambda,
\]
where $g$ has independent standard complex Gaussian entries.  Assume that
for every realization and every label,
\begin{equation}
 a_\ell q=qa_\ell q,
 \qquad [a_\ell^*a_\ell,q]=0.
 \label{eq:decomposed-cutoff}
\end{equation}
Then the restriction of $\sum_\ell a_\ell^*a_\ell$ to the cutoff is the positive operator
\begin{equation}
 Z(g):=\sum_{\ell\in\Lambda}a_\ell(g)^*a_\ell(g)q
       =\sum_{\ell\in\Lambda}|a_\ell(g)q|^2.
 \label{eq:decomposed-positive-loss}
\end{equation}
Each allowed entry has a unique label, but jump operators with different labels
obtained from the same matrix need not be independent.

The following decompositions satisfy these assumptions.
\begin{enumerate}
\item For exact positive gaps, take
$\ell(r,s)=\lambda_s-\lambda_r>0$ and
$q=\mathbf{1}_{(-\infty,E]}(\mathcal{H})$.
At a fixed gap, different source energies have different target energies;
hence $a_\ell^*a_\ell$ commutes with $\mathcal{H}$ and with this $q$.
The cutoff must include complete energy eigenspaces: a rank cutoff that
splits a degenerate eigenspace need not commute with $a_\ell^*a_\ell$.
\item For jump operators between pairs of clusters, label an entry by its source
and target clusters $(i,j)$, $i<j$, and take $q$ to be a cumulative cluster projection.
Each $a_\ell^*a_\ell$ is then supported in its single source cluster.  Grouping these
pairs by the difference between their cluster centers also satisfies
\eqref{eq:decomposed-cutoff} under the injective matching established in
Lemma~\ref{lem:sm-block-shifts}.
\end{enumerate}
Pair labels and grouped frequency labels generally give different maps
$\Psi_{a_\ell}$.  The estimate for $Z(g)$ therefore does not identify the
corresponding Lindblad generators.

\begin{lemma}[Moment bound for the stacked jump operators]
\label{Piq3}
Set $A(g)=\sum_{\ell\in\Lambda}|\ell\rangle\otimes a_\ell(g)q$, viewed
as an operator from $q\mathcal{K}$ to
$\mathbb C^{|\Lambda|}\otimes q\mathcal{K}$.
For every $p\geq2$,
\begin{equation}
 \left(\mathbb E\operatorname{Tr}|A(g)|^p\right)^{1/p}
       \leq C\sqrt p\,s_0^{1/2+1/p},
 \qquad
 \left(\mathbb E\operatorname{Tr}Z(g)^p\right)^{1/p}
       \leq Cp\,s_0^{1+1/p}.
 \label{eq:stacked-gaussian-moments}
\end{equation}
In Section~\ref{app:decomposed-gaussian-loss}, $C$ denotes a universal
constant, which may change between
inequalities.
\end{lemma}
\begin{proof}
Write $A(g)=\sum_{r,s}g_{rs}B_{rs}$, where
$B_{rs}=m_{rs}|\ell(r,s)\rangle\otimes E_{rs}$ and both indices belong
to the cutoff.  This last assertion follows from
$a_\ell q=qa_\ell q$.  The column covariance is
\[
 \sum_{r,s}B_{rs}^*B_{rs}
   =\sum_s\left(\sum_r|m_{rs}|^2\right)E_{ss}.
\]
There are at most $s_0$ diagonal entries, each bounded by $s_0$.
The row covariance is diagonal in the output basis, with entries
\[
 t_{\ell,r}=\sum_{s:\,\ell(r,s)=\ell}|m_{rs}|^2.
\]
Since each entry has one label and $p/2\geq1$,
\[
 \sum_{r,\ell}t_{\ell,r}^{p/2}
 \leq\sum_r\left(\sum_\ell t_{\ell,r}\right)^{p/2}
 \leq s_0^{1+p/2}.
\]
Thus both covariance square roots have Schatten $p$ norm at most
$s_0^{1/2+1/p}$.  Applying Lemma~\ref{Piq} to these rectangular
coefficients gives the first bound.  The identity $Z=A^*A$ and the first bound with
$2p$ in place of $p$ give the second.
\end{proof}

\begin{prop}[Decomposed estimate]
\label{prop:decomposed-gaussian-loss}
Let $g^{(1)},\ldots,g^{(M)}$ be independent copies and set
\[
 \widehat K_{M,q}=\frac{1}{Ms_0}\sum_{\nu=1}^M Z(g^{(\nu)}),
 \qquad K_q^{\rm av}=\frac{1}{s_0}\mathbb E Z(g).
\]
For every $p\geq2$,
\begin{equation}
 \left(\mathbb E\left\|
       \widehat K_{M,q}-K_q^{\rm av}
                     \right\|_\infty^p\right)^{1/p}
       \leq\frac{Cp^2s_0^{1/p}}{\sqrt M}.
 \label{eq:decomposed-normalized-moment}
\end{equation}
In particular, for $0<\eta,\zeta<1$, it suffices to take
\begin{equation}
 p\geq\max\{2,\log s_0,\log(1/\zeta)\},
 \qquad M\geq C\eta^{-2}p^4
 \label{eq:decomposed-matrix-count}
\end{equation}
to obtain
$\mathbb P(\|\widehat K_{M,q}-K_q^{\rm av}\|_\infty>\eta)\leq\zeta$.
For $R$ prescribed cutoffs, the estimate holds simultaneously for the
same matrices if $\log s_0$ is replaced by the largest logarithm of a
cutoff rank and $\log(1/\zeta)$ by $\log(R/\zeta)$.
\end{prop}
\begin{proof}
Lemma~\ref{Piq3} gives $\|Z\|_{L_p(S_p)}\leq Cp\,s_0^{1+1/p}$, and
Lemma~\ref{Ro}, applied to the copies $Z(g^{(\nu)})$, gives
\[
 \bigl\|\widehat K_{M,q}-K_q^{\rm av}\bigr\|_{L_p(S_p)}
   \leq\frac{Cp}{\sqrt M}\,\frac{\|Z\|_{L_p(S_p)}}{s_0}
   \leq\frac{Cp^2s_0^{1/p}}{\sqrt M}.
\]
The operator norm is bounded by the unnormalized Schatten $p$ norm, which
proves \eqref{eq:decomposed-normalized-moment}.
Under \eqref{eq:decomposed-matrix-count}, $s_0^{1/p}\leq e$;
increasing the universal constant makes this moment bound at most
$\eta/e$.  Markov's inequality then gives failure probability at most
$e^{-p}\leq\zeta$.  The simultaneous assertion follows by a union bound;
independence between the cutoff events is unnecessary.
\end{proof}

The mean is
\[
 K_q^{\rm av}=
 \frac{1}{s_0}\sum_{s\in q\mathcal{K}}
            \left(\sum_{r\in q\mathcal{K}}|m_{rs}|^2\right)E_{ss},
\]
where the sums refer to basis indices in the cutoff and omitted entries
have coefficient zero.  Both the empirical average of $Z(g)$ and its mean are divided
by $s_0$ in \eqref{eq:decomposed-normalized-moment}.  Neither an upper reach
parameter nor a factor $|\Lambda|$ enters the matrix count.
An estimate with tolerance $\eta$ gives a positive lower bound only on a
subspace where the normalized mean $K_q^{\rm av}$ exceeds $\eta$.
For Gaussian cluster blocks on which the filter equals one,
Proposition~\ref{prop:sm-sampling} gives the required fixed lower bound
with the sharper count
$O(\gamma^{-1}+\log(J/\zeta))$.

\begin{cor}[Haar matrices and exact designs]
\label{cor:decomposed-haar-loss}
The normalized estimate \eqref{eq:decomposed-normalized-moment} and its
consequences also hold, up to a universal constant, for independent
matrices $g^{(\nu)}=\sqrt d\,u^{(\nu)}$ with $u^{(\nu)}$ Haar distributed.
For an even integer $p\geq2$, they hold for independent samples from an
exact unitary $t$-design with $t\geq2p$.
\end{cor}
\begin{proof}
The cutoff identities are deterministic.  Apply the Gaussian--Haar
comparison \eqref{MP} to the seminorm
$g\mapsto\|A(g)\|_{S_{2p}}$.  This gives the second estimate in
\eqref{eq:stacked-gaussian-moments} for $g=\sqrt d\,u$ as well.
The proof of Proposition~\ref{prop:decomposed-gaussian-loss} then uses only
this moment bound and independence between the $M$ matrices.
The covariance identity
$\mathbb E(\overline{g_{rs}}g_{r't})=\delta_{rr'}\delta_{st}$ gives
the same mean $\mathbb EZ$ as in the Gaussian model.
For integer $p$, $\operatorname{Tr}Z(\sqrt d\,u)^p$ is a polynomial
of degree $p$ in the entries of $u$ and degree $p$ in their conjugates.
An exact design of the stated order reproduces this moment and the mean
$\mathbb EZ$, so the same proof applies.
\end{proof}

For an approximate design, one must also bound the moment error, taking
account of how the polynomial coefficients depend on the dimension, and
include this error in the estimate.
The corollary does not extend the sharper bound for Gaussian Gram matrices in
Proposition~\ref{prop:sm-sampling} to Haar matrices or designs.

\section{Numerical methods}
\label{app:numerics}

In every example the Hamiltonian is diagonal in a known basis.
The law of
the Gaussian matrices is invariant under unitary conjugation, so we work in
the eigenbasis and the results depend on the spectrum only. The entries are
real, and all runs use $\epsilon=10^{-2}$ and the stage times of the
preparation theorems.

\paragraph{Stages.}
For the first method, the corners are spectral projections computed from the
known spectrum histogram, and we implement a greedy algorithm for the design of the stage cutoffs as follows.
For an active corner $P_j$, a target $P_{j-1}$ and an output
projection $Q_j$, we take the least cutoff
$\alpha_j'=\max\{\lambda_{q_j},\lambda_{N_j}-\lambda_{q_j+1}\}$ allowed by
Proposition~\ref{prop:boxass}, $\alpha_j=1.1\,\alpha_j'$ and
$\beta_j=\lambda_{N_j}$, and the mass of Proposition~\ref{prop:gammaass}.
Dynamic programming over the distinct eigenvalues gives, in $O(L^2)$
evaluations for $L$ eigenvalues, the schedule that maximizes $\min_j\gamma_j$
and the schedule that minimizes $\sum_j\gamma_j^{-1}$.  The filters have smooth, comopactly supported
ramps built from the $t\mapsto e^{-1/t}$, and each stage uses $20$ Gaussian
matrices. 
The second method follows Section~\ref{sec:second-method}, with
about twice the Marchenko-Pastur estimate $n/(1-2^{-1/2})^2$ of matrices.

\paragraph{Evolution.}
We implemented several timestepping approaches to simulate the quantum evolution.
First, a direct approach in which the density matrix is integrated by the fourth-order Runge-Kutta scheme, with a number of degrees of freedom scaling as $N^2$.
For the second method only the block on the active corner is evolved; the
dropped weight changes the final error by at most
$4\sum_j(\operatorname{Tr}(I-P_j)\rho_j)^{1/2}$, which we check to be
negligible. 
A second approach integrates quantum trajectories according to the waiting-time form of the Lindblad
unraveling~\cite{DalibardCastinMolmer1992,DumZollerRitsch1992,MolmerCastinDalibard1993,Carmichael1993};
see~\cite{PlenioKnight1998} for a review. A pure state evolves under
$\dot\psi=-R_j\psi$, $R_j$ being the loss operator, until $\|\psi\|^2$
reaches a uniform random threshold, located exactly within the step; a jump
operator is then chosen with probability proportional to
$\|b\psi\|^2$.  The averaged dynamics is the population chain
$\dot p=2(W-\operatorname{Diag}(\mathbf 1^{\mathsf T}W))p$ with
$W_{rt}=\mathbb E|a_{rt}|^2$, obtained by averaging the generator over the
Gaussian entries, and is integrated exactly. With this approach, the number of degrees of freedom scales as $N$.

\paragraph{Spin Hamiltonian for large $n$.}
For the second method on $\mathcal H_Z$ we follow trajectories without
forming the Gaussian matrices.  With $G^{(k)}$ the stack of the blocks
$\Pi_lg^{(\nu)}\Pi_k$, $l<k$, the loss operator of stage $j$ is
$\bigoplus_kB_k/(MN_j)$ with $B_k=G^{(k)*}G^{(k)}$.  The $B_k$ are
independent Wishart matrices, which we draw directly and diagonalize once,
so that the evolution between jumps is exact.  A jump needs $G^{(k)}\psi_k$
only; writing $G^{(k)}=O_kB_k^{1/2}$ with an isometry $O_k$ that is Haar
distributed and independent of $B_k$, we sample $O_k$ only at the vectors
queried along the trajectories.  This solver agrees with the dense one for
$n\leq6$, and for $n=18$ a realization takes four hours and $150$ gigabytes.
The code is the Julia package \texttt{Lindblad}, commit \texttt{fc7d574}.

\bibliographystyle{abbrvurl}
\bibliography{main_new,References/ground_state_preparation_selected}

@article{DengisKoenigPastawski2014,
  author = {Dengis, John and K{\"o}nig, Robert and Pastawski, Fernando},
  title = {An optimal dissipative encoder for the toric code},
  journal = {New Journal of Physics},
  volume = {16},
  pages = {013023},
  year = {2014},
  doi = {10.1088/1367-2630/16/1/013023},
  eprint = {1310.1036},
  archivePrefix = {arXiv},
  url = {https://arxiv.org/abs/1310.1036}
}

@article{DingLandauLiLinZhang2026,
  author = {Ding, Zhiyan and Landau, Zeph and Li, Bowen and Lin, Lin and Zhang, Ruizhe},
  title = {Polynomial-time preparation of low-temperature {Gibbs} states for two-dimensional toric code},
  journal = {Journal of Mathematical Physics},
  volume = {67},
  pages = {052201},
  year = {2026},
  doi = {10.1063/5.0302877},
  eprint = {2410.01206},
  archivePrefix = {arXiv},
  url = {https://arxiv.org/abs/2410.01206}
}

@article{KoenigPastawski2014,
  author = {K{\"o}nig, Robert and Pastawski, Fernando},
  title = {Generating topological order: No speedup by dissipation},
  journal = {Physical Review B},
  volume = {90},
  pages = {045101},
  year = {2014},
  doi = {10.1103/PhysRevB.90.045101},
  eprint = {1310.1037},
  archivePrefix = {arXiv},
  url = {https://arxiv.org/abs/1310.1037}
}

@book{SteinShakarchi2003,
  author    = {Stein, Elias M. and Shakarchi, Rami},
  title     = {Fourier Analysis: An Introduction},
  series    = {Princeton Lectures in Analysis},
  volume    = {1},
  publisher = {Princeton University Press},
  address   = {Princeton, NJ},
  year      = {2003}
}

@misc{yu2026bathassistedcoolingresets,
      title={Bath-assisted cooling without resets}, 
      author={Xie-Hang Yu and Zherui Chen and Lin Lin},
      year={2026},
      eprint={2609.40328},
      archivePrefix={arXiv},
      primaryClass={quant-ph},
      url={https://arxiv.org/abs/2609.40328}, 
}

@article{Kraus2008,
  author = {Kraus, B. and B{\"u}chler, H. P. and Diehl, S. and
            Kantian, A. and Micheli, A. and Zoller, P.},
  title = {Preparation of entangled states by quantum {Markov} processes},
  journal = {Physical Review A},
  volume = {78},
  number = {4},
  pages = {042307},
  year = {2008},
  doi = {10.1103/PhysRevA.78.042307},
  eprint = {0803.1463},
  archivePrefix = {arXiv}
}

@article{VerstraeteWolfCirac2009,
  author = {Verstraete, Frank and Wolf, Michael M. and Cirac, J. Ignacio},
  title = {Quantum computation and quantum-state engineering driven by dissipation},
  journal = {Nature Physics},
  volume = {5},
  number = {9},
  pages = {633--636},
  year = {2009},
  doi = {10.1038/nphys1342},
  eprint = {0803.1447},
  archivePrefix = {arXiv}
}

@article{DingChenLin2024,
  author = {Ding, Zhiyan and Chen, Chi-Fang and Lin, Lin},
  title = {Single-ancilla ground state preparation via {Lindbladians}},
  journal = {Physical Review Research},
  volume = {6},
  number = {3},
  pages = {033147},
  year = {2024},
  doi = {10.1103/PhysRevResearch.6.033147},
  eprint = {2308.15676},
  archivePrefix = {arXiv}
}

@article{DingLiLin2024,
  author = {Ding, Zhiyan and Li, Xiantao and Lin, Lin},
  title = {Simulating Open Quantum Systems Using {Hamiltonian} Simulations},
  journal = {PRX Quantum},
  volume = {5},
  number = {2},
  pages = {020332},
  year = {2024},
  doi = {10.1103/PRXQuantum.5.020332}
}

@article{Li2025Dissipative,
  author = {Li, Hao-En and Zhan, Yongtao and Lin, Lin},
  title = {Dissipative ground state preparation in ab initio electronic structure theory},
  journal = {npj Quantum Information},
  volume = {11},
  pages = {183},
  year = {2025},
  doi = {10.1038/s41534-025-01124-8},
  eprint = {2411.01470},
  archivePrefix = {arXiv}
}

@article{Zhan2026Rapid,
  author = {Zhan, Yongtao and Ding, Zhiyan and Huhn, Jakob and
            Gray, Johnnie and Preskill, John and Chan, Garnet Kin-Lic and Lin, Lin},
  title = {Rapid Quantum Ground State Preparation via Dissipative Dynamics},
  journal = {Physical Review X},
  volume = {16},
  number = {1},
  pages = {011004},
  year = {2026},
  doi = {10.1103/wzb3-dbg9},
  eprint = {2503.15827},
  archivePrefix = {arXiv}
}

@article{LiebSchultzMattis1961,
  author = {Lieb, Elliott and Schultz, Theodore and Mattis, Daniel},
  title = {Two soluble models of an antiferromagnetic chain},
  journal = {Annals of Physics},
  volume = {16},
  number = {3},
  pages = {407--466},
  year = {1961},
  doi = {10.1016/0003-4916(61)90115-4}
}

@article{Dankert2009,
  author = {Dankert, Christoph and Cleve, Richard and Emerson, Joseph and Livine, Etera},
  title = {Exact and approximate unitary 2-designs and their application to fidelity estimation},
  journal = {Physical Review A},
  volume = {80},
  number = {1},
  pages = {012304},
  year = {2009},
  doi = {10.1103/PhysRevA.80.012304},
  eprint = {quant-ph/0606161},
  archivePrefix = {arXiv}
}

@article{JungeZeng2013,
  author = {Junge, Marius and Zeng, Qiang},
  title = {Noncommutative {Bennett} and {Rosenthal} inequalities},
  journal = {The Annals of Probability},
  volume = {41},
  number = {6},
  pages = {4287--4316},
  year = {2013},
  doi = {10.1214/12-AOP771},
  eprint = {1111.1027},
  archivePrefix = {arXiv}
}

@article{LustPiquard1986,
  author = {Lust-Piquard, Fran{\c c}oise},
  title = {In{\'e}galit{\'e}s de {Khintchine} dans {$C_p$} {$(1<p<\infty)$}},
  journal = {Comptes Rendus de l'Acad{\'e}mie des Sciences.
             S{\'e}rie I. Math{\'e}matique},
  volume = {303},
  number = {7},
  pages = {289--292},
  year = {1986}
}

@article{LustPiquardPisier1991,
  author = {Lust-Piquard, Fran{\c c}oise and Pisier, Gilles},
  title = {Non commutative {Khintchine} and {Paley} inequalities},
  journal = {Arkiv f{\"o}r Matematik},
  volume = {29},
  number = {2},
  pages = {241--260},
  year = {1991},
  doi = {10.1007/BF02384340}
}

@book{MarcusPisier1981,
  author = {Marcus, Michael B. and Pisier, Gilles},
  title = {Random {Fourier} Series with Applications to Harmonic Analysis},
  series = {Annals of Mathematics Studies},
  volume = {101},
  publisher = {Princeton University Press},
  address = {Princeton, NJ},
  year = {1981},
  isbn = {9780691082929},
  url = {https://doi.org/10.1515/9781400881536},
  note = {Linked electronic reissue published in 2016}
}

@incollection{Vershynin2012,
  author = {Vershynin, Roman},
  title = {Introduction to the non-asymptotic analysis of random matrices},
  booktitle = {Compressed Sensing: Theory and Applications},
  editor = {Eldar, Yonina C. and Kutyniok, Gitta},
  publisher = {Cambridge University Press},
  pages = {210--268},
  year = {2012},
  doi = {10.1017/CBO9780511794308.006},
  eprint = {1011.3027},
  archivePrefix = {arXiv}
}

@article{PoyatosCiracZoller1996,
  author = {Poyatos, J. F. and
            Cirac, J. I. and
            Zoller, P.},
  title = {Quantum reservoir engineering with laser cooled trapped ions},
  journal = {Physical Review Letters},
  volume = {77},
  number = {23},
  pages = {4728--4731},
  year = {1996},
  doi = {10.1103/PhysRevLett.77.4728}
}

@article{Diehl2008,
  author = {Diehl, S. and
            Micheli, A. and
            Kantian, A. and
            Kraus, B. and
            B{\"u}chler, H. P. and
            Zoller, P.},
  title = {Quantum states and phases in driven open quantum systems with cold atoms},
  journal = {Nature Physics},
  volume = {4},
  number = {11},
  pages = {878--883},
  year = {2008},
  doi = {10.1038/nphys1073},
  eprint = {0803.1482},
  archivePrefix = {arXiv}
}

@article{AbramsLloyd1999,
  author = {Abrams, Daniel S. and Lloyd, Seth},
  title = {Quantum Algorithm Providing Exponential Speed Increase for Finding Eigenvalues and Eigenvectors},
  journal = {Physical Review Letters},
  volume = {83},
  number = {24},
  pages = {5162--5165},
  year = {1999},
  doi = {10.1103/PhysRevLett.83.5162},
  eprint = {quant-ph/9807070},
  archivePrefix = {arXiv}
}

@article{PoulinWocjan2009,
  author = {Poulin, David and
            Wocjan, Pawel},
  title = {Preparing ground states of quantum many-body systems on a quantum computer},
  journal = {Physical Review Letters},
  volume = {102},
  number = {13},
  pages = {130503},
  year = {2009},
  doi = {10.1103/PhysRevLett.102.130503},
  eprint = {0809.2705},
  archivePrefix = {arXiv}
}

@article{GeTuraCirac2019,
  author = {Ge, Yimin and
            Tura, Jordi and
            Cirac, J. Ignacio},
  title = {Faster ground state preparation and high-precision ground energy estimation with fewer qubits},
  journal = {Journal of Mathematical Physics},
  volume = {60},
  number = {2},
  pages = {022202},
  year = {2019},
  doi = {10.1063/1.5027484},
  eprint = {1712.03193},
  archivePrefix = {arXiv}
}

@article{LinTong2020,
  author = {Lin, Lin and
            Tong, Yu},
  title = {Near-optimal ground state preparation},
  journal = {Quantum},
  volume = {4},
  pages = {372},
  year = {2020},
  doi = {10.22331/q-2020-12-14-372},
  eprint = {2002.12508},
  archivePrefix = {arXiv}
}

@article{DongLinTong2022,
  author = {Dong, Yulong and
            Lin, Lin and
            Tong, Yu},
  title = {Ground-state preparation and energy estimation on early fault-tolerant quantum computers via quantum eigenvalue transformation of unitary matrices},
  journal = {PRX Quantum},
  volume = {3},
  number = {4},
  pages = {040305},
  year = {2022},
  doi = {10.1103/PRXQuantum.3.040305},
  eprint = {2204.05955},
  archivePrefix = {arXiv}
}

@misc{FarhiGoldstoneGutmannSipser2000,
  author = {Farhi, Edward and Goldstone, Jeffrey and Gutmann, Sam and Sipser, Michael},
  title = {Quantum Computation by Adiabatic Evolution},
  year = {2000},
  eprint = {quant-ph/0001106},
  archivePrefix = {arXiv},
  url = {https://arxiv.org/abs/quant-ph/0001106}
}

@article{JansenRuskaiSeiler2007,
  author = {Jansen, Sabine and Ruskai, Mary-Beth and Seiler, Ruedi},
  title = {Bounds for the adiabatic approximation with applications to quantum computation},
  journal = {Journal of Mathematical Physics},
  volume = {48},
  number = {10},
  pages = {102111},
  year = {2007},
  doi = {10.1063/1.2798382},
  eprint = {quant-ph/0603175},
  archivePrefix = {arXiv}
}

@article{AlbashLidar2018,
  author = {Albash, Tameem and
            Lidar, Daniel A.},
  title = {Adiabatic quantum computation},
  journal = {Reviews of Modern Physics},
  volume = {90},
  number = {1},
  pages = {015002},
  year = {2018},
  doi = {10.1103/RevModPhys.90.015002},
  eprint = {1611.04471},
  archivePrefix = {arXiv}
}

@article{Lee2023,
  author = {Lee, Seunghoon and
            Lee, Joonho and
            Zhai, Huanchen and
            Tong, Yu and
            Dalzell, Alexander M. and
            Kumar, Ashutosh and
            others},
  title = {Evaluating the evidence for exponential quantum advantage in ground-state quantum chemistry},
  journal = {Nature Communications},
  volume = {14},
  number = {1},
  pages = {1952},
  year = {2023},
  doi = {10.1038/s41467-023-37587-6},
  eprint = {2208.02199},
  archivePrefix = {arXiv}
}

@article{Temme2011,
  author = {Temme, K. and
            Osborne, T. J. and
            Vollbrecht, K. G. and
            Poulin, D. and
            Verstraete, F.},
  title = {Quantum {Metropolis} sampling},
  journal = {Nature},
  volume = {471},
  number = {7336},
  pages = {87--90},
  year = {2011},
  doi = {10.1038/nature09770},
  eprint = {0911.3635},
  archivePrefix = {arXiv}
}

@article{ChenKastoryanoBrandaoGilyen2025,
  author = {Chen, Chi-Fang and
            Kastoryano, Michael and
            Brand{\~a}o, Fernando G. S. L. and
            Gily{\'e}n, Andr{\'a}s},
  title = {Efficient quantum thermal simulation},
  journal = {Nature},
  volume = {646},
  number = {8085},
  pages = {561--566},
  year = {2025},
  doi = {10.1038/s41586-025-09583-x},
  eprint = {2303.18224},
  archivePrefix = {arXiv}
}

@misc{ChenKastoryanoGilyen2023,
  author = {Chen, Chi-Fang and
            Kastoryano, Michael J. and
            Gily{\'e}n, Andr{\'a}s},
  title = {An efficient and exact noncommutative quantum {Gibbs} sampler},
  year = {2023},
  eprint = {2311.09207},
  archivePrefix = {arXiv},
  primaryClass = {quant-ph}
}

@article{DingZhanPreskillLin2026,
  author = {Ding, Zhiyan and
            Zhan, Yongtao and
            Preskill, John and
            Lin, Lin},
  title = {Simple and efficient end-to-end quantum thermal and ground state preparation},
  journal = {Nature Physics},
  year = {2026},
  doi = {10.1038/s41567-026-03389-y},
  eprint = {2508.05703},
  archivePrefix = {arXiv}
}

@article{Seki2026,
  author = {Seki, Kazuhiro and
            Kikuchi, Yuta and
            Hayata, Tomoya and
            Yunoki, Seiji},
  title = {Dissipative ground-state preparation of a quantum spin chain on a trapped-ion quantum computer},
  journal = {Physical Review Applied},
  volume = {26},
  number = {1},
  pages = {014062},
  year = {2026},
  doi = {10.1103/rjdv-bxgq},
  eprint = {2601.08137},
  archivePrefix = {arXiv}
}

@misc{Cubitt2023,
  author = {Cubitt, Toby S.},
  title = {Dissipative ground state preparation and the {Dissipative Quantum Eigensolver}},
  year = {2023},
  eprint = {2303.11962},
  archivePrefix = {arXiv},
  primaryClass = {quant-ph}
}

@article{PollaHerasymenkoOBrien2021,
  author = {Polla, Stefano and
            Herasymenko, Yaroslav and
            O'Brien, Thomas E.},
  title = {Quantum digital cooling},
  journal = {Physical Review A},
  volume = {104},
  number = {1},
  pages = {012414},
  year = {2021},
  doi = {10.1103/PhysRevA.104.012414},
  eprint = {1909.10538},
  archivePrefix = {arXiv}
}

@article{Matthies2024,
  author = {Matthies, Anne and
            Rudner, Mark and
            Rosch, Achim and
            Berg, Erez},
  title = {Programmable adiabatic demagnetization for systems with trivial and topological excitations},
  journal = {Quantum},
  volume = {8},
  pages = {1505},
  year = {2024},
  doi = {10.22331/q-2024-10-23-1505},
  eprint = {2210.17256},
  archivePrefix = {arXiv}
}

@article{MartiMansurogluHartmann2025,
  author = {Marti, Lucas and
            Mansuroglu, Refik and
            Hartmann, Michael J.},
  title = {Efficient quantum cooling algorithm for fermionic systems},
  journal = {Quantum},
  volume = {9},
  pages = {1635},
  year = {2025},
  doi = {10.22331/q-2025-02-18-1635},
  eprint = {2403.14506},
  archivePrefix = {arXiv}
}

@article{Temme2010,
  author = {Temme, K. and
            Kastoryano, M. J. and
            Ruskai, M. B. and
            Wolf, M. M. and
            Verstraete, F.},
  title = {The {$\chi^2$}-divergence and mixing times of quantum {Markov} processes},
  journal = {Journal of Mathematical Physics},
  volume = {51},
  number = {12},
  pages = {122201},
  year = {2010},
  doi = {10.1063/1.3511335},
  eprint = {1005.2358},
  archivePrefix = {arXiv}
}

@article{KastoryanoTemme2013,
  author = {Kastoryano, Michael J. and
            Temme, Kristan},
  title = {Quantum logarithmic {Sobolev} inequalities and rapid mixing},
  journal = {Journal of Mathematical Physics},
  volume = {54},
  number = {5},
  pages = {052202},
  year = {2013},
  doi = {10.1063/1.4804995},
  eprint = {1207.3261},
  archivePrefix = {arXiv}
}

@article{GaoRouze2022,
  author = {Gao, Li and
            Rouz{\'e}, Cambyse},
  title = {Complete entropic inequalities for quantum {Markov} chains},
  journal = {Archive for Rational Mechanics and Analysis},
  volume = {245},
  number = {1},
  pages = {183--238},
  year = {2022},
  doi = {10.1007/s00205-022-01785-1},
  eprint = {2102.04146},
  archivePrefix = {arXiv}
}

@article{KastoryanoBrandao2016,
  author = {Kastoryano, Michael J. and
            Brand{\~a}o, Fernando G. S. L.},
  title = {Quantum {Gibbs} samplers: the commuting case},
  journal = {Communications in Mathematical Physics},
  volume = {344},
  number = {3},
  pages = {915--957},
  year = {2016},
  doi = {10.1007/s00220-016-2641-8},
  eprint = {1409.3435},
  archivePrefix = {arXiv}
}

@article{Bardet2023,
  author = {Bardet, Ivan and
            Capel, {\'A}ngela and
            Gao, Li and
            Lucia, Angelo and
            P{\'e}rez-Garc{\'\i}a, David and
            Rouz{\'e}, Cambyse},
  title = {Rapid thermalization of spin chain commuting {Hamiltonians}},
  journal = {Physical Review Letters},
  volume = {130},
  number = {6},
  pages = {060401},
  year = {2023},
  doi = {10.1103/PhysRevLett.130.060401},
  eprint = {2112.00593},
  archivePrefix = {arXiv}
}

@article{RouzeStilckFrancaAlhambra2026,
  author = {Rouz{\'e}, Cambyse and
            Stilck Fran{\c{c}}a, Daniel and
            Alhambra, {\'A}lvaro M.},
  title = {Efficient thermalization and universal quantum computing with quantum {Gibbs} samplers},
  journal = {Nature Physics},
  volume = {22},
  number = {7},
  pages = {1123--1127},
  year = {2026},
  doi = {10.1038/s41567-026-03246-y},
  eprint = {2403.12691},
  archivePrefix = {arXiv}
}

@article{TongZhan2025,
  author = {Tong, Yu and
            Zhan, Yongtao},
  title = {Fast mixing of weakly interacting fermionic systems at any temperature},
  journal = {PRX Quantum},
  volume = {6},
  number = {3},
  pages = {030301},
  year = {2025},
  doi = {10.1103/h1dx-ps5p},
  eprint = {2501.00443},
  archivePrefix = {arXiv}
}

@misc{ChenBrandao2021,
  author = {Chen, Chi-Fang and
            Brand{\~a}o, Fernando G. S. L.},
  title = {Fast thermalization from the eigenstate thermalization hypothesis},
  year = {2021},
  eprint = {2112.07646},
  archivePrefix = {arXiv},
  primaryClass = {quant-ph}
}

@article{Brunner2025,
  author = {Brunner, Eric and
            Coopmans, Luuk and
            Matos, Gabriel and
            Rosenkranz, Matthias and
            Sauvage, Frederic and
            Kikuchi, Yuta},
  title = {{Lindblad} engineering for quantum {Gibbs} state preparation under the eigenstate thermalization hypothesis},
  journal = {Quantum},
  volume = {9},
  pages = {1843},
  year = {2025},
  doi = {10.22331/q-2025-08-29-1843},
  eprint = {2412.17706},
  archivePrefix = {arXiv}
}

@misc{ShtankoMovassagh2021,
  author = {Shtanko, Oles and
            Movassagh, Ramis},
  title = {Preparing thermal states on noiseless and noisy programmable quantum processors},
  year = {2021},
  eprint = {2112.14688},
  archivePrefix = {arXiv},
  primaryClass = {quant-ph}
}

@misc{HaganWiebe2025,
  author = {Hagan, Matthew and
            Wiebe, Nathan},
  title = {The thermodynamic cost of ignorance: thermal state preparation with one ancilla qubit},
  year = {2025},
  eprint = {2502.03410},
  archivePrefix = {arXiv},
  primaryClass = {quant-ph}
}

@misc{Langbehn2025,
  author = {Langbehn, Josias and
            Mouloudakis, George and
            King, Emma and
            Menu, Rapha{\"e}l and
            Gornyi, Igor and
            Morigi, Giovanna and
            Gefen, Yuval and
            Koch, Christiane P.},
  title = {Universal cooling of quantum systems via randomized measurements},
  year = {2025},
  eprint = {2506.11964},
  archivePrefix = {arXiv},
  primaryClass = {quant-ph}
}

@article{Molpeceres2025,
  author = {Molpeceres, Daniel and
            Lu, Sirui and
            Cirac, J. Ignacio and
            Kraus, Barbara},
  title = {Quantum algorithms for cooling: a simple case study},
  journal = {Physical Review Research},
  volume = {7},
  number = {3},
  pages = {033162},
  year = {2025},
  doi = {10.1103/4hx7-xnhw},
  eprint = {2503.24330},
  archivePrefix = {arXiv}
}

@article{ChenLiLuYing2025,
  author = {Chen, Hongrui and
            Li, Bowen and
            Lu, Jianfeng and
            Ying, Lexing},
  title = {A randomized method for simulating {Lindblad} equations and thermal state preparation},
  journal = {Quantum},
  volume = {9},
  pages = {1917},
  year = {2025},
  doi = {10.22331/q-2025-11-20-1917},
  eprint = {2407.06594},
  archivePrefix = {arXiv}
}

@article{DalibardCastinMolmer1992,
  author = {Dalibard, Jean and
            Castin, Yvan and
            M{\o}lmer, Klaus},
  title = {Wave-function approach to dissipative processes in quantum optics},
  journal = {Physical Review Letters},
  volume = {68},
  number = {5},
  pages = {580--583},
  year = {1992},
  doi = {10.1103/PhysRevLett.68.580}
}

@article{DumZollerRitsch1992,
  author = {Dum, R. and
            Zoller, P. and
            Ritsch, H.},
  title = {{Monte Carlo} simulation of the atomic master equation for spontaneous emission},
  journal = {Physical Review A},
  volume = {45},
  number = {7},
  pages = {4879--4887},
  year = {1992},
  doi = {10.1103/PhysRevA.45.4879}
}

@article{MolmerCastinDalibard1993,
  author = {M{\o}lmer, Klaus and
            Castin, Yvan and
            Dalibard, Jean},
  title = {{Monte Carlo} wave-function method in quantum optics},
  journal = {Journal of the Optical Society of America B},
  volume = {10},
  number = {3},
  pages = {524},
  year = {1993},
  doi = {10.1364/JOSAB.10.000524}
}

@book{Carmichael1993,
  author = {Carmichael, Howard},
  title = {An Open Systems Approach to Quantum Optics},
  series = {Lecture Notes in Physics Monographs},
  publisher = {Springer},
  address = {Berlin, Heidelberg},
  year = {1993},
  doi = {10.1007/978-3-540-47620-7}
}

@article{PlenioKnight1998,
  author = {Plenio, M. B. and
            Knight, P. L.},
  title = {The quantum-jump approach to dissipative dynamics in quantum optics},
  journal = {Reviews of Modern Physics},
  volume = {70},
  number = {1},
  pages = {101--144},
  year = {1998},
  doi = {10.1103/RevModPhys.70.101}
}

@inproceedings{CleveWang2017,
  author = {Cleve, Richard and
            Wang, Chunhao},
  title = {Efficient quantum algorithms for simulating {Lindblad} evolution},
  booktitle = {44th International Colloquium on Automata, Languages, and Programming (ICALP 2017)},
  series = {Leibniz International Proceedings in Informatics (LIPIcs)},
  volume = {80},
  pages = {17:1--17:14},
  year = {2017},
  doi = {10.4230/LIPIcs.ICALP.2017.17},
  eprint = {1612.09512},
  archivePrefix = {arXiv}
}

@article{JordanWigner1928,
  author = {Jordan, P. and
            Wigner, E.},
  title = {{\"U}ber das {Paulische} {\"A}quivalenzverbot},
  journal = {Zeitschrift f{\"u}r Physik},
  volume = {47},
  number = {9--10},
  pages = {631--651},
  year = {1928},
  doi = {10.1007/BF01331938}
}

@article{OsborneNielsen2002,
  author = {Osborne, Tobias J. and
            Nielsen, Michael A.},
  title = {Entanglement in a simple quantum phase transition},
  journal = {Physical Review A},
  volume = {66},
  number = {3},
  pages = {032110},
  year = {2002},
  doi = {10.1103/PhysRevA.66.032110},
  eprint = {quant-ph/0202162},
  archivePrefix = {arXiv}
}

\end{document}